\documentclass[
  10pt,
  a4paper,
  parindent 0cm,
  parsep 1cm
  ]{scrartcl}

\usepackage[a4paper]{geometry}

\usepackage{braket}
\usepackage{amsfonts} 
\usepackage{amssymb} 
\usepackage{bbm}
\usepackage{graphicx}  
\usepackage{psfrag}    
\usepackage{amsmath}
\usepackage{amsthm}
\usepackage{enumerate}
\usepackage{color}

\usepackage{array}
\usepackage{booktabs}
\usepackage{longtable}
\usepackage{lmodern}
\usepackage[T1]{fontenc}
\usepackage[latin1]{inputenc}
\usepackage{MnSymbol}
\usepackage{mathtools}

\allowdisplaybreaks[1]

\ifx\widecheck\undefined
  \DeclareFontFamily{U}{mathx}{\hyphenchar\font45}
  \DeclareFontShape{U}{mathx}{m}{n}{
    <5> <6> <7> <8> <9> <10>
    <10.95> <12> <14.4> <17.28> <20.74> <24.88>
    mathx10
  }{}
  \DeclareSymbolFont{mathx}{U}{mathx}{m}{n}
  \DeclareFontSubstitution{U}{mathx}{m}{n}
  \DeclareMathAccent{\widecheck}     {0}{mathx}{"71}
\fi

\def\idty{{\leavevmode\rm 1\mkern -5.4mu I}} 
\def\id{{\rm id}}                            
\def\Rl{{\mathbb R}}\def\Cx{{\mathbb C}}     
\def\Nl{{\mathbb N}}     

\def\ketbra #1#2{\vert #1\rangle \langle #2\vert}
\def\kettbra#1{\ketbra{#1}{#1}}
\def\tr{\mathop{\rm Tr}\nolimits}

\mathchardef\ree="023C \mathchardef\imm="023D  

\def\BB{{\mathcal B}}
\def\HH{{\mathcal H}}

\def\Forall{\qquad \text{for all }}

\newcommand{\beq}{\begin{equation}}
\newcommand{\eeq}{\end{equation}}
\newcommand{\rem}[1]{}  

\newtheorem{thm}{Theorem}[section]
\newtheorem{lem}[thm]{Lemma}
\newtheorem{prop}[thm]{Proposition}
\newtheorem{cor}[thm]{Corollary}
\newtheorem{defn}[thm]{Definition}
\newtheorem{remark}[thm]{Remark}

\def\dom{{\rm dom}\,}

\newcommand{\sw}{\mathfrak{S}}
\newcommand{\sfunc}[1]{\sw(\Rl^{#1})}
\newcommand{\sfu}{\sw(\phaseS)}
\newcommand{\sop}{\sw(\HH)}
\newcommand{\sopL}{\sw_L(\HH)}
\newcommand{\sopR}{\sw_R(\HH)}
\newcommand{\sopnot}{\sw_0(\HH)}

\newcommand{\ssec}[1]{\mathfrak{s}_{#1}}

\newcommand{\sopd}{\sw'(\HH)}
\newcommand{\sfud}{\sw'(\phaseS)}

\newcommand{\phaseS}{X}

\newcommand{\Htot}{H_{\rm tot}}
\newcommand{\indm}[1]{I_{#1}}
\newcommand{\hsbe}[1]{E_{#1}}
\newcommand{\schtn}[1]{\mathcal{T}_{#1}(\HH)}
\newcommand{\tc}{\mathcal{T}(\HH)}
\newcommand{\hs}{\mathcal{HS}(\HH)}
\newcommand{\bh}{\BB(\HH)}
\newcommand{\ch}{\mathcal{C}(\HH)}

\newcommand{\kernel}[1]{K^{#1}}
\newcommand{\matr}[1]{{#1}}

\newcommand{\wq}[1]{{\mathfrak{W}[#1]}}
\newcommand{\wf}[1]{{\mathfrak{W^{-1}}[#1]}}

\newcommand{\otimesA}{\otimes}
\newcommand{\otimesH}{\otimes_{\HH}}
\newcommand{\otimesS}{\otimes_{s}}

\newcommand{\SP}{\operatorname{span}}

\newcommand{\myrho}{T}

\begin{document}
\title{Schwartz operators}
\author{\large M. Keyl\\[-.2em]
{\normalsize TU M\"unchen, Fakult\"at Mathematik,}\\[-.2em] {\normalsize
  Boltzmannstr. 3, 85748 Garching, Germany}\\[-.2em] {\normalsize \texttt{Michael.Keyl@tum.de}}\\[.5em]
\large J. Kiukas\\[-.2em]
{\normalsize School of Math. Sci., Univ. Nottingham,}\\[-.2em] {\normalsize
  University Park, Nottingham, NG7 2RD, UK}\\[-.2em] {\normalsize \texttt{Jukka.Kiukas@nottingham.ac.uk}}\\[.5em]
\large R. F. Werner\\[-.2em]
{\normalsize Inst. Theor. Physik, Leibniz Univ.}\\[-.2em] {\normalsize
  Hannover, Appelstr. 2, 30167 Hannover, Germany}\\[-.2em] {\normalsize \texttt{Reinhard.Werner@itp.uni-hannover.de}}}
\date{\large \today}
\maketitle

\begin{abstract}
  In this paper we introduce Schwartz operators as a non-commutative analog of
  Schwartz functions and provide a detailed discussion of their
  properties. We equip them in particular with a number of different (but
  equivalent) families of seminorms which turns the space of Schwartz
  operators into a Frechet space. The study of the topological dual leads to
  non-commutative tempered distributions which are discussed in detail as well. We show
  in particular that the latter can be identified with a certain class of
  quadratic forms, therefore making operations like products with bounded (and
  also some unbounded) operators and quantum harmonic analysis available to
  objects which are otherwise too singular for being a Hilbert space
  operator. Finally we show how the new methods can be applied by studying
  operator moment problems and convergence properties of fluctuation operators.
\end{abstract}

\section{Introduction}
The theory of tempered distributions is used extensively in various areas of mathematical physics, in order to regularise singular objects, most notably "delta-functions" that often appear as a result of some convenient idealisation (e.g. plane wave, or point interaction). The well-known intuitive idea is to make sense of the functional
\begin{equation}\label{functional}
\phi(f)=\int f(x)\phi(x)dx, \quad f\in \mathcal S(\mathbb R^{2N})
\end{equation}
to cases where $\phi$ is no longer a function, by making use of the highly regular behaviour of the Schwartz functions $f\in S(\mathbb R^{2N})$. Considering applications to quantum theory, tempered distributions have proved a powerful tool in e.g. field theoretical settings, and quantisation, where the space of Schwartz functions typically appears as a dense subspace of the relevant Hilbert space, leading to rigged Hilbert space constructions \cite{antoine,hennings}. However, the approach based on the direct quantum analogy of \eqref{functional} seems to be completely missing in literature. In order to demonstrate this analogy, take $\mathbb R^{2N}$ to be the phase space (of e.g. a classical $N$-particle system), so that the associated quantum system is given by the standard representation of the Canonical Commutation Relations on the Hilbert space $\mathcal H=L^2(\mathbb R^N)$. In view of quantum-classical correspondence theory \cite{Folland,Werner}, it is clear that the proper analogue of $\phi(f)$ should be
$$
\Phi(S)={\rm tr}[S\Phi],\quad S\in \mathcal S(\mathcal H),
$$
with a suitable class $\mathcal{S(H)}$ of "very regular" \emph{operators} taking the role of Schwartz functions. As it stands, the trace makes sense e.g. if $\Phi$ is a bounded operator and $S$ a trace class operator on $\mathcal H$ (analogous to the above classical integral making sense for integrable $f$ and bounded $\phi$.) With $S$ sufficiently regular, we can relax the requirements for $\Phi$, so that the functional $\Phi(S)$ still gives a finite value.

It is intuitively clear that the appropriate class $\mathcal{S(H)}$ should be
those $S$ for which $Q_i^kP_j^lSP_{j'}^{k'}Q_{i'}^{l'}$ remains bounded for
arbitrary powers $k,l,k',l'\in \mathbb N$, where $Q_i$ and $P_j$,
$i,j=1,\ldots,N$ are the canonical coordinate operators with
$[Q_i,P_j]=i\delta_{ij}\idty$. In the present paper, we investigate this class
of \emph{Schwartz operators} in detail, showing that it indeed becomes a
Frechet space when equipped with the seminorms $S\mapsto
\|Q_i^kP_j^lSP_{j'}^{k'}Q_{i'}^{l'}\|_\infty$; cf. Section
\ref{sec:schwartz-operators}.

As the first application of the theory developed
up to that point we will discuss in Sect. \ref{sec:oper-moment-probl} an \emph{operator moment
problem} which can be regarded as the non-commutative analog of the Hamburger
moment problem which is well known in measure theory. More precisely, the main
question to be answered here is under which conditions a positive Schwartz
operator $\myrho$ is uniquely determined by its moments $\tr(Q^\alpha P_\beta
\myrho)$. In Sect. \ref{sec:temp-distr} we interpret the topological dual
$\mathcal S'(\mathcal H)$ of $\mathcal{S}(\mathcal{H})$ as the quantum
analogue of the space of tempered distributions. We then prove that the
Weyl transform maps $\mathcal S'(\mathcal H)$ bijectively onto $\mathcal S'(\mathbb R^{2N})$,
making the standard theory of distributions available for this quantum
setting. We develop basic harmonic analysis concepts (e.g. Fourier-Weyl
transform and convolutions) in this setting, also providing a natural formulation of the well-known Wigner-Weyl quantisation, which is known to exist as a map from $\mathcal S'(\mathbb R^{2N})$ to the space of certain quadratic forms (see e.g. \cite{Daub80}). Finally, in
Sect. \ref{sec:spectr-dens-distr} we consider physically motivating examples from
mean-field theory where this formalism has already proved to be useful
\cite{Fluct}.

\section{Preliminaries and notations}

\subsection{CCR and phase space correspondence theory}

We begin by recalling basic facts about the phase space formulation of quantum-classical correspondence. We fix the number of degrees of freedom to be $N$, and consider the \emph{phase space} $\phaseS:=\mathbb R^{2N}$ of position-momentum pairs $(q,p)$. It is equipped with the symplectic form $$\{(q,p),(q',p')\}:= q'\cdot p-q\cdot p',$$ and acts irreducibly on the associated quantum system via the standard representation of the CCR relations
$$
W(x)W(y) = e^{i\{x,y\}/2}W(x+y)
$$
on the (position) Hilbert space $\HH=L^2(\Rl^N,dq)$ given by the \emph{Weyl operators} $$W(x) := e^{\frac i2q\cdot p} e^{-iq\cdot P}e^{ip\cdot Q}, \qquad x=(q,p)\in \phaseS,$$
defined via the usual position and momentum operators $Q_i$ and $P_i$, $i=1,\ldots,N$. Recall that $Q_i:\dom (Q_i)\to L^2(\Rl^N)$ is the operator of multiplication by the coordinate $q_i$, and $P_i=-id/(dq_i):\dom(P_i)\to L^2(\Rl^N)$. The domain of $Q_j$ is given by $\dom(Q_i) := \{ \psi\in L^2(\Rl^N)\mid \int |q_i|^2|\psi(q)|^2 d^Nq<\infty\}$, and the domain $\dom(P_i)$ can be defined by its image under the unitary extension of the Fourier transform $F:\sfunc N \to \sfunc N$,
\begin{align*}
(F\psi)(p) = \frac{1}{\sqrt{(2\pi)^N}} \int_{\Rl^N} e^{-ip\cdot q} \psi(q)\, d^Nq.
\end{align*}
In particular, we have $P_i = F^*Q_i F$ holding for each $i$. We denote $Q:=(Q_1,\ldots, Q_N)$ and $P:=(P_1,\ldots,P_N)$, and e.g. $q\cdot P:=q_1P_1+\cdots + q_NP_N$.
Since $(e^{iq\cdot P}\psi)(x) = \psi(x+q)$ for $\psi\in \HH_N$, the Weyl operators act as
\begin{align}\label{Weylrule}
(W(q,p)\psi)(q') &= e^{-\frac i2q\cdot p} e^{ip\cdot q'}\psi(q'-q), & \psi&\in \HH.
\end{align}
We let $\bh$, $\hs$, and $\tc$ denote the set of bounded, Hilbert-Schmidt, and trace class operators on $\HH$, respectively.

In the phase space $\phaseS$ we use the measure $dx:=(2\pi)^{-N}dq dp$. With this choice of normalization,\footnote{Note that the choice of measure means, in particular, that $L^2(\phaseS)=L^2(\phaseS, dx)$. This distinguishes that space from $L^2(\Rl^{2N})=L^2(\Rl^{2N}, dqdp)$ which we are also using.} the \emph{symplectic Fourier transform} of an $f\in L^2(\phaseS)\cap L^1(\phaseS)$, defined as
\begin{align}\label{symplecticF}
&\widehat{f}:\phaseS\to \Cx,& \widehat{f}(x) &= \int e^{-i\{x,y\}} f(y) dy,
\end{align}
extends to a unitary operator on $L^2(\phaseS)$. Similarly, the \emph{Weyl transform} of an operator $T\in \tc$, defined via
\begin{align}\label{Weyltransform}
&\widehat{T}:\phaseS\to \Cx, & \widehat{T}(x) &= \tr [W(x)T],
\end{align}
extends to a unitary operator $\hs\to L^2(\phaseS)$. The symplectic Fourier
transform is its own inverse, and we reserve the symbol $\widecheck {f}$ for the \emph{inverse Weyl transform} of an $f\in L^2(\phaseS)$. This is explicitly given by
\begin{equation}\label{inverseWeyl}
\widecheck{f} = \int W(-y) f(y)dy,
\end{equation}
where the integral exists (e.g.) in the weak-* topology of $\bh$.

\emph{Convolutions} between trace class operators $T$ and functions $f\in L^1(\phaseS)$ are defined \cite{Werner} as follows:
\begin{align*}
f*T :=T*f &:= \int f(x) W(x) T W(x)^*dx, & (S*T)(x) &:= {\rm tr}[SW(x)T_-W(x)^*],
\end{align*}
where $T_-:=\Pi T\Pi$, and $\Pi$ is the parity operator. In general, for any integrable functions and trace class operators, we have $f*T\in \tc$, and $S*T\in L^1(\phaseS)$. Moreover, the convolutions are commutative and associative, and satisfy
\begin{align}\label{fourierconv}
\widehat{f*T} &= \widehat{f}\widehat{T}, & \widehat{S*T} &= \widehat{S}\widehat{T}.
\end{align}

Such maps provide correspondence \cite{Werner} of classical variables (functions on $\phaseS$) and quantum observables (operators on $\HH$).

\subsection{Schwartz functions}

We will start by reiterating the basic facts about Schwartz functions (see e.g. \cite{RSI}), in a way that emphasizes the parallels to Schwartz operators, which we define in the following section.

Given $n\in \mathbb N$, the Schwartz space is the set $\sfunc n$ of infinitely differentiable functions $\varphi:\Rl^n\to \Cx$, for which
\begin{equation}\label{basicseminorm}
\sup_{q\in \Rl^n}|q^{\alpha}(D^{\beta}\varphi)(x)| <\infty, \qquad \text{for all } \alpha,\beta\in \indm n,
\end{equation}
where $\indm n:=\{ \alpha=(\alpha_1,\ldots,\alpha_n)\mid \alpha_i\in \Nl\cup \{0\} \text{ for all } i=1,\ldots,n\}$
is the set of multi-indices, and
\begin{align}\label{xd}
q^\alpha &:=x_1^{\alpha_1}\cdots x_n^{\alpha_n}, &
D^\beta &:=\frac{\partial^{|\beta|}}{\partial q_1^{\beta_1}\cdots \partial q_n^{\beta_n}},
\end{align}
and $|\beta|:=\sum_{i=1}^n \beta_i$. The expressions in \eqref{basicseminorm} define a locally convex metrizable topology on $\sfunc n$, and it is a standard exercise to prove its completeness. There are many other natural choices for a family of seminorms inducing this topology.

In order to avoid confusion, we stress here that Schwartz functions appear in conceptually distinct roles in this paper: as elements of the Hilbert space $\HH$ (i.e. functions on the configuration space $\Rl^N$), as functions on the cartesian product $\Rl^{N}\times \Rl^N$ of two copies of the configuration space, and finally as functions on the phase space $\phaseS=\mathbb R^{2N}$.

\subsubsection{Schwartz functions on the configuration space}

We now consider the class $\sfunc N$, as a subspace of the (configuration) Hilbert space $L^2(\Rl^N, dq)$ where the Weyl representation acts. 
Since $\prod_{i=1}^N(1+q_i^2)^{-1}\in L^2(\Rl^N)\cap L^\infty(\Rl^N)$, it follows easily that $\sfunc N \subset L^2(\Rl^N)\cap L^\infty(\Rl^N)$, and the topology is induced by the natural seminorms
\begin{equation}\label{hilbertschwartznorms1}
\|\varphi\|_{\alpha, \beta} := \left(\int |q^\alpha (D^\beta \varphi)(q)|^2\, dq\right)^{\frac 12}=\|Q^\alpha P^\beta \varphi\|,
\end{equation}
where
\begin{align}\label{QP}
Q^\alpha &=Q_1^{\alpha_1}\cdots Q_N^{\alpha_N}, &
P^\beta &=P_1^{\beta_1}\cdots P_N^{\beta_N}.
\end{align}
This allows us to state the following characterisation of Schwartz functions, which does not a priori assume differentiability:
\begin{lem}\label{schwartzvector} Let $\varphi\in L^2(\Rl^N)$. Then $\varphi\in \sfunc N$ if and only if for each $\alpha,\beta\in \indm N$, the $\sfunc N$-continuous linear functional
$$\sfunc N \ni \psi\mapsto \langle \varphi |P^\beta Q^\alpha \psi \rangle\in \Cx$$
is Hilbert space bounded, that is,
\begin{equation}\label{sfhilbert}
\|\varphi\|_{\alpha, \beta} := \sup\{ |\langle P^\beta Q^\alpha \psi |\varphi\rangle| \mid \psi\in \sfunc N,\, \|\psi\|\leq 1\}<\infty.
\end{equation}
For $\varphi\in \sfunc N$ this expression in fact equals \eqref{hilbertschwartznorms1}.
\end{lem}
\begin{proof} If $\varphi\in L^2(\Rl^N)$, and $\|\varphi\|_{0,\beta}<\infty$ for all $\beta$, then it follows from the selfadjointness of $P^\beta$, and the fact that $\sfunc N$ is a core for $P^\beta$, that $\varphi$ is in the domain of each $P^\beta$, which in particular implies differentiability to all orders. If $\|\varphi\|_{\alpha,\beta}<\infty$ for all $\alpha,\beta$, then the same argument applied to $Q^\alpha$ establishes that $P^\beta\varphi$ is in the domain of $Q^\alpha$, which just means that $\varphi$ is in the domain of $Q^\alpha P^\beta$, and we have $Q^\alpha P^\beta \varphi = i^{|\beta|} q^\alpha D^\beta\varphi$. Hence, the definition \eqref{sfhilbert} is the same as \eqref{hilbertschwartznorms1}, and we have $\varphi\in \sfunc N$. On the other hand, if $\varphi\in \sfunc N$, then $\varphi$ is in the domain of each $Q^\alpha P^\beta$, which implies that $\|\varphi\|_{\alpha,\beta}<\infty$ in \eqref{sfhilbert}.
\end{proof}

\begin{remark}\rm
Of course, the condition \eqref{sfhilbert} just means that $\|\varphi\|_{\alpha, \beta} = \|Q^\alpha P^\beta \varphi\|<\infty$, with the understanding that this means $\varphi$ belonging to the domain of $Q^\beta P^\alpha$. It is however useful to express this explicitly in terms of the subspace $\sfunc N$ itself; in fact, this becomes necessary when we define Schwartz operators in the next section.
\end{remark}

In order to introduce other useful families of seminorms for $\sfunc N$, we define the selfadjoint operators
\begin{align*}
H_i&= \frac 12 (Q_i^2+P_i^2), & \Htot &= \sum_{i=1}^N H_i.
\end{align*}
These have a (common) complete set of eigenfunctions $|\alpha\rangle\in L^2(\Rl^N)$, $\alpha\in \indm N$, given by $|\alpha\rangle(q)= h_{\alpha_1}(q_1)\cdots h_{\alpha_N}(q_N)$, where $h_n$, $n=0,1,\ldots$ are the Hermite functions. The eigenvalues are $H_i|\alpha\rangle = (\alpha_i+\frac 12)|\alpha\rangle$, and $H|\alpha\rangle = \sum_{i=1}^N (\alpha_i+\frac 12) |\alpha\rangle$. We call $\{ |\alpha\rangle\mid \alpha\in \indm N\}$ the \emph{number basis} of $L^2(\Rl^N)$.

For any $N$-tuple $A:=(A_1,\ldots,A_N)$ of operators $A_i$, each acting only on the tensor factor corresponding to the $i$:th coordinate, we let $A^\alpha := A_1^{\alpha_1}\cdots A_N^{\alpha_N}$ for $\alpha\in \indm N$. In particular, $H=(H_1,\ldots,H_N)$, with $H^\alpha= H_1^{\alpha_1}\cdots H_N^{\alpha_N}$, while $\Htot^m$, $m\in \Nl$ is the usual power of the single operator $\Htot$.

The following families of seminorms all induce the topology of $\sfunc N$.
\begin{enumerate}
\item 
$\|\varphi\|_f :=\| f(Q,P)\varphi\|$, where $f:\Rl^N\times \Rl^N\to \Cx$ goes through all polynomials. \item Letting $A_i = \frac{1}{\sqrt 2}(Q_i+iP_i)$ denote the annihilation operators, we get the seminorms $\varphi\mapsto \|f(A,A^*)\varphi\|$, with $f$ as above.
\item $\varphi\mapsto \|H^{\alpha} \varphi\|$, where $\alpha\in \indm N$. Indeed, if $f(A,A^*)$ is any monomial of $A_i,A_i^*$, $i=1,\ldots,N$, with $m_i$ factors equal to either $A_i$ or $A_i^*$ for each $i$, then $\|f(A,A^*)\varphi\|\leq \left\|\prod_{i=1}^N (H_i-\frac 12 +m_i)^{[m_i/2]}\varphi\right\|$, where $[m_i/2]$ is an integer $\geq m/2$.
\item $\varphi\mapsto \|\Htot^n\varphi\|$, where $n=0,1,2,\ldots$. This induces the topology,
because $\|H^{\alpha}\varphi\|\leq \|\Htot^{|\alpha|} \varphi\|$ for every multiindex $\alpha$.
\end{enumerate}
Finally, we introduce an equivalent family of seminorms based on the expansion of $\varphi\in L^2(\Rl^N)$ in the orthonormal basis $\{|\alpha\rangle\}$ of $L^2(\Rl^N)$. In fact, for each $\alpha\in \indm N$ we define
\begin{equation}\label{Nrep}
\|(a_\beta)\|_{\alpha}:= \sqrt{\sum_{\beta\in \indm N} (\beta+1)^{2\alpha} |\langle \beta|\varphi\rangle|^2}, \qquad (a_\beta)\in \ell^2(\indm N),
\end{equation}
where $(\beta+1)^{2\alpha}:=\prod_{i=1}^N (\beta_i+1)^{2\alpha_i}$. We then define the space of \emph{Schwartz (multi)sequences}
\begin{equation}\label{seqSchwartz}
\ssec{N}:= \{(a_\beta) \in \ell^2(\indm N)\mid  \|(a_\beta)\|_{\alpha}<\infty\},
\end{equation}
and equip this with the topology given by the seminorms
$\|\cdot\|_\alpha$. The map
\begin{displaymath}
\varphi\mapsto (a^\varphi_\alpha),\quad \text{where}\quad a_\alpha^\varphi = \langle \alpha |\varphi\rangle,
\end{displaymath}
can easily shown to be a continuous bijection between $\sfunc N$ and $\ssec{N}$.
This representation of the Schwartz functions is called \emph{$N$-representation}. The following result will be useful:
\begin{prop}\label{sfdensity}
\begin{itemize}
\item[(a)] Let $\varphi\in L^2(\Rl^N)$. Then $\varphi\in \sfunc N$ if and only if the basis expansion $\varphi = \sum_\alpha  \langle \alpha |\varphi\rangle |\alpha\rangle$ converges in the topology of $\sfunc N$.
\item[(b)] ${\rm span} \{ |\alpha\rangle \mid \alpha\in \indm N\}$ is dense in $\sfunc N$ in the topology of $\sfunc N$.
\end{itemize}
\end{prop}
\begin{proof} If $\varphi\in L^2(\Rl^N)$, and $\sum_\alpha  \langle \alpha |\varphi\rangle |\alpha\rangle$ converges in $\sfunc N$, then it converges also in $L^2(\Rl^N)$ because $\|\cdot \|_{0}=\|\cdot\|$. Let $\varphi\in \sfunc N$. For each finite set $F\subset \indm N$ define the truncation
$$\varphi_F := \sum_{\beta\in F} \langle \beta |\varphi\rangle |\beta\rangle,$$ which is in the linear span of the number basis. We need to show that the net $\{\varphi_F\mid F\subset \indm N, \, F \text{ finite }\}$ converges to $\varphi$ in the topology of $\sfunc N$. To that end, fix $\alpha\in \indm N$. Since the series defining $\|(a^\varphi_\beta)\|_{\alpha}$ consists of positive terms, we have
$$\|(a^\varphi_\beta)\|_{\alpha}^2 = \sup_{F\subset \indm N, F\, {\rm finite}} \,\sum_{\beta\in F}  (\beta+1)^{2\alpha} |a^\varphi_\beta|^2<\infty.$$
Hence, given $\epsilon>0$ there exists an finite set $F_0\subset \indm N$, such that
$$\sum_{\beta\in \indm N\setminus F_0}  (\beta+1)^{2\alpha} |a^\varphi_\beta|^2<\epsilon.$$
Now if $F\supset F_0$, we have $\|a^{\varphi_F}_\beta-a^\varphi_\beta\|_{\alpha}^2<\epsilon$, because $\langle \beta |\varphi_F-\varphi \rangle=0$ for $\beta \in F$, and $\langle \beta |\varphi_F-\varphi \rangle=\langle\beta|\varphi\rangle$ for $\beta \in \indm N\setminus F$. Hence, the expansion converges. We have now proved (a) and (b).
\end{proof}

\subsubsection{Tensor products of Schwartz functions}
The above discussion of Schwartz functions was done with the identification $\sfunc N\subset \HH=L^2(\Rl^N, dq)$. Now kernel operators on $\HH$ are specified by functions on $\Rl^{2N}$ (understood as cartesian product of two copies of the configuration space), and as it will turn out later, kernels of Schwartz operators are Schwartz functions. For this reason we now briefly review tensor products of Schwartz spaces.

In order to conveniently denote the multi-indices, we map $\indm N\times \indm N$ bijectively onto $\indm{2N}$ via
\begin{equation}\label{alphavee}
(\alpha,\beta)\mapsto \alpha\vee\beta:=(\alpha_1,\ldots,\alpha_N,\beta_1,\ldots,\beta_N).
\end{equation}
This notation will be used frequently in the rest of the paper. The coordinate products $q^{\alpha\vee\beta}$ and derivatives $D^{\alpha\vee\beta}$ will be understood accordingly. Since the \emph{Hilbert space tensor product} $\HH\otimesH \HH$ is just $L^2(\Rl^{2N}, dqdq')$ via the usual identification, we can also use the corresponding notations $Q^{\alpha\vee\beta}$ and $P^{\alpha\vee\beta}$ for operators.

Concerning now the tensor products of Schwartz spaces, we have the following result.
\begin{prop}
\begin{itemize}
\item[(a)] We inject
\begin{align*}
\sfunc N \otimesA \sfunc N &\subset \sfunc {2N}, & \varphi\otimesA \psi\mapsto ((q,q')\mapsto \varphi(q)\psi(q')),
\end{align*}
where algebraic tensor product is meant. The set $\sfunc N \otimesA \sfunc N$ is dense in the topology of $\sfunc {2N}$.
\item[(b)] If $A,B:\sfunc N\to \sfunc N$ are continuous and linear, there exists a unique continuous linear map $A\otimesS B:\sfunc {2N}\to \sfunc {2N}$, such that
$$
(A\otimesS B)(\varphi\otimes \psi) = A\varphi\otimesA B\psi\Forall \varphi,\psi\in \sfunc N.
$$
\item[(c)] If $A,B$ are bounded operators on $L^2(\Rl^N)$, which keep $\sfunc N$ invariant, and their restrictions to $\sfunc N$ are continuous in the topology of $\sfunc N$, we have
$$(A\otimesS B)(\psi)  = (A\otimesH B)(\psi)\Forall \psi\in \sfunc{2N}.$$
\end{itemize}
\end{prop}
\begin{proof}
Concerning (a), we have, in particular, $|\alpha\rangle\otimes |\beta\rangle= |\alpha\vee\beta\rangle$ (see \eqref{alphavee}). Since the basis expansion of an arbitrary $\varphi\in \sfunc {2N}$ converges in $\sfunc{2N}$ by Proposition \ref{sfdensity} (a), it follows that $\varphi$ is in the $\sfunc{2N}$-closure of $\sfunc N\otimesA \sfunc N$. This proves (a). Part (b) is now clear because $\sfunc {2N}$ is complete. Part (c) follows from (b) and Prop. \ref{sfdensity} (a), because the number basis expansion for any $\psi\in \sfunc{2N}$ converges in both topologies, and any sequence converging in $\sfunc {2N}$ also converges in $L^2(\Rl^{2N})$.
\end{proof}

\subsubsection{Schwartz functions on the phase space}

Recall that the phase space is $\phaseS=\Rl^N\times\Rl^N$, where now the first factor is the configuration space, and the second is the momentum space; this "physical" instance of $\Rl^{2N}$ should be kept conceptually separate from the "double configuration space" introduced above for technical reasons. In particular, for $x\in \phaseS$ we have $$x^{\alpha\vee\beta}=q^\alpha p^\beta=q_1^{\alpha_1}\cdots q_N^{\alpha_N}p_1^{\beta_1}\cdots p_N^{\beta_N},$$ and similarly for the derivatives, defining the Schwartz class $\sfu$. From the point of view of the present paper, this constitutes the classical analogy of the class of Schwartz operators to be defined in Sect. \ref{sec:schwartz-operators}. We will then make this analogy more concrete in terms of Fourier-Weyl correspondence and Wigner quantisation.

\subsection{Hilbert-Schmidt operators, their kernels, and the unitary Weyl transform}

As we have seen, the Hilbert space seminorms make the Schwartz functions easier to work with. Anticipating the introduction of Schwartz operators, it is not difficult to guess that an essential role is played by the class of Hilbert-Schmidt operators. Indeed, this is a Hilbert space with the scalar product $\langle T| S\rangle_{\hs} := \tr [T^*S]$, and the Weyl transform is well-known to be a one-to-one map between $\hs$ and $L^2(\phaseS)$.

Before looking at this class, we introduce the more general \emph{Schatten classes}. We first fix some general notations. Let $\HH$ be an arbitrary (complex separable) Hilbert space, let $\bh$ denote the set of bounded operators on $\HH$, and $\|\cdot \|$ the operator norm on $\bh$. The set of compact operators is denoted by $\ch$, and for $p\in 1,2,\ldots$ the corresponding Schatten class is denoted by $\schtn p$. It is a Banach space of those $T\in \bh$ with finite $p$-norm $\|T\|_p := (\tr |T|^p)^{1/p}$, where $|T| := \sqrt{T^*T}$. In particular, $\hs:=\schtn 2$ and $\tc:=\schtn 1$ are the Hilbert-Schmidt and trace classes, respectively. We have $\|T\| \leq \|T\|_p \leq \|T\|_{p-1}\leq \cdots \|T\|_2\leq \|T\|_1$, and the inclusions $\tc \subset \hs \subset \schtn 3 \cdots \schtn p \subset \ch\subset \bh$.

Each compact operator $T\in \ch$ has the (operator norm convergent) \emph{singular value decomposition}
\begin{equation}\label{singularvalue}
T=\sum_{k=1}^\infty c_k |\varphi_k\rangle\langle \psi_k|, \quad c_n\geq 0, \, \lim_n c_n = 0,\quad (\varphi_k) \text{ and }(\psi_k) \text{ orthonormal bases},
\end{equation}
where the $c_k\neq 0$ are \emph{singular values} of $T$, i.e. the eigenvalues of the positive compact operator $|T|=\sqrt{T^*T}$. In particular, $\|T\|_p= \left(\sum_k c_k^p\right)^{1/p}$, so $T\in \schtn p$ if and only if $\sum_k c_k^p<\infty$, and in that case, the series in \eqref{singularvalue} converges in $\|\cdot\|_p$. Clearly, a compact operator $T$ is of \emph{finite rank}, i.e. has finite-dimensional range, if and only if $\{ k\mid c_k\neq 0\}$ is a finite set. Consequently, the set of finite rank operators is dense in each $\schtn p$. We actually need a slightly stronger result:

\begin{lem}\label{schwartzdensity} The set of operators of the form
\begin{equation}\label{finiteschwartz}
T=\sum_{k=1}^m |\varphi_k\rangle \langle \psi_k|, \qquad m\in \Nl, \, \, \psi_k,\varphi_k\in \sfunc N \,\, \text{for all } k=1,\ldots,m
\end{equation}
is $\|\cdot\|_p$-dense in each $\schtn p$.
\end{lem}
\begin{proof}
Let $T\in \schtn p$ and $\epsilon>0$. Since \eqref{singularvalue} converges in the $\|\cdot \|_p$-norm, we can find a finite rank operator $T'$ with $\|T-T'\|_p<\epsilon$. On the other hand, $\sfunc N$ is dense in $\HH$, and for arbitrary $\varphi,\psi,\varphi',\psi'\in \HH$ we have
$$\||\varphi\rangle \langle \psi |-|\varphi'\rangle \langle \psi' |\|_1\leq \||\varphi-\varphi'\rangle\langle \psi|\|_1+\||\varphi'\rangle\langle \psi-\psi'|\|_1\leq \|\varphi-\varphi'\|\|\psi\|+\|\psi-\psi'\|\|\varphi'\|.$$
Thus we can find an operator $T''$ of the form \eqref{finiteschwartz} with $\|T'-T''\|_p\leq \|T'-T''\|_1<\epsilon$. This completes the proof.
\end{proof}

We now review the special properties of the Hilbert-Schmidt operators on $L^2(\Rl^N)$, which are conveniently characterised by their kernel functions, as elements of $L^2(\Rl^{2N})$. For this purpose we recall that the basic unitary equivalence
\begin{align}\label{tensorhilbert}
L^2(\Rl^N)\otimesH L^2(\Rl^N)&\simeq L^2(\Rl^{2N}), & \varphi\otimesH \psi&\mapsto ((q,q')\mapsto \varphi(q)\psi(q')).
\end{align}
is conveniently characterized via the identification $$|\alpha\rangle \otimesH |\alpha'\rangle =|\alpha\vee\alpha'\rangle, \Forall \alpha,\alpha'\in \indm N$$
of the associated number bases. Moreover, we can also construct an orthonormal basis $\{\hsbe\alpha\mid \alpha\in \indm{2N}\}$ of the Hilbert space $\mathcal{HS}(L^2(\Rl^N))$ by putting $$\hsbe{\alpha\vee\alpha'} :=|\alpha\rangle\langle \alpha'|, \Forall \alpha,\alpha'\in \indm N.$$
The suggestive identification
$$
\hsbe{\alpha}\mapsto |\alpha\rangle, \quad \alpha\in \indm{2N}
$$
plays an important technical role in this paper, the starting point given by the following well-known Lemma, which summarises the properties of the Hilbert-Schmidt operators and their kernels and matrix representations.

\begin{lem}[Hilbert-Schmidt operators]\label{hslemma}
\begin{itemize}
\item[(a)] Let $T$ be a bounded operator on $L^2(\Rl^N)$. Then the following conditions are equivalent:
\begin{itemize}
\item[(i)] $T$ is Hilbert-Schmidt.
\item[(ii)] There exists a unique $\kernel{T}\in L^2(\Rl^{2N})$, called the \emph{kernel} of $T$, such that
$$
\langle \psi |T\varphi\rangle = \int_{\Rl^{2N}} \overline{\psi(q)}\kernel{T}(q,q')\varphi(q') \, dq' dq' =\langle \psi\otimes C\varphi| \kernel{T}\rangle \Forall \psi,\varphi\in \HH,
$$
where $C:L^2(\Rl^N)\to L^2(\Rl^N)$ is the complex conjugation map.
\item[(iii)] We have $(\matr{T}_{\alpha})\in \ell^2(\indm{2N})$, where
$$
\matr{T}_{\alpha\vee\alpha'}:=\langle \alpha|T|\alpha'\rangle \Forall \alpha,\alpha'\in \indm N.
$$
$(\matr{T}_{\alpha})_{\alpha\in \indm{2N}}$ is called the \emph{matrix} of $T$.
\end{itemize}
\item[(b)]
For a $T\in \hs$, we have 
$
\langle \alpha|T|\alpha'\rangle = \langle \hsbe{\alpha\vee\alpha'} |T\rangle_{\hs} = \matr{T}_{\alpha\vee\alpha'}=\langle \alpha\vee\alpha' |\kernel{T}\rangle.
$
The maps
\begin{align*}
\hs \ni T&\mapsto \kernel{T}\in L^2(\Rl^{2N}), & \hs\ni T&\mapsto (\matr{T}_\alpha)\in \ell^2(\indm{2N})
\end{align*}
are unitary; in particular, $\hs \simeq L^2(\Rl^{2N}) \simeq \ell^2(\indm{2N})$.
\end{itemize}
\end{lem}
\begin{proof} Part (a): If (i) holds, then $T$ has the singular value decomposition \eqref{singularvalue} with $\sum_k c_k^2 = \|T\|_2^2<\infty$. Now $\{ (q,q')\mapsto \varphi_k(q)\overline{\psi_k(q')}\}$ is an orthonormal set in $L^2(\Rl^N\times \Rl^N,dqdq')$, so $\kernel{T}(q,q') := \sum_k c_k \varphi_k(q)\overline{\psi_k(q')}\in L^2(\Rl^{2N})$ because $\sum_k c_k^2<\infty$, the series converging in the norm. The scalar product of $\kernel{T}$ against a vector $\left((q,q')\mapsto \psi(q)\overline{\varphi(q')}\right)\in L^2(\Rl^N\times \Rl^N)$ is then
$$
\langle \psi |T\varphi \rangle = \langle \psi\otimes C\varphi| \kernel{T}\rangle= \sum_k c_k \langle \psi |\varphi_k\rangle \langle \psi_k|\varphi\rangle = \int_{\Rl^N\times \Rl^N} \overline{\psi(q)}\kernel{T}(q,q')\varphi(q') \, dq dq';
$$
hence (ii) holds. Assuming (ii), the function $q\mapsto \int \kernel{T}(q,q')|\alpha\rangle(q')\, dq'$ is in $L^2(\Rl^N)$ because $\kernel T$ is square integrable. The scalar product of this vector against $|\alpha'\rangle$ is just the integral expression given in (ii), so that
$$
\sum_{\alpha'} |\matr{T}_{\alpha\vee\alpha'}|^2 = \int dq \left| \int \kernel T(q,q')|\alpha\rangle(q')\, dq'\right|^2
$$
On the other hand, $\sum_\alpha \left|\int \kernel{T}(q,q')|\alpha\rangle(q')\, dq'\right|^2=\int |\kernel T(q,q')|^2dq'$, which is finite for almost all $q$ by Fubini's theorem. Hence, $\sum_{\alpha\vee\alpha'} |\matr{T}_{\alpha\vee\alpha'}|^2=\|\kernel T\|_2^2<\infty$, so (iii) holds. Assuming (iii), we take $\varphi\in L^2(\Rl^N)$, and define a multisequence $([T\varphi]_\alpha)_{\alpha\in\indm N}$ via
$$[T\varphi]_\alpha := \sum_{\alpha'} \matr{T}_{\alpha\vee\alpha'}\langle \varphi |\alpha'\rangle,$$
where the series converges by Cauchy-Schwartz inequality and (iii). In fact we have
$|[T\varphi]_\alpha|^2\leq $ $\sum_{\alpha'}|\matr{T}_{\alpha\vee\alpha'}|^2\|\varphi\|^2.$
From this we see that
$\sum_{\alpha} |[T\varphi]_\alpha|^2\leq \sum_{\alpha\vee\alpha'} |\matr{T}_{\alpha\vee\alpha'}|^2 \, \|\varphi\|^2$, so $T\varphi:=\sum_\alpha [T\varphi]_\alpha |\alpha\rangle \in L^2(\Rl^N)$, and this defines a bounded operator $T:L^2(\Rl^N)\to L^2(\Rl^N)$. Now $T|\alpha\rangle =  \sum_{\alpha'} \matr{T}_{\alpha\vee\alpha'}|\alpha\rangle$ by definition, so
$$\| T\|_2^2 = \tr T^*T = \sum_\alpha \|T|\alpha\rangle\|^2 = \sum_{\alpha\vee\alpha'} |\matr{T}_{\alpha\vee\alpha'}|^2<\infty.$$
This proves that $T$ is Hilbert-Schmidt, i.e. (i) holds. We have now proved (a).

Concerning (b), 
the maps $T\mapsto \kernel T$ and $T\mapsto (\matr{T}_\alpha)$ are obviously linear. Since
$$\langle \alpha |\kernel{\hsbe{\beta}} \rangle = a^{\hsbe{\beta}}_{\alpha}= \langle \hsbe{\alpha}|\hsbe{\beta}\rangle = \delta_{\alpha,\beta}$$
for each $\alpha,\beta\in \indm{2N}$, it follows that $T\mapsto \kernel{T}$ maps $\hsbe{\beta}\mapsto |\beta\rangle$, and $T\mapsto (\matr{T}_\alpha)$ maps $\hsbe{\beta}\mapsto (\alpha\mapsto \delta_{\alpha,\beta})$ so these transform orthonormal bases of the respective Hilbert spaces bijectively onto each other. Hence they are unitary.
\end{proof}

It is a well-known fact that the unitary extension of the Weyl transform to the Hilbert-Schmidt class can be explicitly written in terms of kernels using Lemma \ref{hslemma}. Since we need an explicit formula for this correspondence, we formulate this fact as a second lemma. \emph{We emphasise the following convention}: while the kernel function of an operator is naturally a function on $\Rl^{N}\times\Rl^N=\Rl^{2N}$ (the cartesian product of two copies of the configuration space), its Weyl transform is a function on the phase space $X$ (product of configuration space and the momentum space), with the renormalised measure $dx$. This conceptual difference is reflected in the following notations.

\begin{lem}\label{wtunitary}
The trace class Weyl transform $T\mapsto \widehat{T}$ extends uniquely to the unitary operator
\begin{align*}
&\hs \to L^2(\phaseS),& T&\mapsto \widehat{T} :=U(\idty\otimesH F^*)V\kernel{T},
\end{align*}
where unitary operators $U$ and $V$ are given by
\begin{align*}
U&:L^2(\Rl^{2N})\to L^2(\phaseS), & (U\psi)(q,p) &:= (2\pi)^{N/2}e^{-\frac i2 q\cdot p} \psi(q,p)\\
V&:L^2(\Rl^{2N})\to L^2(\Rl^{2N}),  & (V\psi)(q,q') &:= \psi(q'-q,q').
\end{align*}
\end{lem}
\begin{proof} Suppose first that $T=| \psi\rangle \langle \varphi |$, where $\varphi,\psi\in \sfu$. Then $\kernel{T}(q,q')=\psi(q)\overline{\varphi(q')}$, so using \eqref{Weylrule}, we get
\begin{align*}
\widehat{T}(q,p) &=\langle \varphi |W(q,p)\psi\rangle=\int_{\Rl^N} dq' \,\,\overline{\varphi(q')} e^{-\frac i2 q\cdot p} e^{ip\cdot q'}\psi(q'-q) \\
&=(2\pi)^{N/2}e^{-\frac i2 q\cdot p}\frac{1}{\sqrt {(2\pi)^N}}\int_{\Rl^N} dq' e^{ip\cdot q'} (V\kernel{T})(q,q')\\
&= (U(\idty\otimesH F^*)V\kernel{T})(q,p).
\end{align*}
By linearity, we conclude that $\widehat{T} = U(\idty\otimesH F^*)V\kernel{T}$ for each $T$ of the form \eqref{finiteschwartz}. Moreover, since $U(\idty\otimesH F^*)V$ is a unitary operator $L^2(\Rl^{2N})\to L^2(\phaseS)$, we have $\|\widehat{T}\|_2 =\|\kernel{T}\|_2 = \|T\|_2<\infty$ by Lemma \ref{hslemma}, so $\widehat{T}\in L^2(\phaseS)$. Let now $T\in \tc$ be arbitrary. By Lemma \ref{schwartzdensity}, we can find a sequence $T_n$ of operators of the form \eqref{finiteschwartz}, with $\lim_n \|T-T_n\|_1=0$. This implies that $\lim_{n} \widehat{T_n}(x)= \widehat{T}(x)$ for each $x$ because $|\tr[TW(x)]|\leq \|T\|_1$. On the other hand, the sequence $(\widehat{T_n})$ converges in $L^2(\phaseS)$ to the limit $U(\idty\otimesH F^*)V\kernel{T}$, because
$$\|\widehat{T_n}-U(\idty\otimes_{\HH} F^*)V\kernel{T}\|_2= \|U(\idty\otimesH F^*)V\kernel{T_n-T}\|_2 = \|T_n-T\|_2\leq \|T_n-T\|_1.$$ This implies that $\widehat{T} = U(\idty\otimesH F)V\kernel{T}$. The proof is complete.
\end{proof}

\section{Schwartz operators}
\label{sec:schwartz-operators}

The idea behind Schwartz operators is the desire to find bounded operators $T$
such that expectation values of the form $\tr(P_LTP_R)$ with two polynomials
$P_{L/R}$ in $P$ and $Q$ are well defined and finite. Intuitively, this is a quantum analogue of the requirement of $\int_{\phaseS} P_L(x) f(x) dx$ being well-defined and finite for a Schwartz function $f$. One should pay attention to how non-commutativity of the quantum case makes it necessary to have \emph{two} polynomials instead of one. One might also wonder how to "quantise" the additional requirement involving the \emph{derivatives} of $f$. As we will see below, derivatives in the quantum case are just polynomial multiplications as well; hence the existence of their expectation values do not require additional conditions.

Since $P_{L/R}$ are unbounded operators, the product $P_LTP_R$ might suffer from domain
problems which we need to address. Accordingly, it is appropriate to follow Lemma
\ref{schwartzvector} and define Schwarz operators in terms of quadratic
forms. We will discuss this approach in detail in Subsection
\ref{sec:definition}. Other topics to be presented in this section include:

\begin{itemize}
\item \textbf{Topological properties.} In Subsection \ref{sec:topol-basic-prop} we will
  show that Schwartz operators form a Frechet space; thereby establishing the analogy to the Schwartz functions from the topological perspective. We will give several
  equivalent families of seminorms.
\item
  \textbf{Alternative characterizations.} Apart from the definition we will discuss
  several alternative characterizations of Schwartz operators: In terms of their
  matrices and their Hilbert-Schmidt kernels, shown to be Schwartz functions
  (Subsect. \ref{sec:kern-matr-repr}), and their ranges (Subsect.
  \ref{sec:range-schw-oper}).
\item
\textbf{Harmonic analysis.} Weyl transforms, convolutions, and Wigner functions of Schwartz operators
  (Subsect. \ref{sec:basic-quant-harm}).
\item
  \textbf{Applications of the range theorem.} The results about the range of a Schwartz
  operator (Thm. \ref{SopRange}) have several interesting and useful
  applications, including the square root of a Schwartz operator,
  regularizations of certain unbounded operators, and a ``cycle under the
  trace'' formula; cf. Subsect. \ref{sec:appl-range-theor}
\item
  \textbf{Operations on Schwartz operators} are finally considered in Subsection
  \ref{sec:oper-schw-oper}. This includes products with polynomially bounded
  operators and differentials.
\end{itemize}

\subsection{Definition}
\label{sec:definition}

It is clear how to formulate the definition of Schwartz operators in analogy with the characterization Lemma \ref{schwartzvector} of Schwartz vectors: for each $\alpha,\beta,\alpha',\beta'\in \indm N$, and $T\in \bh$, the sesquilinear form
\begin{align*}
\sfunc N\times \sfunc N\ni (\psi,\varphi)\mapsto \langle P^\beta Q^\alpha \psi |T  P^{\beta'}Q^{\alpha'}\varphi\rangle\in \Cx
\end{align*}
is clearly well defined and jointly continuous. (Continuity is apparent from the seminorms $\|\cdot\|_{\alpha,\beta}$ of Lemma \ref{schwartzvector}). If this form is Hilbert space bounded, i.e.
\begin{align*}
\|T\|_{\alpha,\alpha',\beta,\beta'} := \sup \left\{\,|\langle  P^\beta Q^\alpha\psi \,|\,T P^{\beta'}Q^{\alpha'}\varphi\rangle|\,\,\big|\,\, \psi,\varphi\in \sfunc N, \, \, \|\psi\|\leq 1, \, \|\varphi\|\leq 1\right\}<\infty,
\end{align*}
then there exists a unique $T_{\alpha,\alpha',\beta,\beta'}\in \bh$ such that
\begin{equation}
\langle P^\beta Q^\alpha \psi |T P^{\beta'} Q^{\alpha'}\varphi\rangle= \langle \psi | T_{\alpha,\alpha',\beta,\beta'}\varphi\rangle,\qquad \psi,\varphi\in \sfunc N.
\end{equation}

\begin{defn} Let $T\in \bh$. If $\|T\|_{\alpha,\alpha',\beta,\beta'}<\infty$ for all $\alpha,\alpha',\beta,\beta'\in \indm N$, we say that $T$ is a \emph{Schwartz operator}. The set of Schwartz operators is denoted by $\sop$.
\end{defn}

It is important to stress that $\|T\|_{\alpha,\alpha',\beta,\beta'}$ is defined as the Hilbert norm of a quadratic form on the subspace of Schwartz functions, instead of just formally setting
\begin{equation}\label{formalnorm}
\|T\|_{\alpha,\alpha',\beta,\beta'} = \|Q^\alpha P^\beta T P^{\beta'}Q^{\alpha'}\|.
\end{equation}
For one thing, the operator $Q^\alpha P^\beta T P^{\beta'}Q^{\alpha'}$ is
\emph{a priori} not necessarily well-defined on any dense domain, because $T$
could map outside the domain of $P^\beta$. (This fact will become even more
relevant when we consider \emph{distributions} in Sect. \ref{sec:temp-distr}) Therefore, the formal definition can easily lead to confusion when trying to determine if $\|T\|_{\alpha,\alpha',\beta,\beta'}<\infty$. On the other hand, \emph{if} $T\in \sop$, then the use of \eqref{formalnorm} is permitted, because in that case $Q^\alpha P^\beta T P^{\beta'}Q^{\alpha'}$ \emph{is} well-defined on $\sfunc N$, and $T_{\alpha,\alpha',\beta,\beta'}$ is its bounded extension. This is a consequence of the following simple lemma:

\begin{lem}\label{schwartzinv} If $\|T\|_{\alpha,0,\beta,0}<\infty$ for all $\alpha,\beta\in \indm N$ then ${\rm Ran} (T)\subset \sfunc N$, and $T:\HH\to \sfunc N$ is continuous.
\end{lem}
\begin{proof}
Since $\|T\|_{\alpha,0,\beta,0}<\infty$ and $\sfunc N$ is dense, we have
$$\sup_{\psi\in \sfunc N, \|\psi\|\leq 1} |\langle  P^\beta Q^\alpha\psi |T\varphi\rangle|<\infty$$
for any $\varphi\in \HH$.
Hence, $T\varphi\in \sfunc N$, with $\|T\varphi\|_{\alpha,\beta}\leq \|T\|_{\alpha,0,\beta,0} \|\varphi\|$ for all $\varphi\in \HH$ by Lemma \ref{schwartzvector}.
\end{proof}

\subsection{Topology and basic properties}
\label{sec:topol-basic-prop}

Clearly, each $T\mapsto \|T\|_{\alpha,\alpha',\beta,\beta'}$ is a seminorm. Since they obviously separate points of $\sop$, and because there are countably many of them, they make $\sop$ a metrizable locally convex topological space.
\begin{prop}\label{Frechet} $\sop$ is a Fr\'echet space.
\end{prop}
\begin{proof}
Let $(T_n)$ be a Cauchy sequence in $\sop$. This means that each $(Q^\alpha P^\beta T_n P^{\beta'}Q^{\alpha'})$ is a Cauchy sequence in $\bh$, and hence converges to some $S_{\alpha,\alpha',\beta,\beta'}\in \bh$. In particular, $(T_n)$ converges to $T:=S_{0,0,0,0}$ in the operator norm of $\bh$. Now fix $\alpha,\alpha',\beta,\beta'\in \indm N$. Since norm convergence implies weak convergence, we have
$$
\langle P^\beta Q^\alpha\psi |T P^{\beta'}Q^{\alpha'} \varphi\rangle= \lim_{n\rightarrow\infty} \langle P^\beta Q^\alpha \psi |T_n P^{\beta'}Q^{\alpha'}\varphi\rangle = \langle \psi | S_{\alpha,\alpha',\beta,\beta'}\varphi\rangle
$$
for each $\psi,\varphi\in \sfunc N$. But $S_{\alpha,\alpha',\beta,\beta'}$ is a bounded operator and $\sfunc N$ is dense in $\HH$, so
$\|T\|_{\alpha,\alpha',\beta,\beta',\infty}= \|S_{\alpha,\alpha',\beta,\beta'}\|<\infty$. Hence, $T\in \sop$, and
$Q^\alpha P^\beta T P^{\beta'} Q^{\alpha'}=S_{\alpha,\alpha',\beta,\beta'}$. Consequently, $\|T-T_n\|_{\alpha,\alpha',\beta,\beta'}\rightarrow 0$.
\end{proof}

We can also use a larger family of seminorms:

\begin{prop}\label{schwartznorms} Let $T\in \sop$.
\begin{itemize}
\item[(a)] If $f_L(Q,P)$ and $f_R(Q,P)$ are arbitrary polynomials of $Q$ and $P$, then
$$f_L(Q,P)Tf_R(Q,P)$$ is well-defined on $\sfunc N$, and has a unique bounded extension to $\HH$. We will use $f_L(Q,P)Tf_R(Q,P)$ to denote also the extension.
\item[(b)] The seminorms $\|T\|_{f_L,f_R} := \|f_L(Q,P)Tf_R(Q,P)\|$, where $f_L$ and $f_R$ go through all polynomials, induce the topology of $\sop$.
\end{itemize}
\end{prop}
\begin{proof} The operator $f_R(Q,P)$ maps $\sfunc N$ into itself, so from Lemma \ref{schwartzinv} it follows that $f_L(Q,P)Tf_R(Q,P)$ is defined on $\sfunc N$. Using the commutation relations $[Q_i,P_i]=\delta_{ij}\idty$, which hold on $\sfunc N$, we can write $f_L(Q,P)Tf_R(Q,P)$ on $\sfunc N$ as a linear combination of terms of the form $Q^\alpha P^\alpha T P^{\alpha'}Q^{\alpha'}$. This proves (a) and (b).
\end{proof}

The following lemma gives some basic properties of Schwartz operators. There we use the notation
$$\sopnot:=\{ T\in\sop\mid T \text{ has finite rank}\}.$$

\begin{lem}\label{basiclem}
\begin{itemize}
\item[(a)] If $T\in \sop$ then $T^*\in \sop$. The map $T\mapsto T^*$ is a topological isomorphism.
\item[(b)] If $T,S\in \sop$ and $A\in \bh$, then $TAS\in \sop$. The map
\begin{align*}
(T,S)&\mapsto TAS,  & \sop\times \sop &\to \sop
\end{align*}
is continuous (in the product topology).
\item[(c)] If $T\in \sop$, we have $\varphi_k,\psi_k\in \sfunc N$ whenever $c_k\neq 0$, in the singular value decomposition \eqref{singularvalue} of $T$.
\item[(d)] $T\in \sopnot$ if and only if $T$ is of the form \eqref{finiteschwartz}.
\end{itemize}
\end{lem}
\begin{proof}
Part (a) is obvious from the definition; in fact, $\|T\|_{\alpha,\alpha',\beta,\beta'} = \|T^*\|_{\alpha',\alpha,\beta',\beta}$. Part (b): If $T,S\in \sop$ then both keep $\sfunc N$ invariant, and the operators $T^*P^\beta Q^\alpha $ and $S P^{\beta'}Q^{\alpha'}$ extend to bounded operators by Prop. \ref{schwartznorms}. Since $A$ is bounded, we have
$$
|\langle T^*P^\beta Q^\alpha \psi |A S P^{\beta'}Q^{\alpha'}\varphi\rangle|\leq \|A\|\, \|T^*P^\beta Q^\alpha\|\,\|S P^{\beta'}Q^{\alpha'}\|\|\psi\|\,\|\varphi\|,\quad \text{for all } \psi,\varphi\in \sfunc N,
$$
so that $TAS\in \sop$, and $$\|TAS\|_{\alpha,\alpha',\beta,\beta'}\leq \|A\| \|T\|_{\alpha,0,\beta,0}\|S\|_{0,\alpha',0,\beta'}.$$ Part (c): Fix $k$. Since $T|\psi_k\rangle\langle \varphi_k | T =|c_k|^2|\varphi_k\rangle\langle \psi_k|$, it follows from (b) that $|\varphi_k\rangle\langle \psi_k|\in \sop$ if $c_k\neq 0$. Let $\varphi\in \sfunc N$ be such that $\langle \psi_k|\varphi\rangle \neq 0$. Then $\varphi_k\in \sfunc N$ by Prop. \ref{schwartzinv}. Since $|\varphi_k\rangle\langle \psi_k|^*= |\psi_k\rangle\langle \varphi_k|\in \sop$ by part (a), the same argument shows that $\psi_k\in \sfunc N$. Part (d): Suppose that $T\in \sop$ has finite rank. Then $c_n\neq 0$ only for finite number of $k$ in the singular value decomposition \eqref{singularvalue}, so by (c), $T$ is of the form \eqref{finiteschwartz}. Conversely, suppose $T$ is of this form. Now
$$
|\langle P^\beta Q^\alpha \psi |\varphi_k\rangle\langle \psi_k|P^{\beta'} Q^{\alpha'}\varphi\rangle|\leq \|Q^\alpha P^\beta\varphi_k\|\, \|Q^{\alpha'} P^{\beta'}\psi_k\| \|\psi\|\, \|\varphi\|,
$$
so $|\varphi_k\rangle\langle \psi_k|\in \sop$ for each $k$. Since $T$ is a linear combination of these, also $T\in \sop$.
\end{proof}

The following lemma establishes that each Schwartz operator is in fact trace class, hence also an element of each Schatten class $\schtn p$.
\begin{lem}\label{traceclasslemma} $\sop\subset \tc$.
\end{lem}
\begin{proof} First note that $H^{-2}:=H_1^{-2}H_2^{-2}\cdots H_N^{-2}$ is a positive trace class operator. In fact, $H^{-2}|\alpha\rangle = \prod_{i=1}^N (\alpha_i+\tfrac 12)^{-2}\, |\alpha\rangle$, so $\tr H^{-2} = c^N$ where $c=\sum_{n=0}^\infty (n+\tfrac 12)^{-2}<\infty$. Now if $T$ is a Schwartz operator, it follows from Prop. \ref{schwartznorms} (a) that the operator $H^2 T$ is bounded, where $H^2= H_1^2\cdots H_N^2$. Hence $T = H^{-2} (H^2 T)\in \tc$.
\end{proof}

Using the trace class operator $H^{-2}=H_1^{-2}\cdots H_N^{-2}$ as in the proof of the above Lemma, we immediately see that for a Schwartz operator $T$, each of the bounded operators $Q^\alpha P^\beta T P^{\beta'}Q^{\alpha'}$ is actually in the trace class, and
$$\|Q^\alpha P^\beta T P^{\beta'}Q^{\alpha'}\|_1\leq \|H^{-2}\|_1 \, \|H^2 Q^\alpha P^\beta T P^{\beta'}Q^{\alpha'}\|.$$
On the other hand, since $\|\cdot \|\leq \|\cdot \|_p\leq \|\cdot \|_1$ for each $1\leq p\leq \infty$, it follows that each seminorm
\begin{equation}\label{pseminorms}
\|T\|_{\alpha,\alpha',\beta,\beta',p}:=\|Q^\alpha P^\beta T P^{\beta'}Q^{\alpha'}\|_p
\end{equation}
is finite for $T\in \sop$, and for each fixed $p=1,2,\ldots$, the family $\{ \|\cdot\|_{\alpha,\alpha',\beta,\beta',p}\mid \alpha,\alpha',\beta,\beta'\in \indm N\}$ induces the topology of $\sop$.

Since boundedness already implies $T\in \sop$, the trace class condition appears to be superfluous as a necessary condition for $T\in \sop$. Same holds for other values of $p$. However, the fact that the Hilbert-Schmidt seminorms
\begin{equation}
\|T\|_{\alpha,\alpha',\beta,\beta',2}=\|Q^\alpha P^\beta T P^{\beta'}Q^{\alpha'}\|_2
\end{equation}
induce the topology of $\sop$ is actually especially useful, because the Hilbert space $\hs$ is in a natural way equivalent to $L^2(\phaseS)$ by Lemma \ref{hslemma}. Using this equivalence, Schwartz operators become Schwartz vectors in $L^2(\phaseS)$. This is the topic of the next section.

\subsection{Kernel and matrix representations}

\label{sec:kern-matr-repr}

The following proposition characterises Schwartz operators in terms of their matrix representations and kernels. In essence, it reduces Schwartz operators to Schwartz functions.
\begin{prop}\label{isomorphisms} Let $T\in \hs$. The following conditions are equivalent:
\begin{itemize}
\item[(i)] $T\in \sop$.
\item[(ii)] $\kernel T\in \sfunc{2N}$, i.e. the kernel of $T$ is a Schwartz function.
\item[(iii)] $(\matr{T}_\alpha)\in \ssec{2N}$, i.e. the matrix of $T$ is a Schwartz sequence.
\end{itemize}
Moreover, the maps
\begin{align*}
\sop\ni T&\mapsto \kernel{T}\in \sfunc{2N},  & \sop\ni T &\mapsto a^T\in \ssec{2N}
\end{align*}
(see Lemma \ref{hslemma}) are topological isomorphisms.
\end{prop}
\begin{proof} We first prove the equivalence of (i) and (ii), and that $T\mapsto \kernel{T}$ is an isomorphism. Recall the characterization of Schwartz vectors: $K\in \sfunc {2N}$ if and only if
\begin{equation}\label{schwartz2N}
\sup\{ |\langle P^\beta Q^\alpha \eta | K\rangle| \mid \eta\in \sfunc{2N},\, \|\eta\|\leq 1 \}<\infty \Forall \alpha,\beta\in \indm{2N}.
\end{equation}
Assuming first (ii), we note that for $\varphi,\psi\in \sfunc N$, we have
$$\langle P^{\beta}Q^{\alpha}\psi |TP^{\beta'}Q^{\alpha'}\varphi\rangle = \langle(P^{\beta}Q^{\alpha}\psi)\otimes C(P^{\beta'}Q^{\alpha'}\varphi)|\kernel{T}\rangle=(-1)^{|\beta'|}\langle P^{\beta\vee\beta'}Q^{\alpha\vee\alpha'}\psi\otimes C\varphi|\kernel{T}\rangle,$$
so it follows immediately from \eqref{schwartz2N} that $T\in \sop$. Now assume (i). Since the operator $\tilde{T}:=Q^\alpha P^\beta TP^\beta Q^\alpha$ is trace class, it has a kernel $\kernel{\tilde T}$ according to Lemma \ref{hslemma}, and $\langle \psi | \tilde{T}\varphi\rangle = \langle \psi\otimes C\varphi |\kernel{\tilde T}\rangle$ for all $\psi,\varphi\in \HH$. On the other hand, for $\psi,\varphi\in \sfunc N$, we can also use the kernel of $T$ to write
$$
\langle \psi\otimes C\varphi |\kernel{\tilde T}\rangle=\langle \psi | \tilde{T}\varphi\rangle = \langle P^{\beta}Q^{\alpha}\psi |TP^{\beta'}Q^{\alpha'}\varphi\rangle =(-1)^{|\beta'|}\langle P^{\beta\vee\beta'}Q^{\alpha\vee\alpha'}\psi\otimes C\varphi|\kernel{T}\rangle.
$$
Since both $\eta\mapsto \langle \eta |\kernel{\tilde T}\rangle$ and $\eta\mapsto \langle P^{\beta\vee\beta'}Q^{\alpha\vee\alpha'}\eta |\kernel{T}\rangle$ are continuous on $\sfunc{2N}$, and the algebraic tensor product $\sfunc N \otimes \sfunc N$ is dense in $\sfunc {2N}$ by Prop. \ref{sfdensity} (c), it follows that
$\langle \eta |\kernel{\tilde T}\rangle=(-1)^{|\beta'|}\langle P^{\beta\vee\beta'}Q^{\alpha\vee\alpha'}\eta|\kernel{T}\rangle$ for all $\eta\in \sfunc {2N}$. This implies \eqref{schwartz2N}, so (ii) holds. We now show that $T\mapsto \kernel{T}$ is a topological isomorphism. It is clearly injective. If $K\in \sfunc {2N}$, it  determines a Hilbert-Schmidt operator $T$ such that $\kernel{T} = K$, and by the equivalence we just proved, $T\in \sop$. Hence, the map is onto. Finally, from the above proof it follows that
$$\|T\|_{\alpha,\alpha',\beta,\beta',2} = \|\tilde T\|_2 = \|\kernel{\tilde T}\|_2= \|Q^{\alpha\vee\alpha'}P^{\beta\vee\beta'}\kernel{T}\|_2= \|\kernel{T}\|_{\alpha\vee\alpha', \beta\vee\beta'},$$
so the topologies are equivalent.

Concerning condition (iii), we already know from the preceding section that $\sfunc {2N}$ is isomorphic to $\ssec{2N}$ via the map $\eta\mapsto (\langle \alpha|\eta\rangle)_{\alpha\in \indm{2N}}$. Since its composition with $T\mapsto \kernel{T}$ is exactly the map $T\mapsto (\matr{T}_\alpha)$, (iii) is equivalent to (ii).
\end{proof}

As an immediate corollary, we get the following density properties of Schwartz operators, using Prop. \ref{sfdensity}.
\begin{prop}\label{schwartzproperties}
\begin{itemize}
\item[(a)] The matrix representation $T= \sum_{\alpha\in \indm{2N}} \matr{T}_{\alpha} \hsbe{\alpha}$ of a $T\in \sop$ converges in the topology of $\sop$.
\item[(b)] ${\rm span}\{ \hsbe{\alpha}\mid \alpha\in \indm{2N}\}$ is dense in $\sop$. In particular, $\sopnot$ is dense in $\sop$.
\end{itemize}
\end{prop}

Another consequence is obtained by estimating the sequence space seminorms as follows:
\begin{align*}
\|(\matr{T}_\alpha)\|_{\beta\vee \beta'}^2&=\sum_{\alpha,\alpha'\in \indm N} (\alpha+1)^{2\beta}(\alpha'+1)^{2\beta'} |\langle \alpha |T|\alpha'\rangle|^2\\
&=\sum_{\alpha,\alpha'\in \indm N} |\langle \alpha |(H+\tfrac 12)^{\beta}T(H+\tfrac 12)^{\beta'}|\alpha'\rangle|^2 =\|(H+\tfrac 12)^{\beta} T(H+\tfrac 12)^{\beta'}\|_2^2\\
\|(\matr{T}_\alpha)\|_{\beta\vee \beta'}^2 &\leq \sum_{\alpha,\alpha'\in \indm N} \left[C_1\left(\sum_{i=1}^N (\alpha_i+\frac 12)\right)^{2|\beta|} \left(\sum_{i=1}^N (\alpha_i+\frac 12)\right)^{2|\beta'|} + C_2\right] |\langle \alpha |T|\alpha'\rangle|^2\\
&=\sum_{\alpha,\alpha'\in \indm N} \left(C_1 |\langle \alpha |\Htot^{|\beta|}T\Htot^{|\beta'|}|\alpha'\rangle|^2+C_2|\langle \alpha |T|\alpha'\rangle|^2\right)\\ &=C_1\|\Htot^{|\beta|}T\Htot^{|\beta'|}\|_2^2+C_2 \|T\|_2^2,
\end{align*}
where $C_1,C_2>0$ only depend on the $\beta_i,\beta_i'$. Hence, it follows from Prop. \ref{isomorphisms} that the family of seminorms
\begin{align*}
T&\mapsto \|H^\alpha T H^{\alpha'}\|_2, & \alpha,\alpha'\in \indm N,
\end{align*}
as well as the family
\begin{align*}
T&\mapsto \|\Htot^m T \Htot^{m'}\|_2, & m,m'\in \Nl,
\end{align*}
induce the topology of $\sop$.

\subsection{Range of a Schwartz operator}

\label{sec:range-schw-oper}

In this section, we will characterize Schwartz operators in terms of their range (Prop. \ref{SopRange} below). This motivates the following definition: We put
$$\sopL:=\{ T\in \bh \mid {\rm Ran}(T)\subset \sfunc N\}.$$

We now need a simple consequence of the closed graph theorem:
\begin{lem}\label{clg}
Let $A$ be closed (unbounded) operator in $\HH_N$, and let $T\in \bh$ be such that ${\rm Ran}(T)\subset \dom(A)$. Then $AT\in \bh$. The same holds true if $A$ is any polynomial of closed operators.
\end{lem}
\begin{proof} Let $(\varphi_n)$ be a sequence of elements of $\HH_N$, converging to a $\varphi\in \HH$, and such that also $(AT\varphi_n)$ converges. (This is equivalent to saying that $(\varphi_n)$ converges in the graph norm of $AT$.) But $T$ is bounded, so $(T\varphi_n)$ converges to $T\varphi$. Since $A$ is closed, and $T\varphi_n\in \dom(A)$ for each $n$, this implies that ($T\varphi\in \dom(A)$ and) $AT\varphi = \lim_n AT_n$. Hence, $AT$ is closed. Since it is everywhere defined, it is bounded by the closed graph theorem (see e.g. \cite[Thm. III.12]{RSI}). The last statement of the lemma follows by induction; recall that by definition $$\dom(A_1A_2) =\{\varphi\in \dom(A_2)\mid A_2\varphi\in \dom(A_1)\}$$ for any two operators $A_1,A_2$.
\end{proof}

The following proposition characterizes $\sopL$, and gives it a natural topology.
\begin{prop}\label{leftSchwartz} Let $T\in \bh$. The following conditions are equivalent:
\begin{itemize}
\item[(i)] $T\in \sopL$.
\item[(ii)] $Q^\alpha P^\beta T\in \bh$ for each $\alpha,\beta\in \indm N$.
\item[(iii)] $\|T\|_{\alpha,0,\beta,0}<\infty$ for each $\alpha,\beta\in \indm N$.
\item[(iv)] $Q^\alpha P^\beta T\in \schtn p$ for some (resp. all $p=1,2,\ldots$), and each $\alpha,\beta\in \indm N$.
\item[(v)] $H^\beta T\in \schtn p$ for some (resp. all $p=1,2,\ldots$), and each $\beta\in \indm N$.
\item[(vi)] $\|(\matr{T}_\alpha)\|_{\beta\vee 0}<\infty$ for each $\beta\in \indm N$.
\end{itemize}
The seminorms
\begin{equation}\label{halfseminorms}
T\mapsto \sup_{\varphi\in \HH, \|\varphi\|\leq 1} \|T\varphi\|_{\alpha,\beta}=\|T\|_{\alpha,0,\beta,0} = \|Q^{\alpha}P^\beta T\|, \Forall \alpha,\beta\in \indm N
\end{equation}
induce a locally convex topology on $\sopL$, which makes it a Frechet space. The families
\begin{align*}
\{T &\mapsto \|Q^\alpha P^\beta T\|_2, \mid\alpha,\beta\in \indm N\},  & \{T&\mapsto \|(\matr{T}_\alpha)\|_{\beta\vee 0}\mid  \beta\in \indm N\}, & \{T&\mapsto \|H^\alpha T\|_2\mid \alpha\in \indm N\}
\end{align*}
of Hilbert space seminorms each induce this topology.
\end{prop}
\begin{proof} Assuming (i), we have $P^\beta T\in \bh$ by Lemma \ref{clg}, because $P^\beta$ is closed (on its full domain). Since $Q^\alpha$ is also closed, and ${\rm Ran}(P^\beta T)\subset \sfunc N$, we can apply Lemma \ref{clg} again to conclude that $Q^\alpha P^\beta T\in \bh$. Hence (ii) holds. (ii) obviously implies (iii), which again implies (ii) by Lemma \ref{schwartzinv}. Assuming (ii), we can multiply $Q^\alpha P^\beta T$ from the left with $H^{-2}$ as in the proof of Lemma \ref{basiclem} (c) to conclude that that (iv) holds for $p=1$, and hence for all $p$. Assuming (iv) for \emph{some} $p$ implies, in particular, that $Q^\alpha P^\beta T\in \bh$, for all $\alpha,\beta$ which again implies that (iv) holds for $p=1$ and thus for all $p$. Trivially, (iv) for all $p$ implies (v) for all $p$. Again (v) holds for all $p$ iff it holds for some $p$, because the above argument with $H^{-2}$ works also here. Assuming (v), we know that $T$ maps into the domain of each $H^\alpha$, and we have
\begin{align*}
\|(\matr{T}_\alpha)\|_{\beta\vee 0}^2&=\sum_{\alpha,\alpha'\in \indm N} (\alpha+1)^{2\beta} |\langle \alpha |T|\alpha'\rangle|^2\\
&=\sum_{\alpha,\alpha'\in \indm N} |\langle \alpha |(H+\tfrac 12)^{\beta}T|\alpha'\rangle|^2 =\|(H+\tfrac 12)^{\beta} T\|_2^2<\infty.
\end{align*}
Hence (vi) holds. Assuming (vi) we have
\begin{align*}
\sup_{\varphi\in \HH, \|\varphi\|\leq 1} \sum_{\alpha\in \indm N} (\alpha+1)^{2\beta} |\langle \alpha |T\varphi\rangle|^2 &=  \sum_{\alpha\in \indm N} (\alpha+1)^{2\beta}\sup_{\varphi\in \HH, \|\varphi\|\leq 1}|\langle T^*|\alpha\rangle |\varphi\rangle|^2\\
&= \sum_{\alpha\in \indm N} (\alpha+1)^{2\beta}\|T^*|\alpha\rangle\|^2\\
& = \sum_{\alpha,\alpha'\in \indm N} (\alpha+1)^{2\beta} |\matr{T}_{\alpha\vee\alpha'}|^2= \|(\matr{T}_\alpha)\|_{\beta\vee 0}^2,
\end{align*}
so in particular, $T\varphi\in \sfunc N$. We have now proved the equivalences. The topology induced by the seminorms \eqref{halfseminorms} is Frechet by a similar argument as in Prop. \ref{Frechet}. Clearly, we can replace the $\sfunc N$-seminorms $\|\cdot\|_{\alpha,\beta}$ in \eqref{halfseminorms} by any other family inducing the topology of $\sfunc N$, and still get the same topology for $\sop$. Using the seminorms of the sequence space $\ssec N$, the last two computations above show that the seminorms
$$
T\mapsto \|(H+\tfrac 12)^{\beta} T\|_2=\|(\matr{T}_\alpha)\|_{\beta\vee 0}= \sup_{\varphi\in \HH, \|\varphi\|\leq 1} \|(a^{T\varphi}_\alpha)\|_\beta
$$
induce the topology of $\sopL$. Since
$$\|Q^{\alpha}P^\beta T\|\leq \|Q^{\alpha}P^\beta T\|_2\leq \|Q^{\alpha}P^\beta T\|_1\leq \|H^{-2}\|_1 \|H^2Q^{\alpha}P^\beta T\|,$$ it follows that also the seminorms $T\mapsto \|Q^{\alpha}P^\beta T\|_2$ induce the topology of $\sopL$.
\end{proof}

It is clear from the above proposition that $\sop\subset \sopL$, while the converse inclusion is obviously not true. Indeed, if $T=|\psi\rangle\langle \varphi |$, then $T\in \sopL$ iff $\psi\in \sfunc N$ (without any condition on $\varphi$). In order to make $T\in \sop$, we also need $\varphi\in \sfunc N$, that is, $T^*\in \sopL$. It turns out that these two conditions characterize $\sop$ completely. We first define
$$
\sopR := \{ T\in \bh \mid T^*\in \sopL \}
$$
We equip $\sopR$ with the topology coming from $\sopL$ in the obvious way. We now have the following characterization:
\begin{prop}\label{rightSchwartz} Let $T\in \bh$. The following conditions are equivalent:
\begin{itemize}
\item[(i)] $T\in \sopR$.
\item[(ii)] $TP^{\beta} Q^{\alpha}$ extends to a bounded operator for each $\alpha,\beta\in \indm N$.
\item[(iii)] $\|T\|_{0,\alpha,0,\beta}<\infty$ for each $\alpha,\beta\in \indm N$.
\item[(iv)] $TP^{\beta} Q^{\alpha}$ extends to a an operator of $\schtn p$ for some (resp. all $p=1,2,\ldots$), and each $\alpha,\beta\in \indm N$.
\item[(v)] $TH^{\beta}$ extends to an operator of $\schtn p$ for some (resp. all $p=1,2,\ldots$), and each $\beta\in \indm N$.
\item[(vi)] $\|(\matr{T}_\alpha)\|_{0\vee\beta}<\infty$ for each $\beta\in \indm N$.
\end{itemize}
The family of seminorms
\begin{equation}\label{halfseminorms2}
T\mapsto \sup_{\varphi\in \HH, \|\varphi\|\leq 1} \|T^*\varphi\|_{\alpha,\beta}=\|T\|_{0,\alpha,0\beta} = \|TP^{\beta} Q^{\alpha}\|, \Forall \alpha,\beta\in \indm N,
\end{equation}
as well as each family
\begin{align*}
\{T &\mapsto \|TP^{\beta} Q^{\alpha}\|_2, \mid\alpha,\beta\in \indm N\},  & \{T&\mapsto \|(\matr{T}_\alpha)\|_{0\vee\beta}\mid  \beta\in \indm N\}, & \{T&\mapsto \|TH^{\alpha}\|_2\mid \alpha\in \indm N\}
\end{align*}
of Hilbert space seminorms induce the topology of $\sopR$.
\end{prop}
\begin{proof} If (i) holds then $Q^{\alpha}P^{\beta}T^*\in \bh$ by Prop. \ref{leftSchwartz}. But the adjoint of this operator is an extension of the densely defined operator $TP^{\beta} Q^{\alpha}$, so (ii) holds. (ii) is clearly equivalent to (iii). Assuming (ii), and noticing that $TP^{\beta} Q^{\alpha}\varphi = TP^{\beta} Q^{\alpha}H^{2}H^{-2}\varphi$ for all $\varphi\in\sfunc N$, we see that the bounded extension of $TP^{\beta} Q^{\alpha}$ is trace class, and hence in each $\schtn p$. If (iv) holds for some $p$ the extension is in particular bounded, so again by multiplying with $H^{-2}$ we see that (iv) holds for all $p$. Assuming this, (v) is clear. Assuming (v) compute $\|(\matr{T}_\alpha)\|_{0\vee\beta}=\|T(H+\tfrac 12)^{\beta}\|_2<\infty$ so (vi) holds. To get (i) from (vi) just note that $\|(\matr{T}_\alpha)\|_{0\vee\beta}=\|(\matr{(T^*)}_\alpha)\|_{\beta\vee 0}$, and use Prop. \ref{leftSchwartz}.
\end{proof}

We now get the following neat characterization of $\sop$ in terms of the ranges of $T$ and $T^*$:

\begin{thm}\label{SopRange}
Let $T\in \bh$. Then $T\in \sop$ if and only if
\begin{align*}
{\rm Ran}(T)&\subset \sfunc N, & {\rm and} &  &  {\rm Ran}(T^*)&\subset \sfunc N.
\end{align*}
The topology of $\sop$ is induced by the seminorms
\begin{align*}
\|T\|_{\alpha,0,\beta,0}&= \sup_{\varphi\in \HH, \|\varphi\|\leq 1} \|T\varphi\|_{\alpha,\beta}, & \|T\|_{0,\alpha,0,\beta} &=\sup_{\varphi\in \HH, \|\varphi\|\leq 1} \|T^*\varphi\|_{\alpha,\beta}, & \alpha,\beta\in \indm N.
\end{align*}
\end{thm}
\begin{proof} If $T\in \sop$ it is clear that the inclusion of the ranges follows. We now assume ${\rm Ran}(T)\subset \sfunc N$ and ${\rm Ran}(T^*)\subset \sfunc N$. From Props. \ref{leftSchwartz} and \ref{rightSchwartz} it follows that $\|(\matr{T}_\alpha)\|_{\beta\vee 0}<\infty$ and $\|(\matr{T}_\alpha)\|_{0\vee\beta}<\infty$ for all $\beta,\beta'\in \indm N$. But now the Cauchy-Schwartz inequality gives
\begin{align*}
\|(\matr{T}_\alpha)\|_{\beta\vee\beta'}^2&=\sum_{\alpha\vee\alpha'\in \indm{2N}} (\alpha\vee\alpha'+1)^{2\beta\vee\beta'} |\matr{T}_{\alpha\vee\alpha'}|^2= \sum_{\alpha,\alpha'\in \indm{N}} (\alpha+1)^{2\beta}(\alpha'+1)^{2\beta'} |\matr{T}_{\alpha\vee\alpha'}|^2\\
&\leq \sqrt{\sum_{\alpha,\alpha'\in \indm{N}} (\alpha+1)^{4\beta} |\matr{T}_{\alpha\vee\alpha'}|^2\sum_{\alpha,\alpha'\in \indm{N}} (\alpha'+1)^{4\beta'} |\matr{T}_{\alpha\vee\alpha'}|^2}\\
&= \|(\matr{T}_\alpha)\|_{2\beta\vee 0}\|(\matr{T}_\alpha)\|_{0\vee2\beta'}<\infty,
\end{align*}
so $(\matr{T}_\alpha)\in \ssec{2N}$, i.e. $T\in \sop$. Moreover, we have
$$\|(\matr{T}_\alpha)\|_{\beta\vee\beta'}\leq \sqrt{\|(\matr{T}_\alpha)\|_{2\beta\vee 0}\|(\matr{T}_\alpha)\|_{0\vee2\beta'}}\leq \frac 12 (\|(\matr{T}_\alpha)\|_{2\beta\vee 0}+\|(\matr{T}_\alpha)\|_{0\vee2\beta'}),$$
which proves that the restricted family
$$
\{ T\mapsto \|(\matr{T}_\alpha)\|_{\beta\vee 0}, \, T\mapsto \|(\matr{T}_\alpha)\|_{0\vee\beta}\mid \beta\in \indm N\}
$$
induces the topology of $\sop$. But according to Props. \ref{leftSchwartz} and \ref{rightSchwartz}, these are the seminorms of $\sopL$ and $\sopR$ put together. This completes the proof.
\end{proof}

\subsection{Applications of the range theorem}
\label{sec:appl-range-theor}

The theorem just proven has a number of interesting consequences which are
collected in this Subsection. The first shows that multiplication by a
Schwartz operator regularizes closable unbounded operators defined on the
Schwartz space.

\begin{prop}\label{sopreg} Let $A$ be a closed operator in $\HH$, with $\dom(A)\supset \sfunc N$, and let $T\in \sop$.
\begin{itemize}
\item[(a)] $AT\in \sopR$, and $TA^*$ is closable with closure in $\sopL$.
\item[(b)] If $\sfunc N$ is invariant for $A$ then $AT\in \sop$, and the extension of $TA^*$ is in $\sop$.
\end{itemize}
\end{prop}
\begin{proof} Part (a): Since $T$ maps all of $\HH$ into $\sfunc N$ by Prop. \ref{SopRange}, it follows from Lemma \ref{clg} that $AT\in \bh$. Now $ATP^\beta Q^\alpha$ is densely defined on $\sfunc N$, and since $TP^\beta Q^\alpha$ extends to a Schwartz operator $\tilde{T}$ which, as such, maps $\HH$ into $\sfunc N$, it follows that $ATP^\beta Q^\alpha$ extends to $A\tilde{T}$ which is again bounded by Lemma \ref{clg}. Hence, it follows from Prop. \ref{rightSchwartz} that $AT\in \sopR$. Since $A$ is closed, the operator $TA^*$ is densely defined, so its adjoint $(TA^*)^*$ is well defined, and is an extension of $AT^*$. Hence $(TA^*)^*=AT^*\in \sopR$ by the previous argument, so $TA^*$ is closable with closure in $\sopL$. Part (b): We already know from (a) that $AT\in \sopR$ and $(TA^*)^{**}\in \sopL$. If $\sfunc N$ is invariant for $A$ then the ranges of the bounded operators $AT$ and $AT^*$ are in $\sfunc N$, so $AT\in \sopL$, and $(TA^*)^{**}=(AT^*)^*\in \sopR$. Hence the claim follows from Prop. \ref{SopRange}.
\end{proof}

This has the following corollary:
\begin{prop}\label{sopreg2} Let $T\in \sop$, and let $A$ and $B$ be closed operators in $\HH_N$ with domains containing $\sfunc N$ as an invariant subspace. Then $ATB^*$ is closable with closure in $\sop$.
\end{prop}
\begin{proof} It follows from Prop. \ref{sopreg} (b) that $A(TB^*)^{**}\in \sop$. But this is obviously a bounded extension of the operator $ATB^*$, which is densely defined (with domain $\dom(B^*)$) because $B^*$ is closed. Hence $ATB^*$ is closable, and the closure must coincide with $A(TB^*)^{**}\in \sop$.
\end{proof}

In Subsect. \ref{sec:topol-basic-prop} we have seen that the eigenvectors of a
positive Schwartz operators are Schwartz functions, however, there is no
statement yet about the eigenvalues. This gap is closed by the next result.

\begin{prop} \label{prop:1}
  The root of a positive Schwartz operator is again a Schwartz operator.
\end{prop}

\begin{proof}
  Let $T$ be a positive Schwartz operator. Take arbitrary $\alpha,\beta \in \indm{N}$, and $\varphi \in \sfunc N$ with $\|\varphi\|=1$. Then
  \begin{align}
    \left\|\sqrt{T} P^\beta Q^\alpha \varphi\right\|^2 &= \left\langle
      P^\beta Q^\alpha \varphi| T P^\beta Q^\alpha \varphi \right\rangle \leq \|T\|_{\alpha,\alpha,\beta,\beta}<\infty,
 \label{eq:19}
  \end{align}
showing that $\sqrt{T}P^\beta Q^\alpha$ extends (in a unique way) to a bounded operator on $\mathcal{H}$.
%
According to the equivalence of (i) and (ii) in Prop. \ref{rightSchwartz} this implies that
  \begin{displaymath}
    \sqrt{T} \in \sopR = \{ T \in \mathcal{B}(\mathcal{H})\,|\, T^* \in
    \sopL \}.
  \end{displaymath}
  But since $\sqrt{T}$ is selfadjoint this implies that $\sqrt{T} \in
  \sopL$ holds as well. Hence the statement follows from Thm. \ref{SopRange}.
\end{proof}

\begin{cor} Let $T\in \sop$. Then $|T|\in \sop$ as well. The singular values $c_k$ of $T$ satisfy the fall-off condition $\sum_k c_k^{\frac{1}{2n}}<\infty$ for all $n\in \mathbb N$.
\end{cor}

The condition in the Lemma is not compatible with any power law, but does not imply polynomial or faster decay, because there can be a subsequence with rapidly increasing $k$, along which $k c_k$ is unbounded. This leaves open the question what the optimal decay statement in the Lemma might be.

\begin{proof} Since $|T|=\sqrt{T^*T}$, the first claim follows from Lemma \ref{basiclem} (a) and (b), and the above Proposition. By applying the Proposition $n$ times to the positive Schwartz operator $|T|$ we
see that the series $\sum_k c_k^{1/2n}$ with the singular values
$c_k$, $k \in \mathbb{N}$ of $T$ converges for all $n \in \mathbb{N}$. (Recall that every Schwartz operator is trace class by Lemma \ref{traceclasslemma}.)
\end{proof}

Finally we can use the last result to prove the following ``cycling under the trace'' result.

\begin{cor}\label{tracecycling}
  For each $T \in \sop$ and each closed operator $A$ such that $\sfunc N\subset \dom(A)\cap \dom(A^*)$, we have $\tr[T A] = \tr[A T]$.
\end{cor}

\begin{proof}
  Note first that for any $T \in \sop$ we have $T\pm T^* \in \sop$, and a selfadjoint $T$ can be decomposed
  according to $T = T_+ - T_-$ with $T_\pm = (T \pm |T|)/2$. Obviously $T_\pm > 0$ and $T_\pm \in \sop$. Hence
  each $T \in \sop$ can be written as a linear combination of four positive Schwartz operators. Therefore it
  is sufficient to prove the statement for $T > 0$. In this case $\sqrt{T}$ exists and is a Schwartz operator,
  $\sqrt{T} A\in \sopL$ and $A \sqrt{T}\in \sopR$ by Prop. \ref{sopreg} (a), hence both are trace class by
  Props. \ref{leftSchwartz} and \ref{rightSchwartz}. Since the Schwartz operator $\sqrt{T}$ is also trace
  class we get
  \begin{displaymath}
    \tr(T A) = \tr\left( \sqrt{T} \sqrt{T} A\right) = \tr\left(\sqrt{T}
      A \sqrt{T}\right) = \tr\left( A \sqrt{T} \sqrt{T}\right) =
      \tr(A T).
  \end{displaymath}
\end{proof}

\subsection{Basic quantum harmonic analysis on Schwartz operators and functions}

\label{sec:basic-quant-harm}

According to the correspondence theory \cite{Werner}, operators on $\HH$ correspond to functions on $\phaseS$ via the convolutions
defined above, and these are compatible with the Fourier-Weyl transform in the sense of \eqref{fourierconv}. This correspondence works also on the level of Schwartz operators and Schwartz functions. 

We know that the symplectic Fourier transform $f\mapsto \widehat{f}$ is a topological isomorphism of $\sfu$ onto itself. The following proposition shows that an analogous statement holds for the Weyl transform.

\begin{prop} $T\in \sop$ if and only if $\widehat T\in \sfu$. The map
$$
\sop\ni T\mapsto \widehat T\in \sfu
$$
is a topological isomorphism.
\end{prop}
\begin{proof}
We prove that the $\widehat T\in \sfu$ if and only if the kernel condition Prop. \ref{isomorphisms} (ii) holds, and that the restriction of the Hilbert-Schmidt Weyl transform is a topological isomorphism between $\sop$ and $\sfu$. In the notation of Lemma \ref{wtunitary}, we have $\widehat{T} = U(\idty\otimes F^*)V\kernel{T}$ for all $T\in \hs$. Now the unitary operator $U$ obviously maps $\sfunc{2N}$ onto $\sfu$, and $V$ keeps $\sfunc {2N}$ invariant, with the corresponding restrictions being continuous in the topology of $\sfunc{2N}$. What we need in addition, is that $\idty\otimes F^*$ has the same property. But we know that $F^*$ keeps $\sfunc N$ invariant, with the restriction being continuous, so we only need to apply Prop. \ref{sfdensity} (e) and (d). The proof is complete.
\end{proof}

Since multiplication by a Schwartz function is continuous in $\sfu$, it follows immediately from \eqref{fourierconv} that for a fixed $S_0\in \sop$, the convolutions keep the Schwartz spaces invariant, and
\begin{align*}
\sfu\to \sop,&\, f\mapsto f*S_0, & \sop\to \sfu, & \,T\mapsto T*S_0
\end{align*}
are continuous. 
These correspondence maps are not surjective; the best one can hope for is that the range is dense. This holds at least for the ground state of $\Htot$:

\begin{lem}\label{schwartzreg} Let $S_0=|0\rangle\langle 0|$, i.e. the ground state of $\Htot$. Then the range of $T\mapsto T*S_0$ is dense in $\sfu$.
\end{lem}
\begin{proof} According to Prop. \ref{sfdensity} (b), the linear span of the vectors $|\alpha\vee\alpha'\rangle$, $\alpha,\alpha'\in \indm N$ is dense in $\sfu$. Fix $\alpha,\alpha'\in \indm N$. Then $|\alpha\vee\alpha'\rangle$ is an eigenfunction of the Fourier-Plancherel operator $F$ on $L^2(\Rl^{2N})$. Since $\widehat{|\alpha\vee\alpha'\rangle}(q,p)= (F|\alpha\vee\alpha'\rangle)(-p,q)$, we have
$$
\widehat{|\alpha\vee\alpha'\rangle}(q,p)= {\rm const.}\, e^{-\frac 12(q^2+p^2)}\prod_{i=1}^N H_{\alpha_i}(-p_i)H_{\alpha'_i}(q_i),
$$
where the $H_n$ are Hermite polynomials. Since $\widehat{S_0}(q,p) = e^{-\frac 14(q^2+p^2)}$, we see that $\widehat{S_0}^{-1}\widehat{|\alpha\vee\alpha'\rangle}\in \sfu$.
Hence there is an $T\in \sop$ such that $\widehat{T} = \widehat{S_0}^{-1}\widehat{|\alpha\vee\alpha'\rangle} $;
this has $T*S_0 = |\alpha,\alpha'\rangle$. This completes the proof.
\end{proof}

Another commonly used correspondence between functions on $\phaseS$ and operators on $\HH$ is the \emph{Weyl quantization}, of which there exists a large amount of literature; see e.g. \cite{Ali} and the references therein. For Schwartz functions, the Weyl quantisation is defined by
\begin{align}\label{wignersq}
\sfu\to \sop, \quad & f\mapsto \wq{f}:=\check{\widehat{f}}_-,
\end{align}
where (as defined above), $f\mapsto \check{f}$ denotes the inverse Weyl transform. Hence, the Weyl quantization also provides a topological isomorphism between $\sfu$ and $\sop$. The value of the inverse transform $\wf{T}:=\widehat{\widehat{T}}_-$ is called the \emph{Wigner function} of $T\in \sop$. By definition, this has
\begin{align*}
\int dx \wf{T}(x) &=\widehat{T}_-(0) = {\rm tr}[T], & \int dx |\wf{T}(x)|^2 = \|T\|_2^2.
\end{align*}
In order to check that this leads to the standard definition, we take $T=|\psi\rangle\langle \varphi|\in \sop$, and compute the Wigner function at $x=(q,p)$:
\begin{align*}
\wf{T}(x) &= \int e^{i\{x,y\}} {\rm tr}[W(y)T] dy\\
&= \int \frac{dq'dp'}{(2\pi)^N} e^{i(q'\cdot p-q\cdot p')}\int dq'' \overline{\varphi(q'')} e^{-iq'\cdot p'/2}e^{ip'\cdot q''}\psi(q''-q')\\
&= \int_{\Rl^N} dq' e^{iq'\cdot p}\frac{1}{\sqrt{(2\pi)^N}} \int_{\Rl^N} dp' e^{-i(q'/2+q)\cdot p'}\frac{1}{\sqrt{(2\pi)^N}}\int_{\Rl^N}dq''\,e^{ip'\cdot q''}\overline{\varphi(q'')}\psi(q''-q')\\
&= \int_{\Rl^N} dq' e^{iq'\cdot p} \overline{\varphi(q'/2+q)}\psi((q'/2+q)-q')
=\int_{\Rl^N} dq' e^{iq'\cdot p} \overline{\varphi(q+q'/2)}\psi(q-q'/2)\\
&=2^N\int_{\Rl^N} dq' e^{2ip\cdot q'} \overline{\varphi(q+q')}\psi(q-q')\\
&=2^N\int_{\Rl^N} dq' e^{2ip\cdot q'} \overline{\varphi(q+q')}\psi(q-q')\\
&=2^N\langle W(-x)\varphi|\Pi W(-x)\psi\rangle = 2{\rm tr}[W(x)\Pi W(x)^*T].
\end{align*}
Since both sides are continuous with respect to $T\in \sop$, we have, in general
\begin{equation}
\wf{T}(x)=2^N{\rm tr}[W(x)\Pi W(x)^*T], \quad T\in \sop.
\end{equation}
From the computation we also get the commonly used formula for the Wigner function:
$$
\wf{|\psi\rangle\langle \varphi|}(x)=\int_{\Rl^N} dq' e^{iq'\cdot p} \overline{\varphi(q+q'/2)}\psi(q-q'/2).
$$
Using \eqref{inverseWeyl}, we get the Wigner quantization of an $f\in \sfu$:
\begin{align*}
\langle \varphi |\wq{f}\psi\rangle &= \int dy \langle \varphi |W(y)\psi\rangle \int e^{i\{x,y\}}f(x)dx
= \int dx f(x) \widehat{\widehat{|\psi\rangle\langle \varphi|}}_-(x)\\
&= \int \frac{dqdp}{(2\pi)^N} f(q,p) \int_{\Rl^N} dq' e^{iq'\cdot p} \overline{\varphi(q+q'/2)}\psi(q-q'/2)\\
&= \int \frac{dqdp}{(2\pi)^N} f(q,p) 2\int_{\Rl^N} dq' e^{2i(q'-q)\cdot p} \overline{\varphi(q')}\psi(2q-q')\\
&= \int_{\Rl^N} dq' \int 2dqdp \frac{1}{(2\pi)^{N}}\,f(q,p)e^{2i(q'-q)\cdot p} \overline{\varphi(q')}\psi(2q-q')\\
&= \int_{\Rl^N} dq' \int dq \frac{1}{(2\pi)^{N}}\int dp \,f\left(\frac{q+q'}{2},p\right)e^{i(q'-q)\cdot p}\, \overline{\varphi(q')}\psi(q).
\end{align*}
Hence, the kernel of $\wq{f}$ is given by
\begin{equation}\label{wignerkernel}
\kernel{\wq{f}}(q,q')= \frac{1}{(2\pi)^{N}}\int dp \,f\left(\frac{q'+q}{2},p\right)e^{i(q-q')\cdot p}.
\end{equation}
From the above computation we also see that
$$
{\rm tr}[\wq{f}T] = \int dx f(x) \wf{T}(x), \qquad T\in \sop.
$$

The basic example is the Gaussian state $T_0=|h_0\rangle\langle h_0|$, where $h_0(q) = \frac{1}{\pi^{1/4}} e^{-\frac 12 x^2}$ is the ground state of $\Htot$. For this, we have $\kernel{T_0}(q,q') = \pi^{-N/2}e^{-\frac 12 (q^2+(q')^2)}$, and
\begin{align*}
\widehat{T_0}(q,p) &= e^{-\frac 14(q^2+p^2)}, &\wf{T_0}(q,p)&= 2^Ne^{-(q^2+p^2)}.
\end{align*}
Note that even though all three functions are Gaussian, the constant in the exponent is different in each case.
\subsection{Operations on Schwartz operators}

\label{sec:oper-schw-oper}

We now describe operator analogues of basic operations on Schwartz functions.

\subsubsection{Multiplication}

We first look at the multiplication of functions. First recall that a function $g:\phaseS\to \Cx$ is \emph{polynomially bounded}, if there exist $m\in \Nl$ and $C>0$ such that
$$
|g(q,p)|\leq C(1+\sum_{i=1}^N (q_i^2+p_i^2))^m.
$$
The set of functions which, together with their derivatives, are polynomially bounded, is denoted by $O_M(\phaseS)$. A differentiable function $g:\phaseS\to \Cx$ defines a continuous map $f\mapsto gf$ on $\sfu$ if and only if $g$ and all its derivatives $D^\alpha g$ are polynomially bounded (see e.g. \cite{RSI}).

It is now easy to formulate an analogous condition for operators: we say that a densely defined operator $A$ is \emph{polynomially bounded from the right (resp. left)} if $A$ (resp. $A^*$) is relatively bounded with respect to $\Htot^m$ for some $m\in \Nl$. If both hold, we simply say that $A$ is \emph{polynomially bounded}. Recall \cite{RSII} that a densely defined operator $A$ is \emph{relatively bounded} with respect to an operator $H_0$ (or \emph{$H_0$-bounded}, for short) if $\dom(H_0)\subset \dom(A)$, and there exist positive constants $a,b>0$ such that
$$
\|A\varphi\|\leq a\|H_0\varphi\|+b\|\varphi\|, \qquad \text{for each }\varphi\in \dom(H_0).
$$
If $H_0$ is positive and selfadjoint (as is the case with each $\Htot^m$), the resolvent $(1+H_0)^{-1}$ is bounded, and maps the whole Hilbert space bijectively onto $\dom(H_0)$. From this it is easy to see that $A$ is $H_0$-bounded if and only if $\dom(H_0)\subset \dom(A)$ and $A(1+H_0)^{-1}$ is bounded. If $A$ is closed, it follows from Lemma \ref{clg} that $A$ is $H_0$-bounded iff $\dom(H_0)\subset \dom(A)$.

We make the following observation:
\begin{cor}\label{cyclingcor} The "cycling under the trace" formula (Cor. \ref{tracecycling}) holds also for every polynomially bounded operator $A$ (regardless of whether $A$ is closed).
\end{cor}
\begin{proof} Since $\sqrt{T}$ is a Schwartz operator, it maps everything into $\sfunc N$, which is contained in $\dom(\Htot^m)$ (for all $m$) and hence also in $\dom(A)\cap\dom(A^*)$. Thus $A\sqrt T$ is everywhere defined, and bounded because $A\sqrt T= (A(1+\Htot^m)^{-1})(1+\Htot^m)\sqrt T$. Similarly, $A^*\sqrt T$ is bounded, and hence its adjoint equals the closure of $\sqrt TA$, which is therefore also bounded. Hence we can use the same argument as in the proof of Cor. \ref{tracecycling}.
\end{proof}

The set of operators $A$ such that each $Q^\alpha P^\beta A$ is polynomially bounded from the right, is denoted by $O_{MR}(\HH)$, and the set of operators $A$ such that $AP^\beta Q^\alpha$ is polynomially bounded from the left is denoted by $O_{ML}(\HH)$. We also define $O_{M}(\HH)=O_{ML}(\HH)\cap O_{MR}(\HH)$. Clearly the adjoint operation is a bijective map between $O_{MR}(\HH)$ and $O_{ML}(\HH)$.

\begin{prop}\label{ppalgebras} Every element $A\in O_{MR}(\HH)$ maps $\sfunc N$ into itself. If $A\in O_{MR}(\HH)$ and $B\supset A$, then $B\in O_M(\HH)$. The sets $O_{MR}(\HH)$, $O_{ML}(\HH)$, and $O_{M}(\HH)$ are algebras with respect to the usual addition and multiplication of unbounded operators.
\end{prop}
\begin{proof} First note that for $A\in O_{MR}(\HH)$ and $\alpha,\beta$, there exists an $m$ such that $\sfunc N\subset \dom(\Htot^m)\subset\dom(Q^\alpha P^\beta A)$. This implies that the domain of $A$ contains $\sfunc N$ as an invariant subspace, and also shows that any extension of $A$ is in $O_{MR}(\HH)$. This proves the first two claims. Consequently, the product $AB$ of two polynomially bounded operators $A$ and $B$ is densely defined, with domain containing $\sfunc N$. For fixed $\alpha,\beta$, we now choose $m$ as above, and such that
$Q^\alpha P^\alpha A (1+\Htot^m)^{-1}$ is bounded. Then we choose $m'$ such that $\dom(\Htot^{m'})\subset \dom(B)$ and $(1+\Htot^m)B (1+\Htot^{m'})^{-1}$ is bounded; in particular, $B$ maps $\dom(\Htot^{m'})$ into $\dom(\Htot^m)$. Hence $\dom(\Htot^{m'})\subset \dom (Q^\alpha P^\beta AB)$, and
$$
Q^\alpha P^\beta AB (1+\Htot^{m'})^{-1}=Q^\alpha P^\beta A(1+\Htot^{m})^{-1}(1+\Htot^{m})B (1+\Htot^{m'})^{-1}
$$
is bounded. Hence, $O_{MR}(\HH)$ is closed under multiplication. Since $\dom(A+B)=\dom (A)\cap \dom (B)$ by definition, it is clear that $O_{MR}(\HH)$ also contains all linear combinations of its elements. This proves that $O_{MR}(\HH)$ is an algebra. Since $B^*A^*\subset (AB)^*$ for any two unbounded operators, it follows that $O_{ML}(\HH)$, and consequently also $O_{M}(\HH)$, is an algebra.
\end{proof}

\begin{remark}\label{remarkpoly} \rm Clearly, every bounded operator is polynomially bounded. More importantly, so is each polynomial of $Q$ and $P$. Indeed, by noting that $\dom(\Htot^m)$ is given explicitly as
\begin{equation}\label{Hdom}
\dom(\Htot^m) := \left\{\varphi\in \HH \mid \sum_{\alpha'\in \indm N} \left(\sum_{i=1}^N\alpha'_i\right)^m |\langle \alpha |\varphi\rangle|^2<\infty\right\},
\end{equation}
and writing each $Q_i$ and $P_i$ in terms of the ladder operators $A_i, A_i^*$ we see that for given $\alpha, \beta\in \indm N$ there exists an $m\in \Nl$ such that $\dom(\Htot^m)\subset \dom(Q^\alpha P^\beta)$. Hence, $Q^\alpha P^\beta (1+\Htot^m)^{-1}$ is bounded by Lemma \ref{clg}. More generally, each polynomial of closed operators with domain containing $\dom(\Htot^m)$ for some $m$, is polynomially bounded from the right.
\end{remark}

\begin{remark}\rm Note that even though every bounded operator is polynomially bounded, $\mathcal{B(H)}$ is \emph{not} included in any of the sets $O_{MR}(\HH)$, $O_{ML}(\HH)$ and $O_M(\HH)$, because e.g $A\in O_{MR}(\HH)$ requires $\sfunc N\subset \dom(Q^\alpha P_\alpha A)$. A rank one operator $|\varphi\rangle\langle \psi|$ with $\psi\in \sfunc N$ and $\varphi\notin \dom(Q)$, is in $O_{MR}(\HH)$ but not in $O_{ML}(\HH)$.
\end{remark}

Concerning multiplication, we now have the following:
\begin{prop}\label{polymultcont} If $A\in O_{MR}(\HH)$ (resp. $A\in O_{ML}(\HH)$), the multiplication
$T\mapsto AT$ (resp. $T\mapsto TA$) is a continuous map from $\sop$ into itself.
\end{prop}
\begin{proof} Let $A\in O_{MR}(\HH)$. Given $\alpha,\alpha',\beta,\beta'\in \indm N$ we can find $m\in \Nl$ such that $\dom(\Htot^m)\subset \dom(Q^\alpha P^\beta A)$, and $Q^\alpha P^\beta A(1+\Htot^m)^{-1}$ is bounded. But this implies that
$$\| Q^\alpha P^\beta ATP^{\beta'}Q^{\alpha'}\| \leq \|Q^\alpha P^\beta A(1+\Htot^m)^{-1}\| \|(1+\Htot^m)TP^{\beta'}Q^{\alpha'}\|,$$
showing that $T\mapsto AT$ is continuous. The claim concerning $A\in O_{ML}(\HH)$ and $T\mapsto TA$ is proved by taking the adjoint, which is continuous on $\sop$ by Lemma \ref{basiclem} (a).
\end{proof}

\subsubsection{Differentiation}

Concerning differentiation, it is again useful to first look at the function analogue. For $f\in \sfu$ let $f_y$ denote the translation of $f$, i.e. $f_y(x) = f(x-y)$. For a fixed $x$ we have
$$
D_y^\alpha f_y(x)|_{y=0} = D_y^\alpha f(x-y)|_{y=0} = (-1)^{|\alpha|} (D^\alpha f)(x).
$$
Hence, we can write each derivative $D^\alpha f\in \sfu$ as the derivative at $y=0$ of the translation $y\mapsto f_y$, in the weak (pointwise) sense. Since we know that the translations of operators in $\sop$ are represented by $y\mapsto W(y)TW(-y)$, we can use this connection to \emph{define} a derivative for Schwartz operators:

Given $\alpha\in \indm{2N}$ and $T\in \sop$ we define
\begin{align}\label{derivative}
&D^\alpha T\in \sop, & \langle \varphi|D^\alpha T\psi\rangle &= (-1)^{|\alpha |} D^\alpha \langle \varphi|W(y)TW(y-)\psi\rangle|_{y=0}.
\end{align}
Using the definition of the Weyl operators, we can express this explicitly as a polynomial of $Q$ and $P$, which shows that the result is indeed a Schwartz operator.

Another natural way of formulating differentiation is by means of the \emph{commutator (Lie) derivative}: for an arbitrary operator $A\in O_{M}(\HH)$, we define
\begin{equation}\label{liederiv}
\mathcal L_{A}(T)=[A,T].
\end{equation}
We now get the following result:

\begin{prop}\label{dercont}
For each $\alpha\in \indm{2N}$, the derivative $T\mapsto D^\alpha T$ is a continuous map from $\sop$ into itself. For each $A\in O_M(\HH)$, the Lie derivative $\mathcal L_A$ is a continuous map from $\sop$ to itself. The following identify holds.
\begin{equation}\label{liederexpansion}
D^{\alpha\vee \beta}T=(-i)^{|\alpha|}i^{|\beta|}\mathcal L_{P_N}^{\alpha_n}\circ\cdots \circ \mathcal L_{P_1}^{\alpha_1}\circ\mathcal L_{Q_N}^{\beta_N}\circ\cdots\circ\mathcal L_{Q_1}^{\beta_1}(T).
\end{equation}
\end{prop}
\begin{proof} The map $\mathcal L_A$ is continuous due to Prop. \ref{polymultcont}, and the expansion of $D^{\alpha\vee\beta}T$ can be verified by direct computation, using the fact that $W(q,p)=e^{iq\cdot p/2}e^{-iq\cdot P}e^{ip\cdot Q}$, and noting that the phase factor $e^{ip\cdot q/2}$ does not contribute to the derivative. The expansion also shows the continuity of $T\mapsto D^{\alpha\vee\beta} T$.
\end{proof}
Furthermore, the derivative combines naturally with convolutions; the following identities follow immediately from the Weyl relations:
\begin{align}
\widehat{D^\alpha f}&=(D_y^\alpha e^{-i\{x,y\}}|_{y=0})\widehat f, \quad f\in \sfu, \label{derc1}\\
\widehat{D^{\alpha}T}&=(D_y^\alpha e^{i\{x,y\}}|_{y=0})\,\widehat T, \quad T\in \sop.\label{derc2}\\
D^\alpha (S*T)&=D^\alpha S * T=S*D^\alpha T, \quad T,S\in \sop,\label{derc3}\\
D^\alpha (f*T)&=D^\alpha f * T=f*D^\alpha T, \quad f\in \sfu,\,S\in \sop. \label{derc4}
\end{align}

\section{Application 1: Operator moment problems}
\label{sec:oper-moment-probl}

In order to demonstrate usefulness of the above development, we consider the operator version of \emph{moment problems}, which is a classic topic in measure theory. For instance, in the Hamburger
moment problem \cite[X.1]{RSII} we are asking for  conditions on  a sequence of
real numbers $m_n$, $n \in \mathbb{N}$ under which a measure $\mu$ on $\mathbb{R}$
exists, such that the $m_n$ become the moments of $\mu$, i.e.
\begin{displaymath}
  m_n = \int_{-\infty}^\infty x^n \mu(dx) ,
\end{displaymath}
and whether $\mu$ is uniquely determined by the $m_n$. As a non-commutative
analog we now replace $\mu$ by a positive trace-class operator $\myrho \in
\mathcal{B}_*(\mathcal{H})$ on a Hilbert space $\mathcal{H}$ and for a set of
(in general unbounded) operators $X_1, \dots X_k$ we look at the expectation values
\begin{displaymath}
  m_f = \tr (f(X_1, ..., X_k) \myrho),
\end{displaymath}
where $f$ runs over all polynomials in the $X_j$. Of course this equation is
not well defined for all $\myrho$ and we have to adjust our definition to cope
with possible domain problems. For the case where the $X_j$ are just the
canonical position and momentum operators $Q_1, \dots, Q_n; P_1, \dots, P_n$,
however, we have done exactly that in the preceding section. Indeed, Schwartz operators are exactly those trace class operators for which we can formulate the following problem: For each pair of
multiindices $\alpha, \beta \in \indm{N}$ define the numbers
\begin{equation}\label{eq:1}
  m_{\alpha,\beta} = \tr(Q^\alpha P^\beta \myrho),
\end{equation}
which we will call henceforth the \emph{moments of $\myrho$}. Is $T$ uniquely determined by its moments?

For later use we
can introduce the operators $R_1, \dots, R_{2n} = Q_1, \dots, P_N$ and the
multi index notation
\begin{displaymath}
  R^A = R_{a_1} \dots R_{a_{|A|}}, \quad A = (a_1, \dots, a_{2n})   \in \{1,
  \dots, 2n\}^{|A|}
\end{displaymath}
which is different from the one used earlier (for that reason we are using
another family symbols for the indices). Direct sums of multi-indices can also
be defined
\begin{displaymath}
  A \vee B = (a_1, \dots, a_{|A|}, b_1, \dots b_{|B|}) \quad
  \text{obviously}\ |A \vee B| = |A| + |B|,
\end{displaymath}
as well as a conjugation
\begin{displaymath}
  \overline{A} = (a_{2n}, \dots, a_1) \quad \text{if} \quad A = (a_1, \dots, a_{2n}).
\end{displaymath}
The purpose of the latter is given by the equation $(R^A)^* \phi =
R^{\overline{A}} \phi$ for all $\phi \in \sop$. With this notation we can define
\begin{displaymath}
  m_A = \tr(R^A \myrho),\quad A \in F
\end{displaymath}
where $F$ is the set of all multiindices $A$ with arbitrary length
$|A|$. Obvioulsy the $m_{\alpha,\beta}$ form a subfamily of the
$m_A$. The converse is not true, since the $R_k$ can appear in any
order. Howerver, by using the canonical commutation relations we can express
each $m_A$ as a linear combination of some of the
$m_{\alpha,\beta}$. Hence both sets of moments contain exactly the same information.

The main result of this section is the uniqueness theorem  (Thm. \ref{thm:1})
which states that under a technical condition (analyticity;
cf. Def. \ref{def:1}) each Schwartz operator is uniquely determined by its
moments. In order to prove this we need an additional tool, which is discussed
in the next subsection.

Before we come to this let us add some short remarks about two related
topics. Firstly, the existence question for the moments in
Eq. (\ref{eq:1}). As in the classical case a positivity condition (which can be
easily formulated in terms of the associative *-algebra generated by the $Q_i$
and $P_i$) is sufficient, but (and this different from the commutative case)
not sufficient. To fill this gap we need representation theory of the
Heisenberg Lie algebra \cite{Folland}. Our second remark concerns the Pauli
problem, i.e. the question whether the distribution for position and momentum
\emph{together} are sufficient to determine the corresponding density
matrix. The answer is known to be: no; cf. e.g. \cite{StuSing}. This implies
in particular the moments $m_{\alpha,0}$ and $m_ {0, \beta}$ are not sufficient
to determine the density operator, while (at least for analytic Schwartz
operators) uniqueness can be guaranteed if all of the $m_{\alpha,\beta}$ are
known (cf. Thm. \ref{thm:1} below).

\subsection{Purifications}

A purification of a positive trace class operator $\myrho$ on a Hilbert space $\mathcal{H}$
is a pair $(\mathcal{K}, \Omega)$ consisting of another Hilbert space
$\mathcal{K}$ and a vector $\Omega \in \mathcal{H} \otimes \mathcal{K}$
satisfying $\myrho = \tr_{\mathcal{K}}( \kettbra{\Omega} )$,
where $\tr_{\mathcal{K}}$ denotes the partial trace over the second tensor
factor. A purification is called minimal if $\Omega$ is cyclic for the von
Neumann algebra $\mathcal{B}(\mathcal{H}) \otimes \idty$. Existence and
uniqueness of the GNS construction implies immediately that each $\myrho$ admits
a unique (up to unitary equivalence) purification.

Our goal is now to look at purifications of positive Schwartz operators. To
formulate and prove the main theorem we need the following definition.

\begin{defn} \label{def:1}
  A Schwartz operator $\myrho$ is called \emph{analytic} if there exist constants $C, K
  > 0$ such that its moments satisfy
  \begin{displaymath}
    |m_{\overline{A} \vee A}| \leq C^2 K^{2 |A|} (|A|!)^2
  \end{displaymath}
  for all $A \in \{1, \dots, 2n\}^{|A|}$.
\end{defn}

We can reformulate Definition \ref{def:1} in terms of the square root of
$\myrho$ (which is again in $\sop$ according to Prop. \ref{prop:1}). Note that the
following result also shows that $m_{\overline{A} \vee A}$ is always
positive. Hence the modulus in Def. \ref{def:1} is redundant.

\begin{prop} \label{def:1a}
  A positive Schwartz operator $\myrho \in \sop$ is analytic iff
  there are constants $C, K > 0$ such that
  \begin{displaymath}
    \left\| R^A \sqrt{\myrho} \right\|_2 \leq C K^{|A|} |A|!
  \end{displaymath}
  holds for all $A \in \{1,\dots,2n\}^{|A|}$ of arbitrary length $|A|$.
\end{prop}

\begin{proof}
  The statement immediately follows from
  \begin{displaymath}
    \left\| R^A \sqrt{\myrho} \right\|_2^2 = \tr\left( (R^A \sqrt{\myrho})^*
    R^A \sqrt{\myrho} \right) = \tr\left( \sqrt{\myrho} R^{\overline{A}} R^A
      \sqrt{\myrho} \right) = \tr( R^{\overline{A} \vee A} \myrho) =
    m_{\overline{A}\vee A}.
  \end{displaymath}
\end{proof}

Finally, we come to the main result of this subsection. The main trick is to
use the kernel of $\sqrt{\myrho}$ as the purification of $\myrho$.

\begin{prop} \label{prop:2}
  Consider the minimal purification $(\mathcal{K}, \Omega)$ of a positive,
  analytic Schwartz operator $\myrho \in \sop$. The linear hull of $\{ R^A
  \otimes \idty \Omega \, | \, A \in F \}$ is dense in
  $\mathcal{H} \otimes \mathcal{K}$ and all its elements are analytic vectors
  for the canonical operators $R_1, \dots, R_{2n}$.
\end{prop}

\begin{proof}
  We start with the kernel function $\Omega \in \mathcal{H} \otimes
  \mathcal{H}$ of $\sqrt{\myrho}$. Since $A \otimes \idty \Omega$ is for each
  $A \in \mathcal{B}(\mathcal{H})$ the kernel function of $A \sqrt{\myrho}$ we
  have:
  \begin{align*}
    \langle \Omega, A^*A \otimes \idty \Omega \rangle &= \langle A \otimes
    \idty \Omega, A \otimes \idty \Omega\rangle \\
    &= \tr\left( (A \sqrt{\myrho})^* A \sqrt{\myrho}\right) = \tr\left(\sqrt{\myrho}
      A^* A \sqrt{\myrho}\right)\\
    &= \tr(\myrho A^*A).
  \end{align*}
  Hence for any positive operator we have
  \begin{equation}\label{eq:7}
    \tr(\myrho A) = \langle \Omega, A \otimes \idty \Omega \rangle.
  \end{equation}
  Since each bounded operator can be written as a linear combination of four
  positive operators Eq. (\ref{eq:7}) holds for any $A$. Hence $\Omega$ is a
  purification of $\myrho$. To get a minimal purification consider the Schmidt
  decomposition of $\Omega$:
  \begin{equation}\label{eq:8}
    \Omega = \sum_n \lambda_n \phi_n \otimes \psi_n,\quad \lambda_n > 0
  \end{equation}
  with orthonormal systems $\phi_n \in \mathcal{H}$ and $\psi_n
  \mathcal{H}$. If necessary we extend $\phi_n$ to a complete orthonormal
  system (this might require a renumbering if $n$ already runs over $n \in
  \mathbb{N}$). Now define $\mathcal{K}$ as the closed subspace of $\mathcal{H}$
  generated by the $\psi_n$. Obviously $\Omega \in \mathcal{H} \otimes
  \mathcal{K}$. Hence we only have to show that $\Omega$ is cyclic. To this
  end consider operators $\ketbra{\phi_m}{\phi_n}$ for arbitrary $m,
  n$. Applying $\ketbra{\phi_m}{\phi_n} \otimes \idty$ to $\Omega$ we get
  according to (\ref{eq:8}) $\lambda_n \phi_m \otimes \psi_n$. Since
  $\lambda_n \neq 0$ we can generate all elements of the basis $\phi_m \otimes
  \psi_n$ of $\mathcal{H} \otimes \mathcal{K}$ that way. In other words
  $\Omega$ is cyclic and the pair $(\mathcal{K}, \Omega)$ is the minimal
  purification of $\myrho$.

  In our next step we look at the vectors $R^A \Omega \in \mathcal{H}
  \otimes \mathcal{K}$ with $A \in \{1, \dots, 2n\}^{|A|}$. we have to show
  that they are analytic for the family of operators $R_1, \dots, R_{2n}$, i.e.
  \begin{equation}\label{eq:9}
    \| (R \otimes \idty)^B \Omega \| \leq C K^{|B|} (|B|)!
  \end{equation}
  with constants $C, K \in \mathbb{R}^+$ and for all $B \in \{1, \dots,
  2n\}^{|B|}$ of arbitrary length. Since $\myrho$ is analytic we get from
  Definition \ref{def:1}
  \begin{equation} \label{eq:10}
    \| (R \otimes \idty)^{A \vee B} \Omega \| = \left\| R^{A
        \vee B} \sqrt{\myrho} \right\|_2 \leq C K^{|A|+|B|} (|A| + |B|)!
  \end{equation}
  since $\sqrt{\myrho}$ is a Schwartz operator, hence, $\Omega$ a Schwartz
  function and $R^{A \vee B} \otimes \idty \Omega$ is the kernel of $R^{A
    \vee B} \sqrt{\myrho}$. To get the form given in Eq. (\ref{eq:9}) let us
  rewrite $(|A| + |B|)!$ as
  \begin{displaymath}
    (|A| + |B|)! = { |A| + |B| \choose |B| } |A|! |B|! \, .
  \end{displaymath}
  The binomial coefficient can be estimated from above in terms of the binomial
  expansions of $2^{|A| + |B|} = (1+1)^{|A|+|B|}$:
  \begin{displaymath}
    { |A| + |B| \choose |B| } \leq 2^{|A| + |B|}.
  \end{displaymath}
  Hence, with $\tilde{C} = C (2K)^{|A|} |A|!$ and $\tilde{K} = 2K$ we get
  (\ref{eq:10}):
  \begin{displaymath}
    \| (R \otimes \idty)^B (R\otimes\idty)^A \Omega \| \leq \tilde{C}
    \tilde{K}^{|B|} B!\, ,
  \end{displaymath}
  which shows that $(R\otimes\idty)^A \Omega$ is analytic.

  It remains to show that the subspace $D \subset \mathcal{H} \otimes
  \mathcal{K}$ generated by the family $(R\otimes\idty)^A \Omega$ is dense
  in $\mathcal{H} \otimes \mathcal{K}$. To this end we use Nelsons results on
  analytic vectors (\cite{Nelson}, cf. also   \cite[X.6]{RSII} and
  \cite[Ch. 4.3]{Folland}) which imply that for each $\phi \in D$ and each
  Weyl operator $W(x)$ the vector $ W(x)\phi$ can be written
  as a (norm-convergent) series involving terms of the form $R^A\phi x^A
  (|A|!)^{-1}$, with $x^A = x_{a_1} \cdots x_{a_{|A|}}$. Hence  $W(x) \phi \in
  \overline{D}$ and therefore
  \begin{displaymath}
    D_1 = \SP \{ W(x) \otimes \idty \Omega\, | \, x \in \mathbb{R}^{2n} \} \subset
    \overline{D}.
  \end{displaymath}
  Now recall that $\Omega$ is cyclic for the algebra $\mathcal{B}(\mathcal{H})
  \otimes \idty$, hence the space $\{A \otimes \idty \Omega   \, | \,  A \in
  \mathcal{B}(\mathcal{H}) \}$ is dense in $\mathcal{H} \otimes
  \mathcal{K}$. Moreover, finite linear combination of Weyl operators are
  norm-dense in $\mathcal{B}(\mathcal{H})$. Therefore $D_1$ is dense in
  $\mathcal{H} \otimes \idty$, too, and $D_1 \subset \overline{D}$ implies
  $\overline{D} = \mathcal{H} \otimes \mathcal{K}$, which was to show.
\end{proof}

\subsection{Uniqueness}

We are now ready to prove uniqueness of the moment problem in the following
form:

\begin{thm} \label{thm:1}
  Consider two positive, analytic Schwartz operators $\myrho_1, \myrho_2 \in
  \sop$ such that
  \begin{equation} \label{eq:2}
    \tr( Q^\alpha P^\beta \myrho_1) = m_{\alpha\beta} = \tr(Q^\alpha P^\beta
    \myrho_2) \quad \forall \alpha,\beta \in \indm{N}.
  \end{equation}
  Then we have $\myrho_1 = \myrho_2$.
\end{thm}

\begin{proof}
  As already stated at the beginning of this section we can use the moments
  $m_A$ instead of $m_{\alpha\beta}$. In other words Eq. (\ref{eq:2}) is
  equivalent to
  \begin{equation} \label{eq:11}
    \tr(R^A \myrho_1) = \tr(R^A \myrho_2) \quad \forall A \in F,
  \end{equation}
  which is the relation we will use in the following. Now consider the minimal
  purifications $(\mathcal{K}_j, \Omega_j)$ of $\myrho_j$, $j=1,2$. If $f$ is a
  polynomial in the $R_k \otimes \idty$ satisfying $f \Omega_1 = 0$ we get
  \begin{displaymath}
    0 = \langle R^A \otimes \idty \Omega_1, f \Omega_1 \rangle = \langle
    \Omega_1, R^{\overline{A}} \otimes \idty f \Omega_1, \Omega_1\rangle =
    \tr (R^{\overline{A}} f \myrho_1)
  \end{displaymath}
  The operator $R^{\overline{A}} f$ is a linear combination of terms $R^B$ for
  some $B \in F$. Hence, using Eq. (\ref{eq:11}) we get
  \begin{displaymath}
    \langle R^A \otimes \idty \Omega_2, f \Omega_2 \rangle = 0,
  \end{displaymath}
  and since $A$ is arbitrary we get $f \Omega_2 = 0$ due to cyclicity of
  $\Omega_2$. Hence we can define a map $U : D_1 \rightarrow D_2$, $D_j = \SP \{R^A
  \otimes \idty\Omega_j\, | \, A \in F \}$ by
  \begin{equation} \label{eq:12}
    U : D_1 \rightarrow D_2, \quad D_j = \SP \{R^A \otimes \idty\Omega_j\, |
    \, A \in F \} \quad U R^A \Omega_1 = R^A \Omega_2.
  \end{equation}
  Reversing the roles of $\myrho_1$ and $\myrho_2$ we see that $U$ is
  invertible. Furthermore, using again Eq. (\ref{eq:11}) we get
  \begin{displaymath}
    \langle R^A \Omega_1, R^B \Omega_1 \rangle = \langle R^A \Omega_2, R^B
    \Omega_2 \rangle,
  \end{displaymath}
  which shows that $U$ extends to a unitary $\overline{D_1} \rightarrow
  \overline{D_2}$. Cyclicity of the $\Omega_j$ finally shows that
  (\ref{eq:12}) defines a unique unitary $\mathcal{H} \otimes
  \mathcal{K}_1 \rightarrow \mathcal{H} \otimes \mathcal{K}_2$. Furthermore we
  have for $k=1,\dots,2n$ and $A \in F$
  \begin{equation} \label{eq:13}
    U (R_k \otimes \idty) R^A \otimes \idty \Omega_1 =  (R_k \otimes
    \idty) R^A \otimes \idty \Omega_2 = (R_k \otimes
    \idty) U R^A \otimes \idty \Omega_1,
  \end{equation}
  Hence $[U, R_k \otimes \idty] \phi = 0$ for all $\phi \in D_1$. Note that
  this usually does not imply $[U, \exp(i x R_k)] = 0$ for some $x \in
  \mathbb{R}$. According to Prop. \ref{prop:2}, however, the elements in $D_1$
  are analytic vectors such that $W(x) \otimes \idty \phi$ with $\psi \in
  D_1$ and a Weyl operator $W(x)$, $x \in \mathbb{R}^{2n}$ can be written as a
  norm convergent series of terms of the form $R^A \phi x^A
  (|A|!)^{-1}$. Therefore Eq. (\ref{eq:13}) implies that $[U, W(x) \otimes
  \idty] = 0$ holds for all $x \in \mathbb{R}^{2n}$. Due to irreducibility of
  the Weyl operators on $\mathcal{H}$ the unitary $U$ has to have to the form
  $\idty \otimes \tilde{U}$ with a unitary $\tilde{U}  : \mathcal{K}_1
  \rightarrow \mathcal{K}_2$. Hence
  \begin{displaymath}
    \myrho_2 = \tr_{\mathcal{K}_2}(\kettbra{\Omega_2}) = \tr_{\mathcal{K}_2}(
    \idty \otimes \tilde{U} \kettbra{\Omega_1} \idty \otimes \tilde{U}^*) =
    \tr_{\mathcal{K}_1}(\kettbra{\Omega_1}) = \myrho_1
  \end{displaymath}
   what was to show.
\end{proof}

\section{Tempered distributions}
\label{sec:temp-distr}

We now proceed with a natural development of the general theory of Schwartz
operators. Starting again with the function analogue, a \emph{tempered
  distribution on $\sfu$} is an element of the topological dual $\sfud$,
i.e. a continuous linear functional $\phi:\sfu\to \Cx$. Similarly, we say that
a \emph{tempered distribution on $\sop$} is a continuous linear functional
$\Phi:\sop\to \Cx$. The space of tempered distributions is the topological
dual of $\sop$, and will be denoted by $\sopd$. It is equipped with the
corresponding weak-* topology. In the following we will discuss some of its
properties, including in particular:
\begin{itemize}
\item
  \textbf{Examples of operators as elements of $\sopd$}, more specifically bounded and polynomially bounded
  operators, are discussed in Subsect. \ref{sec:oper-as-distr}.
\item
  \textbf{Alternative characterizations of elements of $\sopd$} will be formulated
  in terms of their kernel distributions (which are ordinary tempered distributions), and matrix representations
  (Subsect. \ref{sec:home-matr-repr}).
\item
  \textbf{Operations on distributions} are studied in
  Subsect. \ref{sec:oper-distr}. Together with the previous point this allows
  us to do obtain new quadratic forms from existing ones in ways usually not allowed, such as:
  products with (polynomially) bounded operators, differentiation, convolutions
  and Fourier transforms.
  \item \textbf{Weyl-Wigner correspondence} between usual tempered distributions $\sfud$ (generalised classical variables) and the elements of $\sopd$ (generalised operators) will be formulated using the harmonic analysis operations. Some interesting special cases will be pointed out, for instance, quantisation of the delta-distribution and its derivatives.
  \item
  \textbf{Regularity theorem.} In Sect. \ref{sec:regularity-theorem} we will prove a quantum version of the regularity theorem for distributions, showing that \emph{any} element of $\sopd$ can be written as a polynomially bounded quadratic form (not necessarily an operator), which can be obtained from polynomially bounded operators by differentiation and taking linear combinations.

\item
  \textbf{Approximation of distributions} in terms of operators are briefly
  discussed in Subsect. \ref{sec:conv-distr}.
\end{itemize}

\subsection{Operators as distributions}
\label{sec:oper-as-distr}

Just as sufficiently well-behaved functions $g:\phaseS \to \Cx$ define elements of $\sfud$ via the integral formula
\begin{equation}\label{integraldual}
\phi_g(f) = \int g(x) f(x) dx,
\end{equation}
we look for operators $A$ which define elements of $\sopd$ via the trace formula
\begin{equation}\label{tracedual}
\Phi_A(T) = {\rm tr}[AT]={\rm tr}[TA], \quad T\in \sop.
\end{equation}

\subsubsection{Bounded operators}

The most well-behaved operators are elements $S\in \sop$; in fact, $|{\rm tr}[ST]|\leq \|S\|_1\|T\|_{0,0,0,0}$, so $S$ is indeed an element of $\sopd$. Since $\sop$ is norm dense in the Hilbert-Schmidt class (Lemma \ref{schwartzdensity}), we have indeed an embedding (injectivity). Since $S\mapsto \|S\|_1$ is a continuous seminorm on $\sop$, it follows that the embedding is continuous when $\sopd$ is equipped with the weak-* topology.

Similarly, since $\sop\subset \schtn p\subset \bh$ for each $p\in [1,\infty)$, with $\|T\|_p\leq \|T\|_1$ being continuous, it follows that each class $\schtn q$, as well as $\bh$, can be norm-continuously injected into $\sopd$ via \eqref{tracedual}.

\subsubsection{Polynomially bounded operators}
\label{sec:polyn-bound-oper}

More interesting elements of $\sopd$ correspond to unbounded operators $A$. We  first look at the function analogue. If $g$ is polynomially bounded (see the definition above), then the integral in \eqref{integraldual} is well defined, and we have
$$
|\varphi_g (f)| \leq C\int \frac{dqdp}{(2\pi)^N} \left(1+\sum_{i=1}^N (q_i^2+p_i^2)\right)^m |f(q,p)|.
$$
Since the right hand side a continuous seminorm, have $\varphi_g\in \sfud$, i.e. polynomially bounded functions are elements of $\sfud$.

Suppose now that $A$ is an operator that is polynomially bounded (from the right and left), and let $m$ be the associated degree. It follows from the "cycling under the trace" result (Lemma \ref{cyclingcor}) that ${\rm tr}[AT]={\rm tr}[TA]$, hence we can indeed define $\Phi_A$ via \eqref{tracedual} without having to pay attention to the order in which the operators appear inside the trace. Now $AT=A(1+\Htot^m)^{-1}(1+\Htot^m)T$, and we have
$$
|\Phi_A(T)| = | {\rm tr}[AT] | \leq \|A(1+\Htot^m)^{-1}\| \,\|(1+\Htot^m)T\|_1,
$$
which shows that $\Phi_A\in \sopd$. In the case where $A$ is only polynomially bounded from the right, the formula $\Phi_A(T)={\rm tr}[AT]$ still defines an element of $\sopd$, but $\Phi_A(T)={\rm tr}[TA]$ is not guaranteed; it is the appropriate definition in the case where $A$ is polynomially bounded from the left.

\subsection{Homeomorphisms and matrix representation}
\label{sec:home-matr-repr}

From the homeomorphism theorem Prop. \ref{isomorphisms} we immediately obtain the following result:
\begin{prop} \label{prop:4}
$\sopd$ is homeomorphic to $\sw'(\mathbb R^{2N})$, and $\mathfrak{s}'_{2N}$, when they are equipped with the weak-* topology. The homeomorphism can be realised by transposing the homeomorphisms of Prop. \ref{isomorphisms}.
\end{prop}
It will be useful to do the sequence space formulation explicitly: we define the \emph{matrix elements} of a $\Phi\in \sopd$ by
$$
\Phi_{\alpha\vee\alpha'}=\Phi(|\alpha\rangle\langle \alpha'|).
$$
Then we know from the homeomorphism theorem Prop. \ref{isomorphisms} that $\Phi_{\alpha\vee\alpha'}=\phi_{\alpha\vee\alpha'}$, where $\phi_{\alpha\vee\alpha'}:=\phi(|\alpha\vee\alpha'\rangle)$ are the Hermite coefficients of the tempered distribution $\phi\in \sw'(\mathbb R^{2N})$ given by the kernel isomorphism: $\Phi(S)=\phi(K^S)$. Hence we immediately get the following version of the N-representation theorem \cite[Theorem V.14]{RSI} for $\sopd$:
\begin{prop}\label{Nrepd} For a given $\Phi\in \sopd$, there exists $\beta\vee\beta'\in I_{2N}$, and a constant $C>0$ such that
$$|\Phi_{\alpha\vee\alpha'}|\leq C(\alpha\vee\alpha'+1)^{\beta\vee\beta'}, \text{ for all }\alpha\vee\alpha'\in I_{2N}.
$$
Conversely, given coefficients $|a_{\alpha\vee\alpha'}|\leq C(\alpha\vee\alpha'+1)^{\beta\vee\beta'}$, there exists a unique element $\Phi\in \sopd$ such that $\Phi_{\alpha\vee\alpha'}=a_{\alpha\vee\alpha'}$. The matrix representation
$$
\Phi=\sum_{\alpha\vee\alpha'} \Phi_{\alpha\vee\alpha'} |\alpha\rangle\langle\alpha'|
$$
converges in the weak-* topology. In particular, $\sop$ is weak-* dense in $\sopd$.
\end{prop}

\subsection{Operations on distributions}
\label{sec:oper-distr}

New elements of $\sfud$ can be conveniently generated from existing ones by standard operations \cite{RSII}; we now look at the analogous ones for $\sopd$.

\subsubsection{Multiplication and differentiation}
A product of $g\in O_M(\phaseS)$ and $\phi\in \sfud$ is defined via
\begin{align}\label{polymultd}
&g\phi\in \sfud, & (g\phi)(f) &=\phi(gf).
\end{align}
Here $g$ must be in $O_M(\phaseS)$ (i.e. all derivatives polynomially bounded); indeed, this is forced by the fact that $f\mapsto gf$ must be continuous in $\sfu$ to make $g\phi$ a distribution. In an analogous fashion, we can \emph{multiply} an element $\Phi\in \sfud$ either from left or right, with an operator $A\in O_{MR}(\HH)$ or $B\in O_{ML}(\HH)$, respectively:
\begin{align*}
&A\Phi\in \sfud, & (A\Phi)(T) &=\Phi(AT),\\
&\Phi B\in \sfud, & (\Phi B)(T) &=\Phi(TB).
\end{align*}
It follows from Prop. \ref{polymultcont} that these indeed define elements of $\sopd$. As an important example, note that each polynomial of $Q$ and $P$ is in $O_{M}(\HH)$ by Remark \ref{remarkpoly}, hence they define distributions in this sense.

Since the derivative of a $\phi\in \sfud$ is given by
\begin{align}
&D^\alpha \phi \in \sfud, & (D^\alpha \phi)(f) &=(-1)^{|\alpha|}\phi(D^\alpha f),
\end{align}
we define the derivative of $\Phi\in \sopd$ via
\begin{align}
&D^\alpha \Phi \in \sopd, & (D^\alpha \Phi)(T) &= (-1)^{|\alpha|}\Phi(D^\alpha T),
\end{align}
where $D^\alpha T$ was defined in Section \ref{sec:oper-schw-oper}; in particular, $D^\alpha \Phi\in \sopd$ because of Prop. \ref{dercont}. The commutator derivative is defined by
\begin{align}
&\mathcal L_A\Phi \in \sopd, & (\mathcal L_A\Phi)(T) &= -\Phi(\mathcal L_A(T)),
\end{align}
for any $A\in O_{M}(\HH)$. The following continuity result is a direct consequence of the definitions and the corresponding continuity results for $\sop$.
\begin{prop}\label{mdcont} For fixed $g\in O_M(\phaseS)$, $A\in O_{MR}(\HH)$, and $B\in O_{ML}(\HH)$, and $C\in O_{M}(\HH)$, the maps
\begin{align*}
\sfud \ni \phi&\mapsto g\phi\in \sfud, & \sopd\ni \Phi&\mapsto A\Phi\in \sopd,\\
& & \sopd\ni \Phi&\mapsto\Phi B\in \sopd\\
\sfud \ni \phi & \mapsto D^\alpha\phi, & \sopd \ni \Phi & \mapsto D^\alpha\Phi\\
& & \sopd\ni \Phi&\mapsto \mathcal L_C\Phi \in \sopd\\
\end{align*}
are continuous on $\sfud$ and $\sopd$, respectively. For fixed $\phi\in \sfu$, $\Phi\in \sopd$, the maps
\begin{align*}
\sfu \ni f&\mapsto f\phi\in \sfud, & \sop\ni S&\mapsto S\Phi\in \sopd
\end{align*}
are continuous.
\end{prop}

In order to check that the definitions correctly extends derivatives of operators, we first take $B\in O_M(\HH)$; then
$$
(\mathcal L_A\Phi_B)(T) = -{\rm tr}[B (AT-TA)]={\rm tr}[\mathcal L_A(B) T],
$$
due to Cor. \ref{cyclingcor}, where $\mathcal L_A(B)$ is again a polynomially bounded operator by Prop. \ref{ppalgebras}. Hence $\mathcal L_A\Phi_B=\Phi_{\mathcal L_A(B)}$, as expected. By using the fact that the $e^{-iq P}$ and $e^{ipQ}$ commute up to phase (which does not contribute to the derivative), we see that $D^{\alpha\vee\beta}\Phi_B$ is again a polynomially bounded operator, which can be explicitly computed using the \emph{same} commutator expansion \eqref{liederexpansion} as with Schwartz operators:
\begin{equation}\label{liederexpansionII}
D^{\alpha\vee \beta}\Phi_B=(-i)^{|\alpha|}i^{|\beta|}\mathcal L_{P_N}^{\alpha_n}\circ\cdots \circ \mathcal L_{P_1}^{\alpha_1}\circ\mathcal L_{Q_N}^{\beta_N}\circ\cdots\circ\mathcal L_{Q_1}^{\beta_1}(B).
\end{equation}
Note that the prefactor $(-1)^{|\alpha|}$ in the definition is needed to ensure this.

As an example, we let $N=1$ (multidimensional case is similar). Using the canonical commutation relation $[Q,P]=i\idty$, we immediately obtain the basic derivatives:
\begin{align*}
D^{(1,0)}\Phi_Q&=-\idty, & D^{(1,0)}\Phi_P&=0,\\
D^{(0,1)}\Phi_Q&=0, & D^{(0,1)}\Phi_P&=-\idty,\\
D^{(1,0)}\Phi_{QP}&=-P, & D^{(0,1)}\Phi_{QP}&=-Q.
\end{align*}
The following simple example demonstrates how the distributional derivatives of bounded operators are often no longer operators themselves: suppose that $A$ is a bounded operator whose range lies outside the domain of $Q$ (e.g. a rank one operator $|\varphi\rangle\langle \varphi|$ with $\varphi\notin \dom(Q)$). Then we have
$$
D^{0,1} \Phi_A(T) =i (\mathcal L_Q\Phi_A)(T) =-i{\rm tr}[A \mathcal L_Q(T)]=i{\rm tr}[A(TQ-QT)].
$$
Hence $D^{(0,1)}\Phi_A$ corresponds to the quadratic form
$$\sfunc N\times \sfunc N\ni (\psi,\varphi)\mapsto \langle Q\psi|A\varphi\rangle-\langle \psi|AQ\varphi\rangle\in \Cx,$$
which is not an operator, because $\dom(QA)=\{0\}$. Note that this also demonstrates the situation where the "cycling under the trace" formula does not apply.

\subsubsection{Fourier transforms, Wigner function and Weyl quantization}

We first have to define the \emph{parity transformation} on $\sfu$ and $\sop$, via $f_-(x) = f(-x)$ and $T_-=\Pi T\Pi$, where $\Pi$ is the parity operator. Obviously, the parity transformations $f\mapsto f_-$ and $T\mapsto T_-$ are continuous. It is easy to see that parity commutes with the Fourier-Weyl transform, i.e. $\widehat{(f_-)}=(\widehat{f}\,)_-$, and $(\widehat{T})_-=\widehat{(T_-)}$ for $f\in \sfu$, $T\in \sop$. Hence, we can just use the (slightly ambiguous-looking) symbols $\widehat{f}_-$ and $\widehat{T}_-$.

We now make the following definitions:
\begin{align*}
\widehat{\phi}&\in \sfud, & \widehat{\phi}(f) &= \phi(\widehat{f}_-),\\
\check{\phi}&\in \sopd, & \check{\phi}(T) &= \phi(\widehat{T}_-),\\
\widehat{\Phi}&\in \sfud, & \widehat{\Phi}(f) &= \Phi(\check{f}_-),\\
\wf{\Phi}& =\widehat{\widehat{\Phi}}_-\in \sfud\\
\wq{\phi} &=\check{\widehat{\phi}}_-\in \sopd.
\end{align*}
These are well defined because the Fourier-Weyl transform is a topological isomorphism. Here $\phi\mapsto \widehat{\phi}$ is the \emph{symplectic Fourier transform}, $\phi\mapsto \check{\phi}$ is the \emph{inverse Weyl transform}, $\Phi\mapsto \widehat{\Phi}$ is the \emph{Weyl transform}, $\wf{\Phi}$ is the \emph{Wigner function} of $\Phi$, and $\phi\mapsto \wq{\phi}$ is the \emph{Weyl quantization}. Concerning the latter, it is easy to see that
\begin{align*}
\wf{\Phi}(f) &= \Phi(\wq{f}), & \wq{\phi}(T) &= \phi(\wf{T}).
\end{align*}
The definitions are set such that these transforms on distributions extend the corresponding ones defined for operators and functions. To check this, we first note that
\begin{align*}
\widehat{\overline{f}} &=\overline{\widehat{f}_-}, & \widehat{T^*} &=\overline{\widehat{T}_-}, & \text{for } & f\in \sfu, \, T\in \sop.
\end{align*}
For a given $g\in L^2(\phaseS)$, and $S\in \hs$, we have
\begin{align*}
\widehat{\phi_g}(f)&=\int g(x) \widehat{f}_-(x) dx = \int g(x) \overline{\overline{\widehat{f}_-}}(x) dx= \int g(x) \overline{\widehat{\overline{f}}}(x) dx = \int \widehat{g}(x) \overline{\overline{f}}(x) dx = \phi_{\widehat{g}}(f),\\
\check{\phi_g}(T) &= \int g(x) \overline{\overline{\widehat{T}_-}}(x)dx =
\int g(x) \overline{\widehat{T^*}}(x)dx = {\rm tr}[\check{g}T] = \Phi_{\check g}(T),\\
\widehat{\Phi_S}(f) &= \Phi_S(\check{f}_-) = {\rm tr}[S\check{f}_-] = {\rm tr}[S(\check{f}_-)^{**}]={\rm tr}[S(\check{\overline{f}})^{*}]=\int \widehat{S}(x) \overline{\overline{f}}(x)dx = \phi_{\widehat{S}}(f),\\
\wf{\Phi_S}&={\widehat{\widehat{\Phi_S}}}_- = \phi_{\wf{S}},\\
\wq{\phi_g} &= \phi_{\wq{g}}
\end{align*}
so the transformations are correct extensions. Since the inverse Weyl transform is clearly the inverse of the Weyl transform, we trivially have the following:
\begin{prop}\label{distrisomorphisms} Fourier transform is a topological isomorphism of $\sfud$ onto itself. The Weyl transform and Weyl quantization are topological isomorphisms between $\sopd$ and $\sfud$.
\end{prop}

\begin{remark}\rm From the results of section \ref{sec:regularity-theorem} below it follows that the Weyl quantisation (as introduced above) defines a unique map from $\sfud$ into the space of continuous linear maps $\sfunc{N}\to\sw'(\Rl^{N})$. This coincides with the rigged Hilbert space formulation given in \cite{hennings}.
\end{remark}

In order to facilitate the definitions with basic examples, we first briefly consider bounded operators, where the basic result is that the Weyl quantisation of a square-integrable function is Hilbert-Schmidt \cite{Pool}; this follows directly from Lemma \ref{wtunitary}. As a nontrivial example of a non-square integrable bounded function leading to a bounded operator which is not Hilbert-Schmidt, we mention the following result (Prop. 2 in \cite{We88}):
\begin{prop} Let $N=1$, and $f$ the indicator function of a pointed or double sector in $\phaseS$. Then $\wq{\Phi_f}=\Phi_A$ where $A$ is a bounded operator with absolutely continuous spectrum.
\end{prop}
Weyl quantisation of the delta-distribution was perhaps first studied in \cite{DeltaParity}; we can easily reproduce the result with our formalism: Let $\delta_a\in \sfud$ be the delta distribution supported at $a\in \phaseS$, and let $1$ denote the also function $x\mapsto 1$ on $\phaseS$. We then have
\begin{align*}
\widehat{\delta_a}& =\phi_{e^{i\{a,\cdot\}}}, & \check{\delta_a} &= \Phi_{W(-a)},\\
\wq{\delta_a} &=2^N\Phi_{W(a)\Pi W(a)^*}, & \wq{\delta_0}&=2^N\Phi_{\Pi},\\
\widehat{\Phi_{\idty}} &= \delta_0, & \widehat{\phi_{1}} &= \delta_0,\\
\wf{\Phi_{\idty}} &=1.
\end{align*}
In order to check these, we compute e.g.
\begin{align*}
\check{\delta_a}(T) &= \delta_a(\widehat{T}_-)= {\rm tr}[W(-a)T],\\
\wq{\delta_a}(T) &= \delta_a(\wf{T}) = 2^N{\rm tr}[W(a)\Pi W(a)^*T].
\end{align*}
As a second example, we compute the inverse Weyl transform of the derivatives of the delta-distribution:
\begin{align*}
\widecheck{D^\alpha\delta_0}(T) &=(D^\alpha\delta_0)(\widehat T_-)=(-1)^{|\alpha|} \delta_0(D^\alpha\widehat T_-)= (-1)^{|\alpha|} (D^\alpha\widehat T_-)(0)= (-1)^{|\alpha|}D_x^\alpha {\rm tr}[W(-x) T]|_{x=0}.
\end{align*}
Hence, $\widecheck{D^\alpha\delta_0}$ is a certain polynomial of $Q$ and $P$. For instance,
$\widecheck{D^{\alpha\vee 0}\delta_0}=(-i)^{|\alpha|} P^\alpha$, and $\widecheck{D^{0\vee \alpha}\delta_0}=i^{|\alpha|} Q^\alpha$.
Mixed derivatives are more complicated, for instance in case $N=1$, we get
$$
\widecheck{D^{(1,1)}\delta_0}=i\idty/2+PQ.
$$
Finally, we make the following interesting observation:
\begin{prop}\label{derivativequanti} The Weyl quantisations of the derivatives of the delta-distribution are the corresponding derivatives of the parity operator:
$$\wq{D^\alpha \delta_0}=2^N \Phi_{D^\alpha \Pi}, \quad \text{for all }\alpha\in \indm{2N}.$$
\end{prop}
\begin{proof}
We have
\begin{align*}
(\wq{D^\alpha \delta_0})(T) &= (D^\alpha \delta_0)(\wf{T})=(-1)^{|\alpha|} (D^\alpha\wf{T})(0)\\
&=(-1)^{|\alpha|} 2^N D^\alpha_{x}\tr[W(x)\Pi W(-x) T]|_{x=0} =2^N{\rm tr}[D^\alpha \Pi \,T].
\end{align*}
\end{proof}
Hence, highly singular classical distributions can correspond to polynomially bounded operators in the Weyl quantisation. We do not pursue this topic in more depth in this paper. Properties of Weyl quantisation of tempered distributions have been studied considerably, see for instance \cite{Ali,Daub80, Daub83}. However, as we mentioned in the introduction, the quantised object is typically not interpreted as a distribution on its own right, as we do here.

\subsubsection{Convolutions and their transformation properties}

Since we already have convolutions defined between elements of $\sop$, it is completely straightforward to define them between elements of $\sopd$ and $\sop$, as well as between $\sopd$ and $\sopd$, and other admissible combinations, in analogy to the usual convolution of $g\in \sfu$ and $\phi\in \sfud$,
\begin{align}
& \phi*g\in \sfud, & (\phi*g)(f) &= \phi(g_-*f).
\end{align}
This works because $f\mapsto g_-*f$ is continuous in $\sfu$. Since we know (see the preceding section) that also the Schwartz space convolutions
\begin{align*}
\sop\ni T&\mapsto S*T\in \sfu, & \sfu \ni f &\mapsto S*f\in \sop, & \sop\ni T&\mapsto g*T\in \sop
\end{align*}
are continuous for each $g\in \sfu$ and $S\in \sop$, we can set the following definition.

Let $g\in \sfu$, $S\in \sop$, $\phi\in \sfud$, and $\Phi\in \sopd$. We set
\begin{align}\label{defconv}
& \phi*S\in \sopd, & (\phi*S)(T) &= \phi(S_-*T), \\
& \Phi*S\in \sfud, & (\Phi*S)(f) &= \Phi(S_-*f), \\
& \Phi*g\in \sopd, & (\Phi*g)(T) &= \Phi(g_-*T).
\end{align}
Again, continuity of the convolution is a trivial consequence of the preceding observations.
\begin{prop}\label{convcont}
Each convolution is separately continuous with respect to both variables.
\end{prop}

Convolutions are typically characterised by their behaviour under the Fourier transform. It is straightforward to check that the following result holds:
\begin{prop}\label{fourierconvd} Let $g\in \sfu$, $S\in \sop$, $\phi\in \sfud$, and $\Phi\in \sopd$. Then
\begin{align*}
\widehat{\phi*g} &= \widehat{\phi}\widehat{g}, & (\phi * g){\check{}} &= \wq{(\widehat{\phi}\widehat{g})_-}, \\
\widehat{\phi*S} &= \widehat{\phi}\widehat{S}, \\
\widehat{\Phi*S} &= \widehat{\Phi}\widehat{S}, & (\Phi * S){\check{}} &= \wq{(\widehat{\Phi}\widehat{S})_-}, \\
\widehat{\Phi*g} &= \widehat{\Phi}\widehat{g}.\\
\wf{\Phi*g} &= \wf{\Phi}*g_-, & \wq{\phi*g}&= \wq{\phi}*g_-.
\end{align*}
\end{prop}

One can easily compute examples demonstrating the use of these relations in analogy to classical tempered distributions; for instance, we have
\begin{align*}
\delta_0*g &= \phi_g, & \delta_0*S&=\Phi_S, &\text{for all } g\in\sfu, \, S\in \sop.
\end{align*}

Concerning the derivatives, the relations \eqref{derc1}-\eqref{derc4} extend in a straightforward fashion for the appropriate convolutions. This demonstrates how convolutions can be used to regularise distributions; given a bounded operator $A$, the distributional derivative $\Phi:=D^\alpha \Phi_A$ is an unbounded operator, hence clearly "less regular" than $A$. By taking a convolution with an $f\in \sfu$, we get
$$
f*\Phi_A = f* D^\alpha \Phi_A=D^\alpha f *\Phi_A,
$$
which is again a bounded operator because $D^\alpha f$ is a Schwartz function.

\subsubsection{Positive correspondence maps}

Wigner function has the well-known disadvantage of not necessarily being positive for positive operators; a description of quantum-classical correspondence maps which preserve positivity is obtained instead via convolutions (see e.g. \cite{We88,Werner} for discussion). Using the above definitions for convolutions, we can naturally extend these maps to distributions. In fact, for a fixed positive $S_0\in \sop$, the maps
\begin{align*}
\sfud \ni \phi\mapsto \phi* S_0\in \sopd, & &\sopd \ni \Phi\mapsto \Phi* S_0\in \sfud
\end{align*}
provide translation-covariant and (pointwise) positivity preserving correspondence of distributions. One could develop the theory of these maps further, along the lines discussed in \cite{Werner, KLSW12} in the case where $S_0$ is only required to be trace class. We only note here the following consequence of duality and Lemma \ref{schwartzreg}:
\begin{prop} The classical-to-quantum correspondence
\begin{align*}
&\sfud \ni \phi\mapsto \phi* |0\rangle\langle 0|\in \sopd
\end{align*}
induced by the ground state $|0\rangle\langle 0|$ of $\Htot$, is injective.
\end{prop}
\begin{proof} Assuming $\phi\in \sfud$ with $\phi*|0\rangle\langle 0|=0$, we have $\phi(T*|0\rangle\langle 0|)=0$ for all $T\in \sop$ by definition. Since the range of $T\mapsto T*|0\rangle\langle 0|$ is dense in $\sfu$ by Lemma \ref{schwartzreg}, this implies $\phi=0$.
\end{proof}

\begin{remark}\rm
We note that that $\phi\mapsto \phi*|0\rangle\langle 0|$ is the distributional generalisation of the well-known basic instance of \emph{coherent state quantisation} (see e.g. \cite{Ali}). The quantum-to-classical correspondence map $\Phi \mapsto \Phi* |0\rangle\langle 0|$ similarly generalises the Husimi Q-representation often appearing in quantum optical literature.
\end{remark}

\subsection{Regularity theorem}
\label{sec:regularity-theorem}

Above we have already seen that polynomially bounded operators naturally define tempered distributions. One can then ask if all of them arise in this way, and the answer turns out to be negative. However, the most general tempered distribution is not very far from being an operator; in fact, we  have the following result for $\sopd$:

\begin{prop}\label{regprop} Let $A$ be a Hilbert-Schmidt operator on $\mathcal H$, and $P_L$, $P_R$ two arbitrary polynomials of $Q_i$ and $P_j$, $i,j=1,\ldots N$. Then the formula
$$
\Phi(T):={\rm tr}[A \,(P_L T P_R)], \quad T\in \sop
$$
defines an element of $\sopd$. Conversely, every $\Phi\in \sopd$ is of this form.
\end{prop}
\begin{proof} Since $P_L T P_R\in \sop$, the formula is well-defined, and we have
$|\Phi(T)|\leq \|A\|_2\|P_L T P_R\|_2$. Since this is bounded above by a linear combination of seminorms of the form \eqref{pseminorms}, we conclude that $\Phi\in \sopd$.

In order to prove the converse, let $\Phi\in \sopd$. We know from Prop. \ref{Nrepd} that $|\Phi_{\alpha\vee\alpha'}|\leq C(\alpha\vee\alpha'+1)^{\beta\vee\beta'}$ for all $\alpha\vee\alpha'\in I_{2N}$, and a fixed $\beta\vee\beta'$. We define $a_{\alpha\vee\alpha'} =(\alpha\vee\alpha'+1)^{-\beta\vee\beta'-1}\Phi_{\alpha\vee\alpha'}$ (where $-1$ is understood as multiindex with all entries $-1$). Then
$$\sum_{\alpha\vee\alpha'} |a_{\alpha\vee\alpha'}|^2\leq C^2\sum_{\alpha\vee\alpha'} (\alpha\vee\alpha'+1)^{-2}=C^2\prod_{i=1}^{N}\left( \sum_{\alpha_i}(1+\alpha_i)^{-2}\sum_{\alpha_i'}(1+\alpha_i')^{-2}  \right)<\infty,$$
so the series
$$
A=\sum_{\alpha\vee\alpha'} a_{\alpha\vee\alpha'} |\alpha\rangle\langle \alpha'|
$$
converges in the Hilbert-Schmidt norm to a bounded (Hilbert-Schmidt) operator $A$. By Prop. \ref{Nrepd} we have, for any $T\in \sop$,
\begin{align*}
\Phi(T)&=\sum_{\alpha\vee\alpha'} \Phi_{\alpha\vee\alpha'} \langle \alpha'|T|\alpha\rangle\\
&= \sum_{\alpha\vee\alpha'} a_{\alpha\vee\alpha'} (\alpha\vee\alpha'+1)^{\beta\vee\beta'+1}\langle \alpha'|T|\alpha\rangle\\
&= \sum_{\alpha\vee\alpha'} a_{\alpha\vee\alpha'}\langle \alpha'|(H+\tfrac 12)^{\beta+1} T (H+\tfrac 12)^{\beta'+1}|\alpha\rangle\\
&= {\rm tr}[A \, (H+\tfrac 12)^{\beta+1} T (H+\tfrac 12)^{\beta'+1}],
\end{align*}
where the last equality follows because $(H+\tfrac 12)^{\beta+1} T (H+\tfrac 12)^{\beta'+1}$ is a Schwartz operator, hence Hilbert-Schmidt.
\end{proof}
\begin{remark}\rm Since $H^{-2}$ is trace class, the operator $H^{-2}AH^{-2}$ is Hilbert-Schmidt for any bounded operator $A$. Hence we observe that the above proposition holds also when the Hilbert-Schmidt operator $A$ is replaced by a general bounded operator.
\end{remark}
\begin{remark}\rm In the preceding section we defined polynomially bounded operators $B$ as distributions $\Phi_{B}$. From the above regularity theorem it might seem that every $\Phi\in \sopd$ is of the form $\Phi=\Phi_B$ for $B=P_RAP_L$, but this is \emph{not} the case, because the operator $A$ does not necessarily map into the domain of $P_R$, in which case $B$ is not a (densely defined) operator. As a simple example in case $N=1$, take
$$\Phi(T) ={\rm tr}[|\varphi\rangle\langle\varphi| \, QTQ],$$
where $\varphi\in L^2(\mathbb R,dq)$, but $\int |q|^2|\varphi(q)|^2=\infty$, i.e. $\varphi\notin \dom(Q)$. Now while $QTQ$ is certainly well defined for any Schwartz operator $T$, the product $Q|\varphi\rangle\langle \varphi|Q$ is only defined on the trivial domain $\{0\}$. There does not exist an operator $B$ with domain $\dom(B)$ including the Schwartz space $\sfunc{N}$, such that $\Phi=\Phi_A$. In fact, if this were the case, then for any $\psi,\psi'\in \sfunc{N}$ we would have $\langle \psi |A\psi'\rangle=\langle \varphi|Q\psi\rangle\langle Q\psi'|\varphi\rangle$, which is impossible because the left-hand side is continuous in $\psi$ with respect to the Hilbert space norm, but the right-hand side is not.

However, the formal expression $B=P_R A P_L$ can always be interpreted as the \emph{quadratic form}
$$
(\psi,\psi')\mapsto B(\psi,\psi'):=\langle P_R^*\psi |AP_L\psi'\rangle, \quad (\psi,\psi')\in \sfunc{N}\times\sfunc{N}.
$$
\end{remark}
In fact, we have the following result:
\begin{prop} \label{prop:3}
Let $B$ be a sesquilinear form on $\sfunc{N}\times\sfunc{N}$. The following conditions are equivalent:
\begin{itemize}
\item[(i)] $B$ is separately continuous in both arguments.
\item[(ii)] $B$ is jointly continuous.
\item[(iii)] There exists $\beta,\beta'\in I_N$ and a constant $C>0$ such that
$$\left|B\left((H+\tfrac 12)^{-\beta}\psi,(H+\tfrac 12)^{-\beta'}\psi'\right)\right|\leq C\|\psi\|\|\psi'\|.$$
\item[(iv)] there exists a Hilbert-Schmidt operator $A$, and polynomials $P_L$, $P_R$ of $Q$ and $P$, such that $B(\psi,\psi')=\langle P_R^*\psi |AP_L\psi'\rangle$.
\end{itemize}
The associations
$$B(\psi',\psi)=\Phi(|\psi \rangle\langle \psi'|)=\langle \psi',F\psi\rangle$$ 
establishe one-to-one correspondences between sesquilinear forms $B$ satisfying the above equivalent conditions,  elements $\Phi\in \sopd$, and continuous linear operators $F:\sfunc N\to\sw'(\Rl^N)$, where the scalar product 
in the expression with $F$ is the canonical bilinear form between $\sfunc N$ and $\sw'(\Rl^N)$, made conjugate linear in the first argument. 
\end{prop}
\begin{proof} For (i) implies (ii), see \cite{RSI}. Using the seminorms for $\sfunc{N}$ given by powers of $H$, we see that (ii) implies (iii). Assuming (iii), we get $$|B((H+\tfrac 12)^{-\beta-1}|\alpha\rangle,(H+\tfrac 12)^{-\beta'-1}|\alpha'\rangle)|\leq C (\alpha\vee\alpha'+1)^{-1}.$$ Since the $(\alpha\vee\alpha'+1)^{-1}$ are square summable, we have matrix elements of a Hilbert-Schmidt operator on the left hand side. It follows that (iv) holds for $P_L=(H+\tfrac 12)^{\beta+1}$ and $P_R=(H+\tfrac 12)^{\beta'+1}$. Clearly, (iv) implies (i).

The correspondence between $B$ and $\Phi$ follows from the Prop.~\ref{regprop}. Alternatively, we could have used Schwartz kernel theorem and the standard regularity theorem for tempered distributions. The correspondence between $B$ and $F$ follows directly from item ({\em i}), observing that $\sw'(\Rl^N)$ is equipped with the weak-* topology. 
\end{proof}

Some previous work on $\sopd$ has been done in terms of the operators $F$ in the above Proposition. This is particularly natural in the language of rigged Hilbert spaces \cite{antoine} and Wigner-Weyl quantization \cite{hennings}. In this context it is natural to look especially at those $F$ mapping in to $\HH$, or even into $\sfunc N$, like the class $O_M(\HH)$ (see Prop.~\ref{ppalgebras}).

In the case of $\sfud$, the well-known \emph{regularity theorem} (see e.g. \cite[Theorem V.10]{RSI}) states that each tempered distribution can be obtained from a polynomially bounded function by differentiation. The following result is its quantum version:

\begin{thm}\label{regularitythm}
Every element $\Phi\in \sopd$ is of the form
$$
\Phi = \sum_{\alpha: |\alpha|\leq m} D^\alpha \Phi_{A_\alpha},
$$
where each $A_\alpha$ is a densely defined polynomially bounded operator.
\end{thm}
\begin{proof} According to Prop. \ref{regprop}, we can write $\Phi(T)$ as a linear combination of the terms of the form
$${\rm tr}[A (P_LTQ^\alpha P^\beta)]$$ with some bounded operator $A$, a polynomial $P_L$ of $Q$ and $P$, and some $\alpha,\beta\in \indm{N}$. Now the reason why $\Phi$ is not in general given by a densely defined operator, is that $A$ does not need to map $\sfunc N$ into the domain of $Q^\alpha P^\beta$ so we cannot permute the latter to the other side of $A$ inside the trace. However, we can instead commute it through $T$ if we allow for nested commutators with $Q$ and $P$ to appear via iterative applications of the maps $\mathcal L_{Q_i}(\cdot)=[Q_i,(\cdot)]$ and $\mathcal L_{P_i}(\cdot)=[P_i,(\cdot)]$. In fact, $P_LTQ^\alpha P^\beta$ is clearly a linear combination of terms of the form
$$
\tilde P_L \,\mathcal L_{P_N}^{\alpha_N}\circ\cdots \circ \mathcal L_{P_1}^{\alpha_1}\circ\mathcal L_{Q_N}^{\beta_N}\circ\cdots\circ\mathcal L_{Q_1}^{\beta_1}(T),
$$
where $\tilde P_L$ is some polynomial of $Q$ and $P$ (now appearing only on the left side), and some indices $\alpha$ and $\beta$. But according to \eqref{liederexpansion}, this commutator expression is equal to $(-i)^{|\beta|}i^{|\alpha|}D^{\alpha\vee\beta}T$; hence, $\Phi(T)$ is a linear combination of terms of the form ${\rm tr}[A\tilde P_L D^{\alpha\vee\beta} T]$. Since each operator $A\tilde P_L$ is densely defined on $\sfunc N$, and polynomially bounded, the proof is complete.
\end{proof}

\subsection{Convergence of distributions}
\label{sec:conv-distr}

Recall that we have equipped $\sfud$ and $\sopd$ with the weak-* topology. In this section we look at approximations of distributions in the sense of this topology.

\subsubsection{Approximate identity}

The following tool is sometimes useful in this context: We say that a net $(j_\epsilon)_{\epsilon>0}$ of functions $j_\epsilon\in \sfu$ is \emph{an approximate identity} if
$$\lim_{\epsilon\rightarrow 0} \phi*j_\epsilon = \phi, \qquad \text{for each }\phi\in \sfu.$$
The definition was given in terms of the space $\sfud$. However, using the last relation of Prop. \ref{fourierconvd}, as well as Prop. \ref{distrisomorphisms}, we see that $(j_\epsilon)_{\epsilon>0}$ is an approximate identity if and only if
$$
\lim_{\epsilon\rightarrow 0} \Phi*j_\epsilon = \Phi, \qquad \text{for each }\Phi\in \sfud.
$$
If $g\in \sfu$ is such that $\int g(x) dx=1$, then $g_\epsilon(x) = \epsilon^{-2N}g(x/\epsilon)$ defines an approximate identity.

Indeed, $\widehat{g_\epsilon}(x) = \widehat{g}(\epsilon x)$, and $|\widehat{g}(x)|\leq 1= |\widehat{g}(0)|$, so $\widehat{g_\epsilon}$ converges to $1$ uniformly for $x$ in every compact set. Now each $f\in \sfu$, the expression $|x^\alpha D^\beta ((1-\widehat{g_\epsilon}) f)(x)|$ is bounded above by a sum of terms of the form $\epsilon^{|\beta'|} |D^{\beta'}\widehat{g}(\epsilon x)\, x^\alpha D^{\gamma} f(x)|$ with $|\beta'|\neq 0$, and the term $|(1-\widehat{g_\epsilon}) x^\alpha D^\beta f(x)|$. Since $x^\alpha D^\beta f(x)$ vanishes at infinity, it follows that $\|(1-\widehat{g_\epsilon}) f\|_{\alpha,\beta}\rightarrow 0$ as $\epsilon\rightarrow 0+$. Hence, $\lim_{\epsilon\rightarrow 0} g_\epsilon*f=f$ in the topology of $\sfu$, for each $f\in \sfu$, and so $(g_\epsilon)_{\epsilon>0}$ is an approximate identity.

\subsubsection{Approximations of the delta distribution and parity}
By taking $\phi=\delta_0$ in the definition of approximate identity, we notice that any approximate identity $(j_\epsilon)$ approximates the delta distribution, i.e. for $\phi^\epsilon:=\phi_{j_\epsilon}$ we have
\begin{equation}\label{weakdelta}
\lim_{\epsilon\rightarrow 0} \phi^\epsilon=\delta_0.
\end{equation}
Any net $(\phi^\epsilon)_{\epsilon>0}$ satisfying \eqref{weakdelta} is called \emph{an approximation of the delta distribution}. There are various ways of defining functions $g_\epsilon$ such that $(\phi_{g_\epsilon})_{\epsilon>0}$ is an approximation of the delta distribution. For instance, take $g_\epsilon\in L^1(\phaseS)$  such that $\int g_\epsilon(x) dx = 1$ for all $\epsilon>0$, $\sup_{\epsilon>0} \int |g_\epsilon(x)|dx<\infty$, and $\lim_{\epsilon\rightarrow 0} \int_{|x|\geq \delta} |g_\epsilon(x)| dx =0$ for each $\delta>0$.

Since Wigner quantization is an isomorphism between $\sfud$ and $\sopd$, it follows that the Wigner quantization of any approximation of the delta function approximates the distribution $2^N\Phi_{\Pi}$. As an example, we let $g= \wf{T_0}\in \sfu$, i.e. the Wigner function of the ground state of the oscillator. Now
$$
g(q,p)= 2^Ne^{-(q^2+p^2)}.
$$
Since the integral $\int g(x) dx=1$, the net $g_\epsilon(x) = \epsilon^{-2N}g(x/\epsilon)$ is an approximate identity. The kernel of the Wigner quantization of $g_\epsilon$ is given by
$$
K^{\wq{g_\epsilon}}(q,q')= \frac{1}{(\epsilon \sqrt{\pi})^N} e^{-\frac 14\left[(q+q')^2/\epsilon^2+\epsilon^2(q-q')\right]}.
$$
Asymptotically, as $\epsilon\rightarrow 0$, the second term in the exponent becomes negligible, and for fixed $q$, we have
$$\lim_{\epsilon\rightarrow 0} \frac{1}{(\epsilon \sqrt{\pi})^N} e^{-\frac 14(q+\cdot)^2/\epsilon^2}=2^N\delta_{-q}$$
in the dual of $\sfunc N$, so indeed, $\wq{g_\epsilon}\in \sop$ approximates the distribution $2^N\Phi_{\Pi}$. Note that $\wq{g_1}=T_0$, a rank one operator, which is not the case for any $\wq{g_\epsilon}$ with $\epsilon\in (0,1)$. Moreover, $\lim_{\epsilon\rightarrow 0}\|\wq{g_\epsilon}\|_2=\infty$. However, ${\rm tr}[ \wq{g_\epsilon}] = \int dx g_\epsilon(x) dx = 1$ for all $\epsilon>0$, so the trace converges to $1$, even though the parity operator is not in the trace class.

\subsubsection{Finite dimensional approximations}

Consider sequences $(A^{(n)})_n$ and $(B^{(n)})_n$ of operators which are
finite rank, diagonal in the number basis and satisfying
\begin{displaymath}
  \lim_{n \rightarrow \infty} \|(A^{(n)}-\idty) H^{- \gamma}\| =
  \lim_{n\rightarrow\infty} \| (B^{(n)} - \idty) H^{-\gamma'}\| = 0.
\end{displaymath}
For each distribution $\Phi \in \sopd$ we get a sequence $(\Phi^{(n)}_n)$ by
\begin{displaymath}
  \Phi^{(n)}(T) = \Phi(A^{(n)}TB^{(n)}), \quad T \in \sop.
\end{displaymath}
Since $A^{(n)}$, $B^{(n)}$ are of finite rank the distribtution $\Phi^{(n)}$
is given in terms of a finite rank operator with matrix elements
\begin{displaymath}
  \Phi^{(n)}_{\alpha \vee \beta} = A^{(n)}_\alpha \Phi(E_{\alpha \vee \beta}) B^{(n)}_\beta
\end{displaymath}
and where $A^{(n)}_\alpha$, $B^{(n)}_\beta$ are the eigenvalues of $A^{(n)}$
and $B^{(n)}$, respectively. The rank of $\Phi^{(n)}$ is obviously smaller or
equal to the ranks of $A^{(n)}$ and $B^{(n)}$. The following proposition shows that
the $\Phi^{(n)}$ provide an approximation of $\Phi$ in $\sopd$. In other
words: every distribution $\Phi$ can be approximated by a sequcence of finite
rank operators. We will provide an explicit example for this result in
Sect. \ref{sec:finite-dimens-appr}.

\begin{prop} \label{prop:8}
  Considering $(A^{(n)})_n$, $(B^{(n)})_n$, $\Phi \in \sopd$ and
  $(\Phi^{(n)})_n$ as just described, we get $\lim_{n \rightarrow \infty}
  \Phi^{(n)} = \Phi$ in $\sopd$.
\end{prop}

\begin{proof}
  We have to show that $\lim_{n \rightarrow \infty} \Phi^{(n)}(T) = \Phi(T)$
  for all $T \in \sop$, which is equivalent to
  \begin{displaymath}
    \lim_{n \rightarrow\infty} A^{(n)} T B^{(n)} = T \quad\text{in $\sop$}.
  \end{displaymath}
  Hence, with the results of Sect. \ref{sec:kern-matr-repr} we have to look at
  \begin{align}
    \| H^\alpha &(A^{(n)} T B^{(n)} - T) H^{\alpha'}\|_2 \leq \nonumber \\
    & \leq \| H^\alpha (A^{(n)} T B^{(n)} - T B^{(n)})H^{\alpha'}\|_2 + \|
    H^\alpha (TB^{(n)} - T) H^{\alpha'}\|_2 \nonumber \\
    &\leq \| H^\alpha (A^{(n)} - \idty) H^{(-\gamma)} H^{(\gamma)} T B^{(n)} H^{\alpha'}
    \|_2 + \| H^\alpha T H^{\gamma'} H^{-\gamma'} (B^{(n)} - \idty)
    H^{\alpha'}\| \label{eq:20}
  \end{align}
  The first term on the right hand side of (\ref{eq:20}) can be estimated as
  \begin{align}
   \| H^\alpha (A^{(n)} - \idty) H^{(-\gamma)} H^{(\gamma)} T B^{(n)} H^{\alpha'}
    \|_2 &= \| (A^{(n)} - \idty) H^{(-\gamma)} H^{(\gamma+\alpha)} T B^{(n)} H^{\alpha'}
    \|_2 \nonumber \\
    &\leq \| (A^{(n)} - \idty) H^{(-\gamma)} \| \|H^{(\gamma+\alpha)} T
    B^{(n)} H^{\alpha'}\|_2 \label{eq:21}.
  \end{align}
  By assumption $\| (A^{(n)} - \idty) H^{(-\gamma)} \| \rightarrow 0$ if $n
  \rightarrow \infty$. Hence the left hand side of (\ref{eq:21}) vanishes in
  the limit. The second term in (\ref{eq:20}) can be handled similarly, which
  completes the proof.
\end{proof}


\section{Application 2: Spectral densities and distribution valued measures}
\label{sec:spectr-dens-distr}

We now look at a natural application of the theory of the preceding
section. Consider a Hilbert space $\mathcal{H} = L^2(\mathbb{R}^n)$ and a positive operator
valued measure $E: \mathcal{B}(\mathbb{R}) \rightarrow \mathcal{B}(\mathcal{H})$,
where $\mathfrak{B}(\mathbb{R})$ denotes the Borel sigma algebra of the real
line. Since each bounded operator $E(\Delta) \in \mathcal{B}(\mathcal{H})$ can
be regarded as a distribution, we can regard $E$ as a \emph{positive
  distribution} valued measure. More generally, we define $\Phi \in \sopd$ to be
positive if $\Phi(\myrho) \geq 0$ for all positive $\myrho \in \sop$ and a map $E:
\mathfrak{B}(\mathbb{R}) \rightarrow \sopd$ is called a positive distribution
valued measure if $E$ is $\sigma$-additive, $E(\Delta)$ is positive for all
$\Delta \in \mathfrak{B}(\mathbb{R})$ and $E(\emptyset) = 0$, $E(\mathbb{R}) =
\idty$. holds. This reinterpretation opens several possibilities which are
not available for operator valued or spectral measures. We will discuss some
of them using the position operator as an illustrating example. This includes
in particular the reinterpretation and extension of results from mean field
theory \cite{Fluct} where the position operator is approximated in terms of
fluctuation operators of finite spin systems.

\subsection{Spectral densities}

For the rest of this section we consider the case $N=1$, and concentrate on the position operator $Q$. The same discussion can be done for the momentum $P$ on the momentum space, and then transforming back using the Fourier transform. Moreover, the generalisation to more
degrees of freedom is easily possible by adding more tensor factors as
innocent bystanders.

The spectral measure $E(\Delta)$ of $Q$ is
given by
\begin{equation}\label{eq:15}
  E(\Delta) \psi = \chi_\Delta \psi \quad \Delta \in
  \mathfrak{B}(\mathbb{R}),\quad \psi \in \mathcal{H},
\end{equation}
where $\chi_\Delta$ denotes the characteristic function of $\Delta$ and appears
here as the corresponding multiplication operator. It is well known that $E$
does not admit a density with respect to the Lebesgue measure $dq$,
i.e. there is no map $\varepsilon : \mathbb{R} \rightarrow
\mathcal{B}(\mathcal{H})$ such that
\begin{displaymath}
  E(\Delta) = \int_\Delta \varepsilon(q) dq.
\end{displaymath}
holds. However, this problem can be solved if we consider $E$ as a
distribution valued measure instead. To see this look at the quadratic forms
\begin{equation} \label{eq:5}
  \varepsilon_q : \mathfrak{S}(\mathbb{R}) \times \mathfrak{S}(\mathbb{R})
  \rightarrow \mathbb{C}, \quad (\psi,\phi) \mapsto \varepsilon_q(\psi,\phi) =
  \overline{\psi(q)} \phi(q).
\end{equation}
They are positive but not closable. Furthermore they satisfy
$\varepsilon_q(Q\psi,\phi) = \varepsilon_q(\psi,Q \phi) = q
\varepsilon_q(\psi,\phi)$, which resembles a condition on an
``eigenprojection'' of $Q$. In other words we would like to have
\begin{equation} \label{eq:4}
  Q \varepsilon_q = \varepsilon_q Q = q \varepsilon_q\quad \forall q \in \mathbb{R}
\end{equation}
Of course $\varepsilon_q$ is not even an operator. However, we can show that
it is a distribution, and in this sense the statement of Eq. (\ref{eq:4}) is
perfectly well-defined. More precisely the following proposition holds:

\begin{prop} \label{prop:7}
  The family of quadratic forms $\varepsilon_q$ defined in Eq. (\ref{eq:5})
  has the following properties:
  \begin{itemize}
  \item[(a)] \label{item:1}
    $\varepsilon_q$ defines via Thm \ref{prop:3} a unique distribution $\varepsilon_q
    \in \sopd$. The corresponding kernel distribution is $\delta_q \otimes
    \delta_q$ where $\delta_q$ denotes the Dirac delta distribution concentrated at $q$.
  \item[(b)] \label{item:2}
    Considering mutiplications of distributions with polynomially bounded
    operators (cf. Sect. \ref{sec:polyn-bound-oper}) we have
    \begin{displaymath}
      f(Q) \varepsilon_q = \varepsilon_q f(Q) = f(q) \varepsilon_q \quad
      \forall q \in \mathbb{R}
    \end{displaymath}
    for all polynomials $f:\mathbb R\to \Cx$.
  \item[(c)] \label{item:3}
    $Q$ -- as a distribution -- can be reconstructed as a weak integral over
    the $\varepsilon_q$. More precisely
    \begin{displaymath}
      \tr(\myrho f(Q)) = \int_\mathbb{R} f(q) \varepsilon_q(\myrho) dq
    \end{displaymath}
    holds for all $f \in L^\infty(\mathbb{R},dq)$ and $T \in \sop$.
    \item[(d)] \label{item:4} The Wigner function of $\varepsilon_q$ is the delta-distribution on the $q$-variable, i.e.
   $$\wf{\varepsilon_q}=\delta_q\otimes \idty,\quad \text{ for all } q\in \mathbb R.$$
  \end{itemize}
\end{prop}

\begin{proof}
  Part (a). Consider the distribution $\delta_q\otimes\delta_q\in
  \sw'(\Rl^{2N})$. According to Prop. \ref{prop:4} it defines a distribution
  $\varepsilon_q \in \sopd$ such that $\varepsilon_q(\myrho) = \delta_{q}\otimes \delta_q (\kernel{\myrho})$ holds for any Schwartz operator $\myrho$ and its
  kernel function $\kernel{\myrho}$. Hence we have
  \begin{displaymath}
    \varepsilon_q(\ketbra{\phi}{\psi}) = \delta_q\otimes\delta_q (\phi
    \otimes \overline{\psi}) = \phi(q) \overline{\psi}(q),
  \end{displaymath}
  which is the quadratic form from Eq. (\ref{eq:5}). Hence Part (a)
  follows from Prop. \ref{prop:3}.

  Part (b). If $\kernel \myrho$ is the kernel function of $\myrho \in \sop$,
  the kernel functions of $Q \myrho$ and $\myrho Q$ are
  \begin{displaymath}
    \bigl[(Q \otimes \idty) \kernel{\myrho}\bigr](q_1,q_2) = q_1 \kernel{\myrho}(q_1,
    q_2)\quad \text{and}\ \bigl[(\idty \otimes Q)\bigr]\kernel{\myrho}(q_1, q_2) = q_2
    \kernel{\myrho}(q_1,q_2).
  \end{displaymath}
  Hence we get
  \begin{align*}
    \delta_q \otimes \delta_q ( (Q \otimes \idty) \kernel{\myrho}) &= q \kernel{\myrho}(q,q) =
    q \delta_q \otimes \delta_q (
    \kernel{\myrho} ), \\  \delta_q \otimes \delta_q ((\idty \otimes Q) \kernel{\myrho}) &= q \kernel{\myrho}(q,q) =
    q \delta_q \otimes \delta_q( \kernel{\myrho} ).
  \end{align*}
  and with a polynomial $f$
  \begin{displaymath}
    \delta_q \otimes \delta_q ( (f(Q) \otimes \idty) \kernel{\myrho}) = f(q) \delta_q
    \otimes \delta_q ( \kernel{\myrho} ), \quad \delta_q \otimes \delta_q ( (\idty
    \otimes f(Q)) \kernel{\myrho}) = f(q) \delta_q \otimes \delta_q ( \kernel{\myrho} ),
  \end{displaymath}
  which is the statement written in terms of kernels rather than operators and
  distributions in $\sopd$. Hence part (b) follows from
  Prop. \ref{prop:4}.

  Part (c). We use the fact that any $\myrho \in \sop$ can be written
  as a convergent (in the topology of $\sop$) series of terms
  $\ketbra{\psi}{\phi}$ with Schwartz functions $\psi, \phi$;
  cf. Prop. \ref{schwartzproperties}. Furthermore the functionals $\myrho
  \mapsto \tr(\myrho f(Q))$ and $\myrho \mapsto \varepsilon_q(\myrho)$ are continuous
  in this topology (cf. Sect. \ref{sec:polyn-bound-oper}). Hence it is sufficient
  to prove the statement for $\myrho = \ketbra{\psi}{\phi}$, and we get
  \begin{align*}
    \tr(\myrho f(Q)) = \langle \phi, f(Q) \psi \rangle = \int_\mathbb{R} f(q)
    \overline{\phi(q)} \psi(q) dq &= \int_\mathbb{R} f(q) \varepsilon_q
    (\ketbra{\psi}{\phi}) dq\\ &= \int_\mathbb{R} f(q)
    \varepsilon_q(\myrho) dq
  \end{align*}
  which was to show.

  Part (d). We first note that by Lemma \ref{wtunitary}, the kernel is given in term terms of the Weyl transform via $K^\myrho=V^*(\id\otimes F)U^*\widehat \myrho$, where $U$ and $V$ are given explicitly in that Lemma. Hence we can directly compute the Weyl transform:
   \begin{align*}
 \widehat \varepsilon_q(f) &= \varepsilon_q(\check f_-)=\delta_{q}\otimes \delta_q(K_{\check f_-})=K_{\check f_-}(q,q)=((\idty\otimes F)U^*f_-)(0,q)\\
 &=\frac{1}{(2\pi)^{N}}\int dp\, e^{iq\cdot p} f(0,p).
 \end{align*}
Now
\begin{align*}
\wf{\varepsilon_q}(f)&=\widehat{\widehat{\varepsilon_q}}(f_-)=\widehat{\varepsilon_q}(\widehat f)
=\frac{1}{(2\pi)^{N}}\int dp e^{iq\cdot p} \frac{1}{(2\pi)^{N}}\int e^{-i\{(0,p),(q',p')\}} f(q',p') dq'dp'\\
&=\frac{1}{(2\pi)^{N}}\int dp'\left(\frac{1}{(2\pi)^{N}}\int dp e^{iq\cdot p} \int e^{-ip\cdot q'} f(q',p') dq'\right)\\
&= \frac{1}{(2\pi)^{N}}\int dp' f(q,p') =(\delta_q\otimes \idty)(f).
\end{align*}
This completes the proof.
\end{proof}

Part (c) leads to a simple corollary concerning the expectation values $\tr(\myrho E(\Delta))$ which will
be of use in Subsection \ref{sec:appr-posit}.

\begin{cor} \label{cor:1}
  For a Schwartz operator $\myrho$ with kernel $\kernel{\myrho}$, the
  following equation holds for all $a\leq b$:
  \begin{equation}
    \tr(E([a,b]) \myrho) = \int_a^b \kernel{\myrho}(q,q) dq
  \end{equation}
\end{cor}

\begin{proof}  Denote $\Delta = [a,b]$; then we have $E(\Delta) = \chi_\Delta(Q)$, where $\chi_\Delta$ denotes the characteristic function of
  $\Delta$. Hence from Prop. \ref{prop:7} (c) we get
  \begin{displaymath}
    \tr(E(\Delta) \myrho) = \tr(\chi_\Delta(Q) \myrho) = \int_{\mathbb{R}} \chi_\Delta(q) \varepsilon_q(\myrho) dq =
    \int_a^b \varepsilon_x(\myrho) dq.
  \end{displaymath}
  Using the fact that the kernel distribution of $\varepsilon_q$ is, according to
  Prop. \ref{prop:7} (a), given by $\delta_q \otimes \delta_q$ we get (cf. also Prop. \ref{prop:4}):
  \begin{displaymath}
    \int_a^b \varepsilon_q(\myrho) dq = \int_a^b \delta_q \otimes \delta_q (\kernel{\myrho}) dq = \int_a^b \kernel{\myrho}(q,q) dq
  \end{displaymath}
  what was to show.
\end{proof}

\subsection{Finite dimensional approximations}
\label{sec:finite-dimens-appr}

We will turn now to a slightly different topic, namely an example of the
approximation results from Sect. \ref{sec:conv-distr}. To this end consider
the $M-$fold symmetric tensor product $\mathcal{K}_M = (\mathbb{C}^2)^{\otimes M}_+$ of
the two dimensional Hilbert space $\mathbb{C}^2$. The number basis in
$\mathcal{K}_M$ is denoted by $\ket{n;M}$, i.e.
\begin{displaymath}
  \ket{n;M} = {M \choose n}^{1/2} S_M \ket{0}^{\otimes (M-n)} \otimes
  \ket{1}^{\otimes n}
\end{displaymath}
where $S_M$ is projection $(\mathbb{C}^2)^{\otimes M} \rightarrow
(\mathbb{C}^2)^{\otimes M}_+$ and $\ket{0},\ket{1}$ denotes the canonical
basis. On $\mathcal{K}_M$ we can define for each $a \in
\mathcal{B}(\mathbb{C}^2)$ and the density operator $\vartheta = \ketbra{0}{0}$
on $\mathbb{C}^2$ the fluctuation operators
\begin{displaymath}
  F_M(a) = \frac{1}{\sqrt{M}} \left( \sum_{j=1}^M a^{(j)} - \tr(a \vartheta)
  \right),\quad a^{(n)} = \idty^{\otimes (n-1)} \otimes a \otimes
  \idty^{\otimes (M-n)}.
\end{displaymath}
The $F_M(a)$ measure small (of order $\sqrt{M}$) quantum fluctuations around
the \emph{reference state} $\vartheta$. They play a crucial role in
non-commutative versions of the central limit theorem
(cf. e.g. \cite{GVer,Gvet,Taku,RW}) and more recently in the theoretical
discussion of matter-light interactions \cite{Fluct,Narnhofer1,Narnhofer2}.

We will now consider in particular $a=2^{-1/2}\sigma_{1/2}$, where
$\sigma_{1/2}$ denote the Pauli operators, i.e.
\begin{displaymath}
  Q_M = F_M\left(\frac{\sigma_1}{\sqrt{2}}\right) = \sqrt{\frac{2}{M}} L_{M,1}
  \quad
  P_{M} = F_M\left(\frac{\sigma_2}{\sqrt{2}}\right) = \sqrt{\frac{2}{M}} L_{M,2}
\end{displaymath}
here $L_{M,\alpha}$ denote global pseudo-spin operators given by
\begin{displaymath}
  L_{\alpha,M} = \frac{1}{2} \sum_i \sigma_{\alpha}^{(i)}, \quad \alpha=1,2,3.
\end{displaymath}
In addition we can introduce ladder operators
\begin{displaymath}
  A_{M} = \frac{L_{M,+}}{\sqrt{M}} = \frac{1}{\sqrt{2}} (Q_M + i P_M),\quad
  A^*_M = \frac{L_{M,-}}{\sqrt{M}} = \frac{1}{\sqrt{2}} (Q_M - i P_M)
\end{displaymath}
with
\begin{displaymath}
  L_{M,\pm} = L_{M,1} \pm i L_{M,2} =  \frac{1}{2}\sum_{j=1}^M\sigma_\pm^{(j)}
\end{displaymath}
with $\sigma_\pm = \sigma_1\pm i\sigma_2$ in terms of Pauli matrices.

The $Q_M$, $P_M$ are defined on the Hilbert space $\mathcal{K}_M$. However we
can embed the latter into $\mathcal{H} = L^2(\mathbb{R})$ if we identify the basis
elements $\ket{n;M} \in \mathcal{K}_M$ with the $n^{\rm th}$ Hermite function
$\ket{n} \in \mathcal{H}$. In that way the $Q_M, P_M$ and $a_M, a_M^*$ become
finite rank operators on $\mathcal{H}$. We can relate them to odinary position
and momentum and their creation and annihilation operators
\begin{displaymath}
    A = \frac{1}{\sqrt{2}} (Q + i P),\quad
  A^* = \frac{1}{\sqrt{2}} (Q - i P)
\end{displaymath}
by
\begin{equation} \label{eq:22}
  A_M = \omega_M(H-\idty/2) A = A \omega_M(H \idty/2) ,\quad A_M^* = A^*
  \omega_M(H-\idty/2) = \omega_M(H + \idty/2) A^*
\end{equation}
where $A^*A = H - \idty/2$ is the number operator and $\omega_M$ is the function
given by
\begin{displaymath}
  \theta_M(n) =
  \begin{cases}
    \sqrt{1- \frac{n}{M}} & \text{if $0 \leq n \leq M$}\\
    0 & \text{otherwise}.
  \end{cases}
\end{displaymath}
If we introduce in addition the projections $P_K$ onto the span of
$\{\ket{0},\dots,\ket{K}\}$ -- i.e. the image of the embedding introduced
above -- we can rewrite $A_M, A_M^*$ again as
\begin{displaymath}
  A_M = \omega_M(H-\idty/2) a P_{M+1} \quad A_M^* = \omega_M(H+\idty/2) A^* P_M,
\end{displaymath}
since all $\psi \in \mathcal{H}$ with $P_{M+1} \psi = 0$ (or $P_M \psi = 0$)
are anyway in the kernel of $\omega_M(H-\idty/2) A$  (or $\omega_M(H+\idty/2)
A^*$).

Now it is easy to see that $\omega_M$ satisfies
\begin{displaymath}
    \left| 1 -\omega_M(n) \right| \leq \sqrt{\frac{n}{M}}.
\end{displaymath}
which leads to
\begin{displaymath}
  \lim_{M \rightarrow \infty} \| (\omega_M(H-\idty/2) - \idty) H^{-1}\| =
  \lim_{M \rightarrow \infty} \sup_n \frac{1 - \omega_M(n)}{(n + 1/2)} = 0.
\end{displaymath}
Similarly we get for the $P_K$:
\begin{displaymath}
  \lim_{M \rightarrow \infty} \| (P_M - \idty) H^{-1} \| = 0,
\end{displaymath}
which shows that we can apply Prop. \ref{prop:8} to see that the finite rank
operators $A_M, A_M^*$ converge as distributions (i.e. weakly in $\sopd$) to
$A, A^*$ (regarded as elements of $\sopd$, too;
cf. Sect. \ref{sec:oper-as-distr}). The same argument can be applied to
monomials of $A_M$ and $A_M^*$, if we move all $\omega_M$ terms to the left
and all $P_K$ terms to the right (cf. the commutation relations in
(\ref{eq:22}); similar equations also holds for the$P_K$). Since we can
rewrite any polynomial in $P,Q$ as a polynomial in $A,A^*$ we finally get:

\begin{prop} \label{prop:5}
  For each polynomial $f$ the sequence $f(Q_M,P_M)$ converges in $\sopd$ to
  $f(Q,P)$, i.e.
  \begin{displaymath}
    \lim_{M\rightarrow\infty} \tr(f(Q_M,P_M) \myrho) = \tr(f(Q,P) \myrho) \quad \forall
    \myrho \in \sop.
  \end{displaymath}
\end{prop}

If $\myrho \in \sop$ is a density operator (i.e. positive and normalized) the
trace $\tr(f(Q,P)\myrho)$ describes the expectation value of the observable
$f(Q,P)$ (provided it is selfadjoint or has a selfadjoint extension). Hence,
one way to interpret the previous proposition is to regard the operators $Q_M,
P_M$ as finite dimensional approximations of canonical position and momentum. We
will come back to this point in the next subsection.




\subsection{Approximations of position}
\label{sec:appr-posit}

Let us concentrate again on the position operator $Q$. According to the proposition just
proven all the moments
\begin{displaymath}
  m_n = \int_{\mathbb{R}} q^n \tr(\myrho E(dq)) = \tr(\myrho Q^n), \quad n \in \mathbb{N}
\end{displaymath}
exists for any Schwartz operator and the sequences
\begin{displaymath}
    m_{n,M} = \tr(\myrho Q^n_M), \quad M \in \mathbb{N}
\end{displaymath}
converge for each $n$ to $m_n$. Hence, if the measure $\mu(\Delta) = \tr(\myrho
E(\Delta))$ is uniquely determined by the moments, the measures $\mu_M(\Delta)
= \tr(\myrho E_M(\Delta))$, where
\begin{equation} \label{eq:6}
  E_M : \mathfrak{B}(\mathbb{R}) \rightarrow \mathcal{B}(\mathcal{H})
\end{equation}
are the spectral measures of the finite dimensional approximations $Q_M$,
converge weakly to $\mu$ (this is called the method of moments; cf
\cite{moments}). If we could show
this for all $\myrho \in \sop$ we could show that the $E_M$ converge -- as
distribution valued measures -- weakly to $E$. The problem with this reasoning
is that we do not know whether the measure $\mu$ is always (i.e. for each
$\myrho \in \sop$) uniquely defined by the moments (most likely this is not the
case). Fortunately, there is an independent argument to prove the desired
result.

\begin{prop} \label{prop:6}
  Consider the spectral measures $E_M$ (\ref{eq:6}) and $E$ (\ref{eq:15}) of
  the operators $Q_M$ and $Q$ respectively.
  \begin{itemize}
  \item[(a)] \label{item:5}
    For each interval $\Delta = (a,b)$ the operators $E_M(\Delta)$ converge
    strongly to $E(\Delta)$, and
    $$\lim_{M\rightarrow\infty}\tr(E_M(\Delta) \myrho) = \tr(E(\Delta) \myrho)\quad  \text{ for each }\myrho \in
    \sop.$$
  \item[(b)] \label{item:7}
    As a distribution valued measure the $E_M$ converge weakly to $E$, i.e.
  \begin{equation}
    \lim_{M\rightarrow\infty} \int_{\mathbb{R}} f(q) \tr(\myrho E_M(dq)) =
    \int_{\mathbb{R}} f(q) \tr(\myrho E(dq))
  \end{equation}
  holds for each $\myrho \in \sop$ and each continuous
  function $f: \mathbb{R} \rightarrow \mathbb{C}$ vanishing at infinity.
    \end{itemize}
\end{prop}

\begin{proof}
  From Eq. (\ref{eq:22}) it follows
  immediately that for each Hermite function $\psi_n$, $n \in \mathbb{N}$ the
  sequences $A_M \psi_n$ and $A_M^* \psi_n$, $M \in \mathbb{N}$ converge for $M
  \rightarrow \infty$ and fixed $n$ to $A\psi_n$ and $A^*\psi_n$. Furthermore
  it is well known that the operator $Q$ is self-adjoint and admits the
  space $F$ of finite linear combinations of Hermite functions as a core (cf. Example 2 in Sect. X.6 of \cite{RSII}). Since
  the operators $Q_M$ are bounded they are self-adjoint, too (they are
  obviously symmetric) and $F$ is again a core. Therefore, we can apply
  \cite[VIII.25]{RSI} to conclude that the $Q_M$ converge to $Q$ in the strong
  resolvent sense. Together with \cite[VIII.24]{RSI} this shows the first part of (a). Since the strong and
  $\sigma-$strong topology are identical on the unit ball in
  $\mathcal{B}(\mathcal{H})$ we get $\sigma$-strong convergence and since
  $\mathcal{B}(\mathcal{H}) \ni B \mapsto \tr(\myrho B) \in \mathbb{C}$ is
  $\sigma$-weakly continuous for all trace-class operators $\myrho$ we get
  $\lim_{M \rightarrow \infty} \tr(\myrho E_M(\Delta)) = \tr(\myrho E(\Delta))$. Hence we have proved (a).

  To prove (b) note first that it would follow automatically if (a) would hold for all Borel sets. However, since we have shown it only for intervals, (b) requires an additional
  argument. To this end consider first for $\epsilon>0$ an interval $I =
  [a,b]$ such that $\sup_{q \not\in I} |f(q)| < \epsilon/2$. This is always
  possible, since $f$ is (by assumption) vanishing at infinity. Since $\tr(\myrho E(\Delta))$ and $\tr(\myrho
  E_M(\Delta))$, are probabilitiy measures we have $\tr(\myrho E(\mathbb{R} \setminus I)) \leq 1$ and $\tr(\myrho
  E_M(\mathbb{R} \setminus I)) \leq 1$. Hence
  \begin{displaymath}
    \int_{\mathbb{R} \setminus I} |f(q)| \tr(\myrho E(dq)) < \frac{\epsilon}{2} \quad \int_{\mathbb{R} \setminus I}
    |f(q)| \tr(\myrho E_M(dq)) < \frac{\epsilon}{2},
  \end{displaymath}
  this leads to
  \begin{align*}
    \Biggl| \int_{\mathbb{R}} f(q) \tr(\myrho E(dq))  &- \int_{\mathbb{R}} f(q) \tr(\myrho E_M(dq)) \Biggr| \\
    & \leq \epsilon + \Biggl| \int_a^b f(q) \tr(\myrho E(dq)) - \int_a^b f(q) \tr(\myrho E_M(dq)) \Biggr|
  \end{align*}
  which shows that we can restrict our analysis to integrals over the interval $I$.

  To estimate the latter let us choose for each $n \in \mathbb{N}$ a partition $P_n$ of $I$ into $n$ subintervals
  of equal length. Without loss of generality we will assume now that $f$ is
  real valued (otherwise treat real and imaginary part separately). The $P_n$
  give then rise to a sequence of step functions $f_n: [a,b] \rightarrow
  \mathbb{R}$ which are defined by   $f_n(q) = \inf_{y \in J}f(y)$ for $q \in J$
  and $J \in P_n$. It follows immediately that $\lim_{n \rightarrow \infty}
  f_n(q) = f(q)$ and $f_n(q) \leq f(q)$. Dominated convergence therefore implies that there is an
  $n_{\epsilon,1}$ with
  \begin{displaymath}
    \left| \int_a^b f(q) \tr(\myrho E(dq)) - \int_a^b f_n(q) \tr(\myrho E(dq)) \right| < \frac{\epsilon}{2}
    \quad \forall n > n_{\epsilon,1}.
  \end{displaymath}
  Furthermore, since $f$ is continuous and $I$ compact, we can find another
  $n_{\epsilon,2} \in \mathbb{N}$ such that
  \begin{equation}\label{eq:16}
    \sup_{J \in P_n} (\sup_{y \in J} f(y) - \inf_{y \in J} f(y)) <
    \frac{\epsilon}{2 (\tr(\myrho E(I)) +1)} \quad \forall n > n_{\epsilon,2}.
  \end{equation}
  For the rest of the proof we choose one fixed $n > \max(n_{\epsilon,1},n_{\epsilon,2})$.

  Now let us come back to (b). There is an $M_\epsilon \in \mathbb{N}$ such that $M > M_\epsilon$
  implies
  \begin{equation} \label{eq:17}
    \max_{J \in P_n} |\tr(\myrho E(J)) - \tr(\myrho E_M(J))| < \frac{\epsilon}{2n},
  \end{equation}
  and therefore
  |$\tr(\myrho E(I)) - \tr(\myrho E_M(I))| < \epsilon/2$. If we require in addition $\epsilon < 2$ we see that in
  particular $\tr(\myrho E_M(I)) <  \tr(\myrho E(I)) + 1$ holds for all $M > M_\epsilon$. Hence with (\ref{eq:16})
  we get
  \begin{displaymath}
    \left| \int_a^b f_n(q) \tr(\myrho E_M(dq)) - \int_a^bf(x) \tr(\myrho E_M(dq)) \right| < \frac{\epsilon}{2}.
  \end{displaymath}
  Furthermore we get
  \begin{displaymath}
    \left| \int_a^b f_n(q) \tr(\myrho E_M(dq)) - \int_a^b f_n(q) \tr(\myrho E(dq)) \right| < \frac{\epsilon}{2}.
  \end{displaymath}
  from the bound (\ref{eq:17}).

  Now we are abe to pick up all the building blocks. For $M > M_\epsilon$ we have
  \begin{align*}
    \Biggl| \int_a^b f(q) \tr(\myrho E(dq))&- \int_a^b f(q) \tr(\myrho E_M(dq)) \Biggr| \leq \\
    & \leq \Biggl| \int_a^b f(q) \tr(\myrho E(dq)) - \int_a^b f_n(q) \tr(\myrho E(dq)) \Biggr| \\
    & \phantom{\Biggl| \int_a^b} + \Biggl| \int_a^b
    f_n(x) \tr(\myrho E_M(dq)) - \int_a^b f(q) \tr(\myrho E_M(dq)) \Biggr| \\
    &< \epsilon,
  \end{align*}
  what was to show.
\end{proof}

Weak convergence usually does not imply pointwise convergence. We can still find a sequence of (rescaled!)
projections $E_{M_k}(\Delta_{M_k})$, $k \in \mathbb{N}$ which converges in $\sopd$ to $\varepsilon_q$ for some
$q$. More precisely the following Corollary holds:

\begin{cor}
  For each $q \in \mathbb{R}$ there are sequences $(M_k)_k$ and $(I_k)_k$ of positive integers $M_k$ and
  intervals $I_k = [a_k,b_k]$ such
  that
  \begin{itemize}
  \item[(i)]
    For all $k$ we have $a_k < a_{k+1} < b_{k+1} < b_k$.
  \item[(ii)]
    $\lim_{k \rightarrow \infty} a_k = q = \lim_{k\rightarrow\infty} b_k$
  \item[(iii)]
    The sequence of operators $(b_k - a_k)^{-1}(E_{M_k}(I_k))$ converges in $\sopd$ to $\varepsilon_q$.
  \end{itemize}
\end{cor}

\begin{proof}
  We have to construct the sequence $E_{M_k}(I_k)$ such that for each positive $\myrho \in \sop$ and each
  $\epsilon > 0$ there is a $K_{\myrho,\epsilon}$ with
  \begin{equation} \label{eq:18}
    \left| \frac{\tr(E_{M_k}(I_k) \myrho)}{b_k - a_k} - \varepsilon_q(\myrho) \right| < \epsilon \quad \forall k >
    K_{\myrho,\epsilon}
  \end{equation}
  holds. Now consider the kernel $\kernel{\myrho}$ of $\myrho$. According to Cor. \ref{cor:1} we have
  \begin{displaymath}
    \tr(E(I) \myrho) = \int_a^b \kernel{\myrho}(q,q) dq = \bigl(\Re \kernel{\myrho}(\xi_r,\xi_r) + i \Im \kernel{\myrho}(\xi_i,\xi_i)\bigr) (b-a)
  \end{displaymath}
  for $\xi_r, \xi_i \in [a,b]$. Hence, for any pair of sequences $(a_n)_n$ and $(b_n)_n$ converging to $q$
  such that $(a_n)_n$ is strictly increasing and $(b_n)_n$ is strictly decreasing we have
  \begin{displaymath}
    \lim_{n \rightarrow \infty} \frac{\tr(E(I_n) \myrho)}{b_n - a_n} = \kernel{\myrho}(q,q) = \varepsilon_q(\myrho) \quad
    I_n = [a_n,b_n].
  \end{displaymath}
  Now choose for $k \in \mathbb{N}$ the index $n_k \in \mathbb{N}$ such that
  \begin{displaymath}
    \left| \frac{\tr(E(I_{n_k}) \myrho)}{b_{n_k} - a_{n_k}} - \varepsilon_q(\myrho) \right| < \frac{1}{2k},
  \end{displaymath}
  and for this fixed $n_k$ choose $M_k \in \mathbb{N}$ with (cf. Prop. \ref{prop:6}):
  \begin{displaymath}
    \left| \tr(E(I_{n_k}) \myrho) - \tr(E_{M_k}(I_{n_k})\right| < \frac{b_{n_k} - a_{n_k}}{2k}.
  \end{displaymath}
  Together we get
  \begin{align*}
    \Biggl| \frac{\tr(E_{M_k}(I_{n_k}) \myrho)}{b_{n_k} - a_{n_k}} - \varepsilon_q(\myrho) \Biggr| \leq \Biggl|
      \frac{\tr(E_{M_k}(I_{n_k}) \myrho)}{b_{n_k} - a_{n_k}} &- \frac{\tr(E(I_{n_k}) \myrho)}{b_{n_k} - a_{n_k}}
    \Biggr| \\
    &+ \Biggl| \frac{\tr(E(I_{n_k}) \myrho)}{b_{n_k} - a_{n_k}} - \varepsilon_q(\myrho) \Biggr| < \frac{1}{n},
  \end{align*}
  which implies (\ref{eq:18}) if we choose $I_k = I_{n_k}$.
\end{proof}

In other words: averaged densities of the measures $E_{M_k}$ converge weakly to the densities
$\epsilon_q$.


\section{Conclusions}

We have introduced Schwartz operators as a non-commutative analog of Schwartz
functions and have seen that they form a very well behaved class which can be
applied in many contexts of quantum mechanics. For  many unbounded
observables they allow in particular an easy discussion of  expectation values
and related concepts which now can basically be treated in the same way as we would do for bounded observables.

In contrast to other well behaved sets of states like Gaussian
states, Schwartz operators are dense in the set of trace class operators and
can therefore approximate a general quantum state with arbitrary
precision. Therefore one main message of our paper is: If Gaussian states are
too special for your purpose, try Schwartz operators as your next best choice.

If less regular objects have to be discussed Schwartz operators can also be of use; the associated dual (containing more general objects than operators) can be regarded as a
non-commutative version of ordinary tempered distributions. A large family
of quadratic forms is covered, and therefore constructions like products with
bounded and unbounded operators, and harmonic analysis are made available to
otherwise very singular objects. We have formulated non-commutative analogs of some selected elements of the theory of tempered distributions, in particular multiplication by polynomially bounded operators, the distributional derivative, the regularity theorem, Fourier transform and convolutions. Naturally, there are many other topics from the theory of distributions having counterparts in our non-commutative setting; developing a more comprehensive theory is however clearly beyond the scope of a single paper.








\section*{Acknowledgements} J.K. acknowledges support from the European CHIST-ERA/BMBF project CQC, and the \\ EPSRC project EP/J009776/1.

\end{document}